\documentclass[stsy,singleanonrev]{informs4}
\usepackage{eqndefns-left} %
\RequirePackage{tgtermes}
\RequirePackage{newtxtext}
\RequirePackage{newtxmath}
\RequirePackage{bm}
\RequirePackage{endnotes}

\usepackage{xcolor}
\usepackage{hyperref}
\hypersetup{
     hypertexnames = false,
     pdftitle = {Review-Period Sensitivity in Multiclass Queue Scheduling},
     pdfauthor = {Mir Mikdad Talpur and Vahid Sarhangian},
     colorlinks   = true,
     linkcolor    = red,
     citecolor    = blue
}

\OneAndAHalfSpacedXII %

\usepackage{algorithm}
\usepackage{algpseudocode}
\usepackage{tikz}

\usepackage{graphicx}   %
\usepackage[normalem]{ulem}

\usepackage{amsmath,amssymb}

\usepackage{caption}
\usepackage[labelfont=sf]{subcaption}
\AtBeginEnvironment{align}{\vspace{-\parskip}}
\AtBeginEnvironment{align*}{\vspace{-\parskip}}

\AtBeginEnvironment{align}{%
  \setlength{\abovedisplayskip}{0pt}%
  \setlength{\abovedisplayshortskip}{0pt}%
  \setlength{\belowdisplayskip}{0pt}%
  \setlength{\belowdisplayshortskip}{0pt}%
}
\AtBeginEnvironment{align*}{%
  \setlength{\abovedisplayskip}{3pt}%
  \setlength{\abovedisplayshortskip}{3pt}%
  \setlength{\belowdisplayskip}{3pt}%
  \setlength{\belowdisplayshortskip}{3pt}%
}

\newcommand{\vecx}{\mathbf{x}}
\newcommand{\vecu}{\mathbf{u}}
\newcommand{\vecy}{\mathbf{y}}

\newcommand{\vecf}{\mathbf{f}}
\newcommand{\vece}{\mathbf{e}}

\newcommand{\vecgamma}{\boldsymbol{\gamma}}
\newcommand{\vecalpha}{\boldsymbol{\alpha}}
\newcommand{\veczero}{\mathbf{0}}
\newcommand{\vecz}{\mathbf{z}}
\newcommand{\interior}[1]{\operatorname{int}(#1)}
\newcommand{\R}{\mathbb{R}}

\newcommand{\lamda}{\lambda}
 
\newcommand{\tPi}{\tilde{\pi}}

\newcommand{\pd}[2]{\frac{\partial #1}{\partial #2}}
\newcommand{\pdd}[2]{\frac{\partial ^ 2 #1}{\partial #2 ^ 2}}

\newcommand{\mpd}[3]{\frac{\partial ^ 2 #1}{\partial #3 \partial #2}}

\usepackage{natbib}
 \bibpunct[, ]{(}{)}{,}{a}{}{,}%
 \def\bibfont{\small}%

\EquationsNumberedThrough    %

\TheoremsNumberedThrough     %
\ECRepeatTheorems  %

\MANUSCRIPTNO{}

\begin{document}

\RUNAUTHOR{}

\RUNTITLE{Review-Period Sensitivity in Multiclass Queue Scheduling}

\TITLE{Review-Period Sensitivity in Multiclass Queue Scheduling}

\ARTICLEAUTHORS{%

\AUTHOR{Mir Mikdad Talpur}
\AFF{European Molecular Biology Laboratory, Heidelberg, Germany, \EMAIL{mir.talpur@embl.de}}

\AUTHOR{Vahid Sarhangian}
\AFF{Department of Mechanical and Industrial Engineering, University of Toronto, \EMAIL{v.sarhangian@utoronto.ca}}

} %

\ABSTRACT{%
Optimal control of queueing networks is typically studied under continuous-time control. In many service settings, however, managers can adjust decisions only at discrete and potentially infrequent review epochs. We study a multiclass queue scheduling problem in which server assignments can be changed only at the beginning of discrete review periods, and examine how performance depends on the review-period length. We analyze a family of associated fluid control problems parameterized by the review-period length and characterize the first- and second-order sensitivity of the value function. For the two-class case, we derive explicit expressions for these derivatives and characterize their signs. We find that for short review periods, the value function need not be monotone. Once the review period is sufficiently large, the value function becomes monotone nondecreasing and may exhibit a convex or linear region before eventually becoming concave. These results provide insights on when the value of more frequent control, or \emph{flexibility} with respect to server assignments, is higher. We further numerically examine the robustness of these observations for the original stochastic scheduling problem and show that stochasticity smooths the nonmonotonicity at smaller scales, while the qualitative sensitivity patterns remain visible and become more pronounced as the scale of the system increases.

}%


\maketitle
\section{Introduction}\label{sec:Intro}
 There is a vast literature on queueing control with diverse applications to healthcare \citep{armony2015patient}, service \citep{aksin2007modern}, and telecommunication systems \citep{harchol2013performance}, among others. Queueing control problems are typically studied in continuous-time, either through examining the structure of the optimal policy for a uniformized Markov Decision Process (MDP) formulation (e.g., \citealt{stidham1985optimal,badian2022delay,hu2025optimal}) or through analysis of an associated fluid (e.g., \citealt{meyn1997stability,maglaras2000discrete, bauerle2000asymptotic}) or diffusion (e.g., \citealt{ata2005heavy, gurvich2009scheduling, harrison2013brownian}) control problem. Under Markovian assumptions, continuous-time control is typically equivalent to controlling the system upon arrival of new customers or service completions. In many practical settings, it may however not be feasible to continuously intervene with the system due to operational constraints. Instead,  operational control may only be possible at the beginning of discrete (possibly far apart) points in time. Restriction to discrete-time control can fundamentally change the structure of the optimal policy, but also raises the question of whether control is still effective. 

An important example and the focus of this paper is the multiclass scheduling problem where a single or multiple servers attend to multiple classes of customers. When holding costs are linear, the celebrated $c\mu$ rule is optimal for single-server systems, and asymptotically optimal in the (conventional) heavy-traffic regime for multiserver systems \citep{van1995dynamic}.  In ``distributed" queueing systems, where servers are shared among queues at different locations, continuous-time scheduling is not feasible. \cite{chan2021dynamic,chan2022utilizing} study dynamic assignment of nursing staff to different areas of an emergency department at the beginning of 8-12 hour shifts. Using a fluid model, they investigate the structure of the optimal policy for a two-class system and show that $c\mu$ policy is no longer optimal as it may lead to excessive idleness mid-shift. They propose asymptotically optimal policies as well as heuristics for the discrete-time stochastic problem and show that \emph{partial flexibility} with respect to reassignment at the beginning of shifts can provide substantial reduction in waiting times. A similar application, first introduced in \cite{martonosi2011dynamic}, arises in airports where the staff can be shared between different security checkpoints, but their assignment can be changed only at discrete points in time. More recently, \cite{wu2023server} propose dynamic programming (DP)-based heuristic for the routing-scheduling problem with time-varying demand, stochastic travel time, and queue-length constraints.

In this paper, we consider the multiclass scheduling problem in discrete-time and investigate the impact of the frequency of control. More specifically, we are interested in understanding the value of more frequent control as the time between decision epochs, or \emph{review periods}, increases from zero (continuous-time control) to infinity (static allocation). To this end, we consider a fluid approximation of the system. Fluid models have played a foundational role in the analysis of stochastic processing networks including the seminal work of \cite{dai1995positive}, which linked stability of the fluid model to positive Harris recurrence of the underlying queueing network; see also \cite{dai2020processing}. Several studies have focused on developing approximate policies for stochastic networks using a translation of the fluid control solution. These translations often take the form of discrete-review policies as a mechanism for implementing controls derived from the fluid approximations. Review periods are designed so that the resulting stochastic policy tracks a target continuous-time control and is asymptotically optimal; see, e.g., \cite{harrison1998heavy}, \cite{maglaras2000discrete}. In contrast, we treat the review period as a given operational constraint and study its impact on the value of control. 

We consider a fluid approximation formally justified under the conventional heavy-traffic regime, e.g., by uniformly scaling the arrival and service rates as well as the initial condition (see, e.g., \citealt{mandelbaum1995strong}). The approximation has been used in several studies in the literature to gain insights on the structure of ``good" scheduling policies (see, e.g.,  \citealt{chen1994fluid,hu2022optimal,zychlinski2023applications, chen2025optimal}). For our problem, the conventional heavy-traffic framework results in simpler dynamics compared to the many-server fluid approximation (e.g., \citealt{mandelbaum1998strong, whitt2002stochastic}) considered in \cite{chan2021dynamic}. 

We study a family of discrete-time, finite-horizon fluid control problems, where the cost and transition functions are parameterized by the length of the review period. We then investigate the first-order and second-order sensitivity of the value function with respect to this parameter. To this end, we leverage results from the literature on sensitivity analysis of nonlinear programs (see, e.g., \citealt{pacaud2025sensitivity}) and apply them to our sequential decision problem through a DP formulation. Additional difficulty arises in our problem because the cost and transition functions are not continuously differentiable everywhere, and the optimal policy could also be non-differentiable at points where the active constraints change. While other papers have considered sensitivity of DPs (e.g., \citealt{hopp1988sensitivity}) our work appears to be the first to examine the sensitivity with respect to the length of the review period.   Our main contributions and results can be summarized as follows.  
\begin{itemize}
    \item \textbf{Structure of the optimal policy}: We characterize the structure of the optimal policy for the discrete-time multiclass fluid control problem. Specifically, at each review epoch, it is optimal to assign a higher-priority class enough capacity to empty by the end of the period (or all remaining capacity if this is not possible) before allocating any capacity to lower-priority classes. Allocating more than the amount required for clearing at the end of the period results in idleness, introducing a trade-off between diminishing cost reductions in the higher-priority class and the additional holding costs of longer lower-priority queues.   %
    \item \textbf{Sensitivity analysis with respect to the review-period length}: We study a family of DPs parameterized by the length of the review period and characterize the first and second derivatives of the value function with respect to the review period. We do so by implicitly differentiating the Bellman equation of the DP formulation with respect to the length of the review period at regular points where the derivatives exist. 
    \item \textbf{Insights from the two-class problem and extensions}: For the two-class problem, we provide explicit expressions for the first- and second-order derivatives of the value function and characterize their signs. For short review periods, the value function need not be monotone because changing the review period changes both the frequency and timing of review epochs. Once the review period is sufficiently large, the value function is increasing, but its curvature depends on the operating regime. When the optimal policy continues to prioritize the high-priority class despite the resulting within-period idleness, the value function is convex if the lower-priority class eventually clears and concave if it remains positive through the horizon. When the allocation instead adjusts so that the high-priority class clears exactly at the end of the review period, thereby avoiding idleness, the value function is linear. Once both classes receive positive capacity and the high-priority class clears before the end of the period, the value function is concave, reflecting the ability of the optimal allocation to partially mitigate the idleness created by less frequent review. These regimes characterize the value of additional reassignment flexibility, i.e., whether modest reductions in the review period can generate substantial gains or whether considerably more frequent control is needed. We further discuss extensions of these results to the multiclass problem.
    \item \textbf{Numerical study}: We numerically examine the curvature of the value functions and identify parameter regimes where infrequent control can be valuable. We further examine the relevance of our results for the original stochastic problem by numerically solving a MDP formulation of the problem for a sequence of scaled systems. We find that the non-monotonicity is largely smoothed out for the stochastic system at smaller scales, but otherwise the structure of the value function is similar to that of the fluid problem. 
\end{itemize}

The rest of the paper is organized as follows. Section \ref{sec:prob} presents the problem formulation, and Section \ref{sec:fluid} characterizes the optimal policy and contrasts it with continuous-time control. Section \ref{sec:monotone} develops the first- and second-order sensitivity results, while Section \ref{section: two class fluid problem} and Section \ref{section: more than two classes} apply the results to the queueing scheduling problem. Section \ref{sec:numerics} presents the numerical study, and Section \ref{sec:conc} provides a discussion and directions for future work. All proofs are provided in the Appendix. 

\subsection{Notation}\label{subsec:notation}
We use bold notation for model vectors and write all such vectors as
columns. For a differentiable vector-valued function
$F:\mathbb{R}^{\ell}\to\mathbb{R}^{m}$, with
$F(z)=(F^1(z),\ldots,F^m(z))^\top$, let
$y=(y^1,\ldots,y^p)^\top$ denote a block of its arguments, with
$p\leq\ell$. The Jacobian of $F$ with respect to $y$, evaluated at
$z_0$, is
\[
D_yF(z_0)
:=
\left[
\frac{\partial F^a}{\partial y^b}(z_0)
\right]_{a=1,\ldots,m;\,b=1,\ldots,p}
\in\mathbb{R}^{m\times p}.
\]
We use the transposed-Jacobian convention
\[
\frac{\partial F}{\partial y}
:=
(D_yF)^\top
\in\mathbb{R}^{p\times m}.
\]
This convention includes scalar arguments and scalar-valued functions.
Thus, for a scalar-valued function $\phi$, the derivative
$\partial\phi/\partial y=\nabla_y\phi$ is a column gradient, whereas
$\partial F/\partial t\in\mathbb{R}^{1\times m}$ is a row vector
when $t$ is scalar. We use the same orientation for total derivatives
$dF/dt$ of vector-valued mappings. Under this convention, for a differentiable function
$V:\mathbb{R}^{m}\to\mathbb{R}$, the chain rule gives
\[
\left.\nabla_y\bigl[V(F(z))\bigr]\right|_{z=z_0}
=
\frac{\partial F}{\partial y}(z_0)\,
\nabla V(F(z_0)).
\]
Partial derivatives hold the other arguments fixed, whereas total
derivatives account for dependence through all arguments of the
specified mapping. We suppress the evaluation point of a derivative
whenever it is clear from context.
For a twice continuously differentiable scalar function \(\phi\), the Hessian with respect to \(y\) is
\[
    \nabla^2_{yy}\phi(z_0)
    :=
    \left[
        \frac{\partial^2\phi}{\partial y^a\,\partial y^b}(z_0)
    \right]_{a,b=1,\ldots,p}
    \in\R^{p\times p}.
\]
More generally, for argument blocks \(y\in\R^p\) and \(w\in\R^q\), we write
\[
    \nabla^2_{yw}\phi(z_0)
    :=
    \left[
        \frac{\partial^2\phi}{\partial y^a\,\partial w^b}(z_0)
    \right]_{a=1,\ldots,p;\,b=1,\ldots,q}.
\]
In the sensitivity formulas, we use the equivalent compact notation,
\[
    \pdd{\phi}{y}:=\nabla^2_{yy}\phi,
    \qquad
    \mpd{\phi}{y}{w}:=\nabla^2_{wy}\phi.
\]
For vector-valued functions, second derivatives are defined
componentwise. In particular, for a twice continuously differentiable
function $F:\mathbb{R}^{\ell}\to\mathbb{R}^{m}$ with $m>1$,
an argument block $y\in\mathbb{R}^{p}$, and a scalar argument $t$,
we use the matrix layout
\[
\frac{\partial^2 F}{\partial y\,\partial t}
=
\frac{\partial^2 F}{\partial t\,\partial y}
:=
\left[
\frac{\partial^2 F^a}{\partial y^b\,\partial t}
\right]_{b=1,\ldots,p;\,a=1,\ldots,m}
\in\mathbb{R}^{p\times m}.
\]
We write,
\[
\frac{\partial^2 F}{\partial t^2}
:=
\left(
\frac{\partial^2 F^1}{\partial t^2},
\ldots,
\frac{\partial^2 F^m}{\partial t^2}
\right)
\in\mathbb{R}^{1\times m}.
\]
Further, scalar differentiation of row-valued derivative expressions and Jacobian matrices is performed entrywise.

Let \(v:\R^{l+m+r}\to\R\) be convex. For \(v(\vecx,\vecy,\vecz)\), with \(\vecx\in\R^l\), \(\vecy\in\R^m\), and \(\vecz\in\R^r\), a vector \(\vecalpha_{\vecy}\in\R^m\) is a subgradient of \(v\) with respect to \(\vecy\) at \((\vecx,\vecy^*,\vecz)\) if
\begin{equation}\label{def: subgradient}
    v(\vecx,\vecy,\vecz)-v(\vecx,\vecy^*,\vecz)
    \ge
    \vecalpha_{\vecy}^{\intercal}(\vecy-\vecy^*)
    \qquad
    \text{for all }\vecy\in\R^m.
\end{equation}
The subdifferential of \(v\) with respect to \(\vecy\) at \((\vecx,\vecy^*,\vecz)\) is,
\begin{equation}\label{def: subdifferential}
    \partial_{\vecy}v(\vecx,\vecy^*,\vecz)
    :=
    \left\{
    \vecalpha_{\vecy}\in\R^m:
    v(\vecx,\vecy,\vecz)-v(\vecx,\vecy^*,\vecz)
    \ge
    \vecalpha_{\vecy}^{\intercal}(\vecy-\vecy^*)
    \text{ for all }\vecy\in\R^m
    \right\}.
\end{equation}

Finally, for a set \(A\subset\R^l\), the interior is,
$
    \interior{A}
    :=
    \{a\in A:\exists\epsilon>0\text{ such that }B_\epsilon(a)\subset A\},
$
where $
    B_\epsilon(a):=\{x\in\R^l:\|x-a\|\le\epsilon\}.
$
The boundary of \(A\) is,
\[
    \partial A
    :=
    \left\{
    a\in\R^l:
    B_\epsilon(a)\cap A\ne\emptyset
    \text{ and }
    B_\epsilon(a)\cap A^c\ne\emptyset
    \text{ for all }\epsilon>0
    \right\}.
\]

\section{Problem Description}\label{sec:prob}
We consider a multiclass, multiserver queueing system. There are $K$ customer classes and $N$ identical servers. Customers of class $k \in \{1,\ldots,K\}$ arrive according to independent Poisson processes with rate $\lambda^k$ and have service requirements that are exponentially distributed with rate $\mu^k$. At each decision epoch, the decision maker decides how many servers to assign to each class. Let $\delta>0$ denote the time between decision epochs and assume a finite horizon of length $T$. Hence, there are $n(\delta)=\lceil T/\delta \rceil$ decision epochs over the horizon indexed by $i\in\{1,\ldots,n(\delta)\}$. Once the servers are assigned to a class, they serve the customers according to a non-idling First-Come, First-Served (FCFS) policy. If the new number of servers in a period is smaller than the number of customers in service, the last customers to enter service are preempted and put back in the queue. 

Let $\textbf{X}(t):=(X^1(t),\ldots, X^K(t))^\intercal$ denote the system state at time $t\geq0$ where $X^k(t)$ is the number of class $k$ customers at time $t$. Customers of class $k$ incur a linear holding cost with rate $h^k$ while in system. We assume that $h^1\mu^1 > \ldots > h^K\mu^K$ to avoid ties in the $c\mu$ indices, which can lead to nonunique optimal allocations in our fluid analysis. This allows us to focus on structural insights without introducing additional tie-breaking rules. The objective is to find an admissible policy that maps the state at the beginning of each review period to a non-negative server assignment vector $\textbf{U}(t)=(U^1(t),\ldots,U^K(t))^\intercal\in\mathbb Z_+^K$ that satisfies $\sum_{k=1}^K U^k(t)=N$ for all $t\geq 0$ and minimizes the expected finite-horizon cost,
\begin{equation}\label{eq: original stochastic problem}
    \mathbb{E} \left[\sum_{k=1}^K\int_0^{T}  h^k X^k(s)ds | \textbf{X}(0)= \textbf{X}\right]. %
\end{equation}
Denote by $V_i(\textbf{X},\delta)$ the value function for period $i\in\{1,\ldots,n(\delta)\}$ starting with $\textbf{X}$ customers in the system. We are interested in examining how $V_1(\textbf{X}, \delta)$ changes as $\delta$ varies. Because the value function of the stochastic problem is not tractable, we instead focus on its fluid approximation. We revisit the stochastic system in Section \ref{sec:numerics} and numerically examine the validity of the observations based on the fluid model for the stochastic value function. 

\subsection{Associated Fluid Control Problems}
In this section, we present a family of associated fluid control problems, parametrized by the time between decision epochs. %

We consider a family of discrete-time, continuous-state control problems over a finite-horizon of length $T$. Denote by $\delta>0$ the length of a period and let $\vecx_i := (x_i^1, \ldots,  x_i^K)^\intercal$ denote the $K$-dimensional state variable at the beginning of period $i$. Let \(\mathbf{u}_i := (u_i^1,\ldots,u_i^K)^\top\) denote the action in period \(i\), where \(u_i^k\) is the fraction of aggregate service capacity allocated to class \(k\), with feasible action set
\begin{equation}\label{eq: feasible set U}
\mathcal{U}:=\left\{\mathbf{u}\in\mathbb{R}_+^K:\mathbf{u}^\top\mathbf{e}=1\right\},
\end{equation}
where \(\mathbf{e}\) denotes the \(K\)-dimensional vector of ones. We denote by $g: \mathbb{R}_+^{2K + 1} \to \mathbb{R}_+$ the single-period cost function which maps the initial state $\vecx_i$ and action $\vecu_i$ in period $i$ to a non-negative cost. Finally, we denote by $\vecf : \mathbb{R}^{2K + 1} \to \mathbb{R}^{K}$ the vector-valued transition function that maps the current state $\vecx_i$, action $\vecu_i$, and the length of the period $\delta$ to the state at the beginning of the next period.

\begin{remark}\label{model_remark}
   Under the conventional heavy-traffic framework, the dynamics of the $N$-server system is equivalent to the system with $N=1$, in which the single server operates $N$ times faster than the original servers \citep{chen1994fluid}. We thus assume hereafter that $N=1$. The resulting fluid model allows arbitrary fractions of the aggregate capacity to be allocated across classes. Since the original stochastic system permits only integer server assignments, the fluid control problem can be viewed as a continuous-allocation relaxation of the corresponding limiting problem as discussed further in Section \ref{sec: stochastic system experiments}.
\end{remark}

Let $\lambda^k$, $\mu^k$, and $h^k$ respectively denote the arrival rate, service rate, and holding cost rate for class $k\in \{1 , \hdots, K \}$ customers and assume $h^1\mu^1 > \ldots > h^K\mu^K$. Let $x_i^k (s)$ denote the amount of class $k$ fluid $s \in  [0,\delta)$ time units into period $i$. Then, $x_i^k (s)$ evolves according to,
\begin{equation}
\frac{d^+x_i^k(s)}{ds}
=
\begin{cases}
\lambda^k-\mu^k u_i^k,
& x_i^k(s)>0,\\[1mm]
(\lambda^k-\mu^k u_i^k)^+,
& x_i^k(s)=0.
\end{cases}
\end{equation} 
where $\frac{d^+}{ds}$ denotes right-derivative; see, e.g., \cite{meyn2008control}. It follows that the $k$th element of the transition function \textbf{f} satisfies, 
\begin{equation}\label{eq:f_component_+} f^k(x_i^k, u_i^k , \delta) = \left( x_i^k + \delta ( \lambda^k - u_i^k \mu^k)\right)^+.       
\end{equation}
For $x_i^k>0$, denote the \emph{clearing time} of class \(k\) under the period-\(i\) allocation by,
$
\sigma_i^k(x_i^k,u_i^k) := \inf\{s\geq 0: x_i^k(s)=0\}.
$
We have,
\begin{equation}\label{eq:sigma_cleartime}
\sigma_i^k(x_i^k, u_i^k) = \begin{cases}
   \frac{x_i^k}{\mu^k u_i^k - \lambda^k}, & \text{ if } \lambda^k -  u_i^k \mu^k < 0, \\
   \infty, & \text{ otherwise}.
\end{cases} 
\end{equation}
Class \(k\) clears during period \(i\) if and only if \(\sigma_i^k(x_i^k,u_i^k)\leq \delta\).
We can therefore rewrite the transition function for class $k$ given in \eqref{eq:f_component_+} as,
\begin{equation}\label{eq:f_component_interval}
    f^k(x_i^k , u_i^k , \delta) = \begin{cases}
        x_i^k + \delta (\lambda^k - u_i^k \mu^k), & \delta < \sigma_i^k, \\
        0, & \text{else}.
    \end{cases}
\end{equation}
For a period of length $\delta>0$, the minimum capacity required
for class $k$ to be empty at the end of the period is,
\begin{equation}\label{eq: minimal clear capacity}
    \bar{u}_i^k
    :=
    \frac{x_i^k}{\delta\mu^k}
    +
    \frac{\lambda^k}{\mu^k}.
\end{equation}
The vector-valued transition function $\vecf$ is then given by,
\begin{equation}\label{eq:f_vec_sched}
\vecf (\vecx_i, \vecu_i, \delta) = \left( f^1 (x_i^1, u_i^1, \delta) ,\  f^2 (x_i^2, u_i^2, \delta) ,\ldots , f^K (x_i^K, u_i^K, \delta) \right)^\intercal,
\end{equation} %
and the single-period cost function is given by,
\begin{equation}\label{eq:g_sched}
g(\vecx_i, \vecu_i, \delta ) = \sum_{k = 1}^{K} \int_0^\delta h^k \left( x_i^k + s (\lambda^k - u_i^k\mu^k ) \right)^+ ds  = \sum_{k = 1}^K \int_0^{\delta \wedge \sigma_i^k} h^k \left( x_i^k + s (\lambda^k - u_i^k \mu^k)\right) ds .
\end{equation}

The following proposition establishes the convexity and differentiability of the transition and single-period cost functions. 

\begin{proposition}\label{prop: convexity of g and f}
The following properties hold:
\begin{enumerate}
    \item[(i)] The vector-valued mapping $\vecf(\vecx_i, \vecu_i, \delta)$
    is componentwise convex in $(\vecx_i,\vecu_i)$. That is, for each $k \in \{1,\dots,K\}$, the scalar function $f^k(x_i^k, u_i^k, \delta)$ is jointly convex in $(x_i^k, u_i^k)$.
    
    \item[(ii)] For fixed \(\delta>0\), the function \(g(\mathbf x_i,\mathbf u_i,\delta)\) is jointly convex in \((\mathbf x_i,\mathbf u_i)\) and continuously differentiable at all points outside \(\left\{
(\mathbf x_i,\mathbf u_i):
x_i^k=0,\;
u_i^k=\lambda^k/\mu^k
\text{ for some } k \in \{1,\dots,K\}
\right\}.\) For fixed \((\mathbf x_i,\mathbf u_i)\), \(g\) is continuously differentiable in \(\delta\).
    
    \item[(iii)] For each $k \in \{1,\dots,K\}$, the functions $f^k(x_i^k,u_i^k,\delta)$ and $g(\vecx_i,\vecu_i,\delta)$ are twice continuously differentiable in $(\vecx_i,\vecu_i,\delta)$ on any open set that does not intersect
         $\bigl\{(\vecx_i,\vecu_i,\delta) : x_i^k + \delta(\lambda^k - \mu^k u_i^k) = 0
        \text{ for some } k \in \{1,\dots,K\} \bigr\}.$
\end{enumerate}
\end{proposition}
For each $k$, the transition function $f^k(x_i^k,u_i^k,\delta)$ is not differentiable at points where $x_i^k+\delta(\lambda^k-\mu^k u_i^k)=0$, corresponding to both points where a positive queue clears exactly at the end of a period and \emph{maintained-empty} points satisfying $x^k_i=0$ and $u^k_i=\lambda^k/\mu^k$ for some $k \in \{1,\dots,K\}$. When \(x_i^k>0\), the single-period cost function \(g\) remains continuously differentiable when the queue clears exactly at the end of the period, although twice continuous differentiability fails at this boundary.

\textbf{DP Formulation}. Denote by $v_i(\vecx,\delta)$ the optimal cost-to-go for period $i$ starting in state $\vecx$ and with the length of periods fixed at $\delta$.
Define
$n(\delta) := \left\lceil T/\delta\right\rceil$, and 
let $r(\delta) := T-(n(\delta)-1)\delta \in (0,\delta]$
denote the length of the last (possibly partial) period. For
\(i \in\{1,\ldots,n(\delta)\}\), let
\[
\Delta_i(\delta):=
\begin{cases}
\delta, & i<n(\delta),\\
r(\delta), & i=n(\delta),
\end{cases}
\]
denote the length of period \(i\). The last-period value function is then,
\begin{equation}\label{eq:terminal_value}
v_n(\vecx_n,\delta)=\min_{\vecu_n\in\mathcal U} g(\vecx_n,\vecu_n,r(\delta)).
\end{equation}
(Note that we have suppressed the dependence of $n(\delta)$ on $\delta$ for brevity.) For $i \in \{1,2,\dots,n(\delta)-1\}$ the value functions satisfy the Bellman recursion
\begin{equation}\label{eq:fluid_opt}
 v_i(\vecx_i, \delta) = \min_{\vecu_i \in \mathcal{U}}
 \left[g (\vecx_i , \vecu_i, \delta)+ v_{i+1} \bigl(\vecf (\vecx_i, \vecu_i, \delta), \delta \bigr)\right].
\end{equation}
We further define, $q_n(\vecx_n,\vecu_n,\delta):= g(\vecx_n,\vecu_n,r(\delta))$ and
\begin{equation}\label{eq: q_t}
    q_i(\vecx_i , \vecu_i, \delta) := g (\vecx_i , \vecu_i, \delta)+ v_{i+1} (\vecf (\vecx_i, \vecu_i, \delta), \delta ),
\end{equation}
for $i\in \{1,\ldots,n(\delta)-1\}$.
The following proposition establishes the convexity of $q_i$ and the value function $v_i$ for fixed $\delta>0$. 
\begin{proposition}\label{prop: convex and differentiable structure}
   Fix $\delta>0$. The following statements hold: (i)  $q_i(\vecx_i, \vecu_i, \delta)$ is convex in $(\vecx_i, \vecu_i)$; (ii) $v_i(\vecx_i,\delta)$ is convex in $\vecx_i$; and (iii) \(v_i(\vecx_i,\delta)\) is componentwise nondecreasing in \(\vecx_i\), i.e., if \(\vecx_i\leq\vecy_i\) componentwise, then \(v_i(\vecx_i,\delta)\leq v_i(\vecy_i,\delta)\).
\end{proposition}

In closing, we discuss two extreme cases for $\delta$. Let $v_1 (\vecx_1, 0)$ denote the value function for the continuous-time control problem. In this case, the greedy $c\mu$ policy is optimal (see, e.g., \citealt{meyn2008control}).  Under the $c\mu$ policy, capacity is allocated in decreasing order of the $h^k\mu^k$ index. If class $k$ has positive fluid, it receives all remaining capacity until it is emptied (if possible). Once a class is emptied, it is assigned
exactly the minimal capacity $\lambda^k/\mu^k$ required to maintain zero fluid.
If the system is overloaded so that
$
\sum_{l=1}^{K} (\lambda^l/\mu^l) > 1,
$
then there exists a $k^\ast$ such that classes $1,\dots,k^\ast$ receive
positive capacity while class $k^\ast+1$ and all lower-priority classes receive zero
capacity.

At the other extreme, when $\delta=T$, there is a single decision period over the entire horizon,
and no intermediate control actions are possible.
Denote by $v_1(\vecx,T)$ the corresponding value function.
In this case, the optimal allocation is obtained by
minimizing the single-period (convex) cost function \eqref{eq:g_sched} over the feasible set $\mathcal{U}$
with period length $T$.

\section{Structure of the Optimal Fluid Policy}\label{sec:fluid}
In this section, we provide a characterization of an optimal policy for the $K-$class fluid scheduling problem for a fixed $\delta>0$.

\begin{theorem}\label{thm:K_class optimal policy}
    Assume $h^1\mu^1 > h^2 \mu^2 > \dots > h^K \mu^K$. For any period $i$, every optimal policy for the $K-$class fluid scheduling problem satisfies,
    \begin{equation}\label{eq: k class op inequality}
     u_i^k  \geq \min\left( 1 - \sum_{l = 1}^{k-1}u_i^l, \frac{x_i^k}{\Delta_i(\delta) \mu^k} + \frac{\lamda^k}{\mu^k} \right) , \quad  1\leq k \leq K - 1. 
    \end{equation} 
\end{theorem}
Theorem \ref{thm:K_class optimal policy} provides a lower-bound for the optimal allocation to the first $K-1$ classes. The lower-bound for class $k$ in period $i$ is the minimum capacity required to empty the class in that period, and the remaining capacity allocated to classes with a higher $c\mu$ index. The result is a generalization of Theorem 3 of \cite{chan2021dynamic} for multiple (rather than two) classes and under a conventional heavy-traffic regime. Our proof approach is also different from that in \cite{chan2021dynamic}, which relied on implicit differentiation of the optimality equations. Instead, we show that the multi-period cost of any policy that does not satisfy the lower-bounds can be improved, and hence it cannot be optimal. 

Theorem \ref{thm:K_class optimal policy} implies that a greedy policy is optimal as long as it does not result in any idleness during the period. That is, if allocating more capacity to class $k$ does not incur idleness, it is always optimal to do so before allocating more capacity to classes with lower $c\mu$ indexes. This is easily observed to hold for the single-period problem, but the result also indicates that it holds in general for the multi-period problem. Once idleness is incurred in class $k$, any additional capacity allocation requires trading-off further diminishing cost reduction in that class versus lower-priority classes.  

The strict ordering of the $c\mu$ indexes provides a unique set of lower bounds for each class. Intuitively, if the $c\mu$ indexes were equal for some classes, e.g., $h^1 \mu^{1} = h^{2}\mu^{2}$, this lower bound would only need to hold for one of the classes.

\begin{figure}
    \FIGURE{
    \includegraphics[scale=0.65]{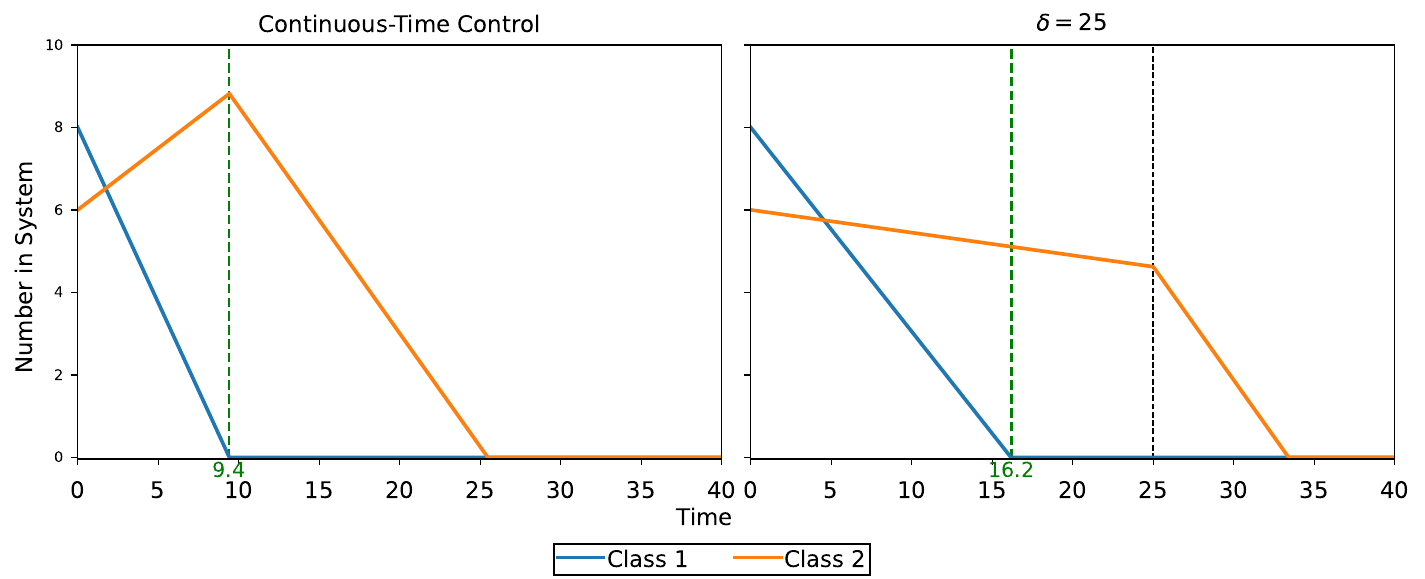}}
    {Optimal trajectories under continuous-time control versus discrete time control with $\delta = 25$. \label{fig:optimal policy discussion example}}
    { System parameters: $(\lambda^1, \lambda^2) = (0.35, 0.3)$, $(\mu^1, \mu^2) = (1.2, 1.2)$,  $(h^1 , h^2) = (4, 1)$, $(x_1^1, x_1^2) = (8, 6)$, $T = 60$.}
\end{figure}

\textbf{Numerical Illustration}. To better illustrate the result, we next provide a two-class example and compare the optimal policy with $\delta>0$ with that of the continuous-time control. Figure \ref{fig:optimal policy discussion example} depicts the optimal trajectories in a two-class system with $(\lambda^1, \lambda^2) = (0.35, 0.3)$, $(\mu^1, \mu^2) = (1.2, 1.2)$,  $(h^1 , h^2) = (4, 1)$, $(x_1^1, x_1^2) = (8, 6)$, and $T = 60$ under continuous-time control (left) and discrete-time control with $\delta=25$ (right). Under continuous-time control, it is optimal to allocate all capacity to class 1 until it is emptied at time $t=9.4$, and then allocate $(\lambda_1/\mu_1)$ to keep it at zero. In the case of discrete-time control, allocating all capacity to class 1 in the first period would result in sub-optimal idleness. Instead, it is optimal to empty class 1 at time $t=16.2$ and allocate some capacity to class 2 starting in the first period.

\section{First- and Second-Order Sensitivity Analysis}\label{sec:monotone}
Our objective is to study how \(v_1(\vecx_1,\delta):= q_1(\vecx_1,\vecu_1^*(\delta),\delta)\) varies with the review-period length \(\delta\).  To this end, we first characterize the optimality conditions for a fixed review-period length $\delta$. We then define regular points where the derivatives of $v_1(\vecx_1,\delta)$ are well-defined.

For any period-\(i\) allocation \(\vecu_i\), the nonnegativity
constraint for class \(k\) is active if \(u_i^k=0\) and inactive if
\(u_i^k>0\). Define the corresponding active and inactive constraint
index sets by,
\begin{equation}
A(\vecu_i):=\{k:u_i^k=0\},
\qquad
B(\vecu_i):=\{k:u_i^k>0\}.
\end{equation}
Fix a reference review length $\delta_0>0$. Throughout the local
sensitivity analysis, $\delta$ denotes a review length varying in a
neighborhood of $\delta_0$. Fix a
selected optimal trajectory at $\delta_0$ and denote it by $
\{(\vecx_i^*,\vecu_i^*)\}_{i=1}^{n(\delta_0)}.
$
When the dependence on the review length is relevant, we explicitly write
$\vecx_i^*(\delta)$ and $\vecu_i^*(\delta)$. All starred quantities
without an explicit argument refer to the base trajectory evaluated at
$\delta=\delta_0$. We set \(A_i^*:=A(\vecu_i^*)\) and
\(B_i^*:=B(\vecu_i^*)\). Thus, \(A_i^*\) and \(B_i^*\) are,
respectively, the active and inactive nonnegativity-constraint sets at
the base optimizer in period \(i\). We also define,
$
R_i^{k,*}
:=
x_i^{k,*}+\Delta_i(\delta_0)(\lambda^k-\mu^k u_i^{k,*}).
$
Accordingly, the boundary \(x_i^{k,*}>0\) and \(R_i^{k,*}=0\) corresponds to a positive queue clearing exactly at the end of a period. 
\subsection{Optimality Conditions}
For multipliers
\(\vecgamma=(\gamma^1,\ldots,\gamma^K)^\intercal\) and \(\beta\), define the first-period Lagrangian at review length \(\delta\) by
\begin{align}\label{eq:Lagrangian}
    \mathcal L_1(\vecx_1,\vecu_1,\delta)
    &=
    q_1(\vecx_1,\vecu_1,\delta)
    -\vecu_1^\intercal\vecgamma
    -\beta\left(\vecu_1^\intercal\vece-1\right) \nonumber\\
    &=
    g(\vecx_1,\vecu_1,\delta)
    +v_2(\vecf(\vecx_1,\vecu_1,\delta),\delta)
    -\vecu_1^\intercal\vecgamma
    -\beta\left(\vecu_1^\intercal\vece-1\right).
\end{align}
At \(\delta=\delta_0\), the necessary and sufficient KKT conditions state that $\vecu_1^*=\left(u_1^{1,*},u_1^{2,*},\ldots,u_1^{K,*}\right)^\intercal$ is an optimal solution if there exist multipliers $(\vecgamma, \beta)$ such that,
\begin{align}
     \partial_{\vecu_1}q_1(\vecx_1, \vecu_1^*, \delta_0) - \beta \mathbf{e} - \vecgamma & \ni \veczero, \label{eq:KKT general Stationary} \\
    \vecu_1^* & \geq \veczero , \label{eq:KKT general Primal Feasibility}\\ 
    \vecu_1^{* \intercal}  \mathbf{e} - 1 & = 0, \label{eq:KKT general Equality Constraint}\\
  \vecgamma & \geq \veczero \label{eq:KKT general Dual Feasibility},\\
  \vecu_1^{* \intercal} \vecgamma & = 0 \label{eq:KKT general Complementary Slackness},
 \end{align}
 where $\partial_{\vecu_1}q_1 (\vecx_1, \vecu_1^*, \delta_0)$ is the subdifferential of $q_1$ evaluated at $(\vecx_1, \vecu_1^*, \delta_0)$; see \eqref{def: subdifferential}. Whenever $q_1$ is differentiable with respect to $\vecu_1$ at $(\vecx_1, \vecu_1^*, \delta_0)$, the subdifferential set contains only the derivative of $q_1$ with respect to $\vecu_1$ at that point, i.e., $\pd{}{\vecu_1} [ q_1 (\vecx_1, \vecu_1^*, \delta_0)]$, and our KKT conditions reduce to,
 \begin{align}
    \pd{}{\vecu_1}q_1(\vecx_1,\vecu_1^*,\delta_0)
  -\beta\mathbf e-\vecgamma
  =
  \pd{}{\vecu_1}\!\left[
  g(\vecx_1,\vecu_1^*,\delta_0)
  +v_2(\vecf(\vecx_1,\vecu_1^*,\delta_0),\delta_0)
  \right]
  -\beta\mathbf e-\vecgamma
  &=\mathbf0.
  \label{eq:KKT Stationary} \\
    \vecu_1^* & \geq \veczero , \label{eq:KKT Primal Feasibility}\\ 
    \vecu_1^{* \intercal}  \mathbf{e} - 1 & = 0, \label{eq:KKT Equality Constraint}\\
  \vecgamma & \geq \veczero \label{eq:KKT Dual Feasibility},\\
  \vecu_1^{* \intercal} \vecgamma & = 0 \label{eq:KKT Complementary Slackness},
 \end{align} see, e.g., \cite{boyd2004convex}.
Using the definition of \(q_1\), the stationarity condition \eqref{eq:KKT Stationary} can equivalently be written as,
\begin{align}
     \beta \mathbf{e}+ \vecgamma   
     & = \pd{g}{\vecu_1}+ \pd{\vecf}{\vecu_1} \pd{v_{2}}{\vecx_{2}}. \label{eq:KKT Stationarity Chain Rule Applied}
\end{align}
Here, \(\vecx_1\) and \(\delta_0\) are fixed when differentiating with respect to \(\vecu_1\).

\subsection{Regularity Conditions}
We first review the relevant conditions from parametric nonlinear programming; see \cite{fiaccoKortanek1983,Jittorntrum1984} and \cite{pacaud2025sensitivity}.

\begin{definition}[NLP regularity conditions]\label{def:nlp_regular_conditions}
Consider a local nonlinear program with decision variable \(\vecy\) and
parameter \(p\), objective \(\phi(\vecy,p)\), equality constraints
\(a^\ell(\vecy,p)=0\), and inequality constraints
\(b^j(\vecy,p)\le0\). Let \(w_\ell\) and \(z_j\ge0\) denote the
multipliers associated with \(a^\ell\) and \(b^j\), respectively, and
define the Lagrangian by,
\[
L(\vecy,p,\mathbf w,\vecz)
:=
\phi(\vecy,p)
+\sum_\ell w_\ell a^\ell(\vecy,p)
+\sum_j z_j b^j(\vecy,p).
\]
The active inequality set at \((\vecy,p)\) is
$
\mathcal I_p(\vecy):=\{j:b^j(\vecy,p)=0\}.
$
At a KKT point \((\vecy^*,p^*)\) with multipliers
\((\mathbf w^*,\vecz^*)\), an active inequality
\(j\in\mathcal I_{p^*}(\vecy^*)\) is strongly active if \(z_j^*>0\)
and weakly active if \(z_j^*=0\). Strict complementary slackness (SCS)
holds if there are no weakly active inequalities.
Linear-independence constraint qualification (LICQ) holds if the
gradients of the equality constraints and the active inequality
constraints are linearly independent. When SCS holds, the critical cone
is,
\[
C_{p^*}(\vecy^*)
:=
\left\{
\mathbf d:
\nabla_{\vecy}a^\ell(\vecy^*,p^*)^\intercal \mathbf d=0\ \forall \ell,\quad
\nabla_{\vecy}b^j(\vecy^*,p^*)^\intercal \mathbf d=0\
\forall j\in\mathcal I_{p^*}(\vecy^*)
\right\}.
\]
The second-order sufficient condition (SOSC) holds if the Hessian of the Lagrangian \(L\) is positive definite on the nonzero directions in this cone:
\[
\mathbf d^\intercal
\nabla_{\vecy\vecy}^2L(\vecy^*,p^*,\mathbf w^*,\vecz^*)
\mathbf d>0, 
\qquad
\mathbf d\in C_{p^*}(\vecy^*),\quad \mathbf d\ne\veczero.
\]
\end{definition}

\begin{proposition}[Parametric NLP sensitivity]
\label{prop:parametric_nlp_sensitivity}
Consider the parametric nonlinear program in Definition
\ref{def:nlp_regular_conditions}. Suppose that its objective and
constraints are \(C^2\) and that, at \(p=p^*\), the primal-dual KKT
solution \((\vecy^*,\mathbf w^*,\vecz^*)\) satisfies LICQ, SCS, and
SOSC. Then, near \(p^*\), the primal-dual KKT solution $
p\mapsto
\bigl(\vecy^*(p),\mathbf w^*(p),\vecz^*(p)\bigr)
$
is locally unique and \(C^1\), with
\((\vecy^*(p^*),\mathbf w^*(p^*),\vecz^*(p^*))
=(\vecy^*,\mathbf w^*,\vecz^*)\).
Moreover, the active inequality set is fixed, LICQ and SCS continue to
hold locally, and the local value function
\(\phi(\vecy^*(p),p)\) is \(C^2\).
\end{proposition}

Proposition \ref{prop:parametric_nlp_sensitivity} is the classical Fiacco sensitivity theorem. The differentiable primal-dual solution selection follows from the result stated in \citet[Theorem 3.2]{pacaud2025sensitivity}; the \(C^2\) local value-function conclusion follows from \citet[Theorem 3.4.1]{fiaccoKortanek1983}. %
The hypotheses of Proposition \ref{prop:parametric_nlp_sensitivity} need not hold everywhere in our
  scheduling problem. In particular, By Proposition
  \ref{prop: convexity of g and f}, the required \(C^2\) smoothness fails
  when a class empties exactly at the end of a period, or starts empty and receives the \emph{maintenance allocation} $\lambda^k/\mu^k$ to keep it empty. Other possible failures of SCS and SOSC are discussed next.

For the base optimal trajectory, define the set of \emph{maintained-empty} classes in period \(i\) by \(E_i:=\{k:x_i^{k,*}=R_i^{k,*}=0\}\). The corresponding reduced feasible set is,
\[
\mathcal U_i^{\mathrm{red}}
:=
\left\{
\vecu_i\in\mathcal U:
u_i^k=\frac{\lambda^k}{\mu^k}\ \text{for every }k\in E_i
\right\}.
\]
For \(\delta\) near \(\delta_0\), the \emph{reduced period-\(i\) problem} is the minimization of \(q_i(\vecx_i,\vecu_i,\delta)\) over \(\mathcal U_i^{\mathrm{red}}\), with \(x_i^k=0\) fixed for \(k\in E_i\). Thus, only the remaining state and allocation coordinates are treated as variables. Equivalently, the \emph{free} allocations satisfy \(\sum_{k\notin E_i}u_i^k=1-\sum_{k\in E_i}\lambda^k/\mu^k\). %

It is easy to verify that LICQ holds for our problem. From \eqref{eq:KKT general Equality Constraint}, it follows that the gradient of our equality constraint is $\mathbf{e}$. By Theorem \ref{thm:K_class optimal policy}, $u_i^{1, *} > 0, \forall i \in \{1, \dots, n\}$, thus the set of gradients of active inequality constraints is always some subset of $\{\mathbf{e}_2, \dots, \mathbf{e}_K\}$, and so the set of equality and active inequality constraints is always linearly independent. %

\textbf{Possible failures of SCS and SOSC}. At a period \(i\) and a generic review length \(\delta\), suppose the continuation value \(v_{i+1}\) is
  differentiable at the induced next state. Holding
  \(\vecx_i\) and \(\delta\) fixed, define the period-\(i\) dynamic
  marginal value of capacity in class \(k\) by,
  \begin{equation}\label{eq:dynamic_marginal_value}
  \mathcal M_i^k
  :=
  -\frac{\partial q_i}{\partial u_i^k}
  =
  \begin{cases}
  -\dfrac{\partial g}{\partial u_i^k}
  -\dfrac{\partial f^k}{\partial u_i^k}
   \dfrac{\partial v_{i+1}}{\partial x_{i+1}^k},
  & i<n(\delta),\\[6pt]
  -\dfrac{\partial g}{\partial u_i^k},
  & i=n(\delta).
  \end{cases}
  \end{equation}
  The derivatives are evaluated at the optimal point under
  consideration, and the last-period running cost uses the residual
  period length \(r(\delta)\).

The second equality in \eqref{eq:dynamic_marginal_value} follows
from the chain rule and the classwise separability of \(\vecf\). At the
base trajectory, stationarity and complementary slackness give
\(\mathcal M_i^{\kappa}=-\beta\) for \(\kappa\in B_i^*\setminus E_i\), whose
nonnegativity constraint is inactive, and
\(\mathcal M_i^k=-\beta-\gamma^k\) for
\(k\in A_i^*\setminus E_i\), whose constraint is active. Hence SCS for
the reduced problem holds exactly when
\(\mathcal M_i^\kappa>\mathcal M_i^k\) for every
\(\kappa\in B_i^*\setminus E_i\) and \(k\in A_i^*\setminus E_i\). Equality
for such a pair is a \emph{dynamic cutoff tie} and is precisely where SCS
fails. In Appendix \ref{ap:single-stage-example} we demonstrate this for the single-stage cost.

Under SCS, a critical direction for the reduced period-\(i\)
problem is a reallocation \(\mathbf d\) satisfying \(d^k=0\) for every
\(k\in A_i^*\setminus E_i\) and
\(\mathbf d^\intercal\mathbf e=0\). The Hessian in given by,
\(\nabla^2_{\vecu_i\vecu_i}q_i(\vecx_i^*,\vecu_i^*,\delta_0)\). Thus
SOSC requires
\(\mathbf d^\intercal\nabla^2_{\vecu_i\vecu_i}q_i
(\vecx_i^*,\vecu_i^*,\delta_0)\mathbf d>0\) for every nonzero critical direction. As shown in Proposition \ref{prop:regular_point_smoothness}, SOSC is automatic once SCS
holds and assuming the following convention for selecting an optimal allocation starting with zero fluid in all queues.

\textbf{A selection convention when starting with} \(\vecx_i=\veczero\). If \(\vecx_i=\veczero\) and
\(\sum_{k=1}^K\lambda^k/\mu^k<1\), any allocation satisfying
\(u_i^k\geq\lambda^k/\mu^k\) for every \(k\) keeps all queues empty and
has the same zero cost. We select the optimizer that assumes
classes \(k \in \{2,\ldots,K\}\) receive the maintenance allocation $\lambda^k/\mu^k$ and any excess
capacity goes to the highest-priority class:
\[
u_i^{k,*}=\frac{\lambda^k}{\mu^k},\quad k\in\{2,\ldots,K\},
\qquad
u_i^{1,*}=1-\sum_{k=2}^K\frac{\lambda^k}{\mu^k}.
\]
Henceforth, base optimal trajectories are selected according to this
convention.

\begin{definition}[Regular point]\label{def:regular_point}
At the fixed initial state \(\vecx_1\), the base review length \(\delta_0>0\)
is a regular point for the fixed base optimal trajectory if there exists an open interval
\(I\subset(0,\infty)\) containing \(\delta_0\) such that
\(n(\delta)=n(\delta_0)=:n\) for every \(\delta\in I\), and the following conditions hold for every period \(i\in\{1,\ldots,n\}\):

\begin{enumerate}
\item \textbf{Smooth-piece condition.}
For every \(k\notin E_i\), we have \(R_i^{k,*}\ne0\).

\item \textbf{Maintained-empty reduction.}
For each \(k\in E_i\),
the local period-\(i\) Bellman problem is represented by the reduced problem defined above, with \(x_i^k=0\) and \(u_i^k=\lambda^k/\mu^k\) fixed.

  \item \textbf{No cutoff tie.}
  At the base point \(\delta=\delta_0\), the reduced period-\(i\)
  Bellman minimization satisfies
  \(\mathcal M_i^r>\mathcal M_i^k\) for every
  \(r\in B_i^*\setminus E_i\) and
  \(k\in A_i^*\setminus E_i\).
\end{enumerate}
\end{definition}

The fixed-period and smooth-piece conditions exclude period-count and
endpoint-clearing boundaries. Condition 2 requires classes in $E_i$ to remain the same as \(\delta\) varies
near \(\delta_0\). Condition 3 imposes SCS through the strict marginal inequalities.  

\begin{proposition}[Smoothness and stability at regular points]
\label{prop:regular_point_smoothness}
If \(\delta_0\) is a regular point at \(\vecx_1\), then there is an open
interval \(J\subseteq I\) containing \(\delta_0\) such that:

\begin{enumerate}
\item[(i)]
For every period \(i\), in reduced coordinates, \(q_i\) is \(C^2\)
near \((\vecx_i^*,\vecu_i^*,\delta_0)\), the reduced period-\(i\)
optimizer is unique and depends \(C^1\) on the nonfixed components of
\(\vecx_i\) and on \(\delta\), its active
nonnegativity-constraint set is locally constant, and
\(v_i\) is \(C^2\) near
\((\vecx_i^*,\delta_0)\).

\item[(ii)]
The reduced Bellman problems admit
a unique optimal trajectory
\((\vecx_i^*(\delta),\vecu_i^*(\delta))_{i=1}^n\), which is \(C^1\) in
\(\delta\) and satisfies
\(
(\vecx_i^*(\delta_0),\vecu_i^*(\delta_0))
=(\vecx_i^*,\vecu_i^*)
\)
for every \(i\).
\end{enumerate}
\end{proposition}

\subsection{First-Order Sensitivity}
The following theorem provides a characterization of the first derivative of the value function with respect to $\delta$ at a regular point, which does not depend on the derivative of the optimal policy. 
\begin{theorem}\label{thm:delta_der_value_function}
    Suppose \(\delta_0\) is a regular point at \(\vecx_1\). Then \( v_1(\vecx_1,\delta)\) is differentiable at \(\delta_0\), and its derivative is given by,
\begin{equation}\label{eq:delta_der_value_function}
         \pd{}{\delta}v_1 = \pd{g}{\delta} + \pd{\vecf}{\delta} \pd{v_{2}}{\vecx_{2}} + \pd{v_{2}}{\delta},
\end{equation}
where all terms are evaluated at \(\delta=\delta_0\) along the base optimal trajectory. %
\end{theorem}

Theorem \ref{thm:delta_der_value_function} decomposes the marginal effect of changing the review period into three components. The first term captures the direct effect on the current-period cost, the second captures the effect of the resulting change in the next-period state, and the third captures the effect of the review-period length on the continuation value. Importantly, the first-order sensitivity does not depend explicitly on how the optimal allocation changes with \(\delta\). This decomposition allow us to establish monotonicity in the queueing regimes considered in Section \ref{section: two class fluid problem} and \ref{section: more than two classes}.

\subsection{Second-Order Sensitivity}
Next, we differentiate the first-order envelope expression to obtain the second derivative of \(v_1\). To distinguish the local optimizer mapping from its fixed base value, we write \(\vecu_1^*(\delta)\) whenever the optimizer is differentiated. After the main result, Lemmas \ref{lem:kkt_cancellation_regular_regime} and \ref{lem:kkt_differentiated_stationarity} record the KKT identities used to compute this sensitivity and to simplify the two-class expressions studied in Section \ref{section: two class fluid problem}.

\begin{theorem}\label{thm:delta_second_der_value_function}
    Suppose \(\delta_0\) is a regular point at \(\vecx_1\). Then \(v_1(\vecx_1,\delta)\) is twice differentiable at \(\delta_0\), holding \(\vecx_1\) fixed, and its second derivative is given by,
    \begin{equation}\label{eq:delta_second_der_value_function}
    \begin{aligned}
        \pdd{}{\delta}v_1 =
        &  \frac{d\vecu_1^*(\delta)}{d\delta}\mpd{g}{\delta}{\vecu_1} + \pdd{g}{\delta} + \left( \frac{d\vecu_1^*(\delta)}{d\delta}\mpd{\vecf}{\delta}{\vecu_1} +\pdd{\vecf}{\delta}\right)  \pd{v_{2}}{\vecx_{2}}  \\
        & + \pd{\vecf}{\delta}  \left( \left( \frac{d\vecu_1^*(\delta)}{d\delta} \pd{\vecf}{\vecu_1} + \pd{\vecf}{\delta}  \right) \pdd{v_{2}}{\vecx_{2}} + \mpd{v_{2}}{\vecx_{2}}{\delta} \right)^\intercal  \\
        & + \left( \frac{d\vecu_1^*(\delta)}{d\delta} \pd{\vecf}{\vecu_1} + \pd{\vecf}{\delta} \right)\mpd{v_{2}}{\delta}{\vecx_{2}} + \pdd{v_{2}}{\delta},\label{eq:sec_sens}
    \end{aligned}
\end{equation}
where all terms are evaluated at \(\delta=\delta_0\) along the base optimal trajectory.
\end{theorem}

In contrast to the first-order sensitivity, the second derivative depends on how the optimal allocation itself changes with \(\delta\). Thus, curvature reflects both the direct effect of changing the review period and the endogenous adjustment of the scheduling policy. In Section \ref{section: two class fluid problem} we use this result to characterize the sign of the second-order derivative in different clearing regimes, thereby identifying when the marginal value of more frequent control is increasing or decreasing.

Along the local regular regime supplied by Proposition \ref{prop:regular_point_smoothness}, let \(\beta(\delta)\) and \(\vecgamma(\delta)\) denote the corresponding KKT multipliers, with \(\beta=\beta(\delta_0)\) and \(\vecgamma=\vecgamma(\delta_0)\). The following two lemmas collect the KKT identities used to compute the optimizer sensitivity and to simplify the two-class results.

\begin{lemma}[KKT cancellation along a regular regime]\label{lem:kkt_cancellation_regular_regime}
Suppose \(\delta_0\) is a regular point at \(\vecx_1\). Interpret the KKT conditions for the reduced first-period problem when \(E_1\neq\emptyset\), and extend the resulting derivative to the full allocation vector by setting \(d u_1^{k,*}(\delta)/d\delta=0\) for \(k\in E_1\). Then, at \(\delta=\delta_0\) along the local regular regime,
\begin{equation}\label{eq: KKT Derivative of Equality Constraint}
     \frac{d\vecu_1^*(\delta)}{d\delta}\mathbf{e} = 0,
\end{equation}
\begin{equation}\label{eq:KKT Complentary Slack with u delta derivatives}
     \frac{d\vecu_1^*(\delta)}{d\delta}\vecgamma = 0,
\end{equation}
and,
\begin{equation}\label{eq: Vector FONC}
    \frac{d\vecu_1^*(\delta)}{d\delta} \left( \pd{g}{\vecu_1}  + \pd{\vecf}{\vecu_1} \pd{v_{2}}{\vecx_{2}}  \right) = 0.
\end{equation}
\end{lemma}

Lemma \ref{lem:kkt_cancellation_regular_regime} shows that changes in the optimal allocation cancel from the first-order expression. For the second-order sensitivity, however, we need to characterize how the optimal allocation itself changes with \(\delta\), which we obtain by differentiating the KKT stationarity condition.

\begin{lemma}[Differentiated KKT stationarity]\label{lem:kkt_differentiated_stationarity}
Suppose \(\delta_0\) is a regular point at \(\vecx_1\). On the free variables of the reduced first-period problem, and with zero derivatives assigned to maintained-empty components, at \(\delta=\delta_0\) along the local regular regime,
\begin{equation}
\begin{aligned}
\frac{d\beta(\delta)}{d\delta}\mathbf{e}
+
\left(
    \frac{d\boldsymbol{\gamma}(\delta)}{d\delta}
\right)^\top
&=
\frac{\partial^2 g}{\partial \mathbf{u}_1^2}
\left(
    \frac{d\mathbf{u}_1^*(\delta)}{d\delta}
\right)^\top
+
\frac{\partial^2 g}{\partial \mathbf{u}_1\,\partial\delta}
+
\frac{\partial^2 \mathbf{f}}
     {\partial\delta\,\partial\mathbf{u}_1}
\frac{\partial v_2}{\partial\mathbf{x}_2}
\\
&\quad+
\frac{\partial\mathbf{f}}{\partial\mathbf{u}_1}
\left[
    \frac{\partial^2 v_2}{\partial\mathbf{x}_2^2}
    \left(
        \frac{d\mathbf{u}_1^*(\delta)}{d\delta}
        \frac{\partial\mathbf{f}}{\partial\mathbf{u}_1}
        +
        \frac{\partial\mathbf{f}}{\partial\delta}
    \right)^\top
    +
    \frac{\partial^2 v_2}
         {\partial\mathbf{x}_2\,\partial\delta}
\right].\label{eq: der_of_fonc_interior}
\end{aligned}
\end{equation}
\end{lemma}

The regularity conditions of Definition \ref{def:regular_point} exclude points at which a positive queue clears exactly at a period endpoint, because the state transition is not differentiable there. Consequently, the preceding sensitivity results do not apply at such points. The relevant endpoint-clearing
cases where the derivatives exist are treated directly in Sections \ref{section: two class fluid problem}
and \ref{section: more than two classes}.

\section{The Two-Class Case}\label{section: two class fluid problem}
In this section, we study the two-class scheduling problem to gain insights on how the value function changes as
\(\delta\) varies. For each review length \(\delta\), let
\(\vecu_i^*(\delta)\) denote the selected optimal period-\(i\)
allocation and let \(\vecx_i(\delta)\) denote its induced state. In
pointwise formulas throughout the scheduling results below and their
proofs, we suppress the argument \(\delta\). We retain it
whenever dependence on \(\delta\) is differentiated or used to define a
threshold.

We start by computing the optimal allocations for a fixed $\delta$. Theorem~\ref{thm:K_class optimal policy} implies that, in each
period \(i\), class \(1\) either receives all capacity or receives enough capacity to clear by the end of that period. Since there are only two
classes we also have $u_i^{2,*}(\delta)=1-u_i^{1,*}(\delta)$.
For  each period where class 1 cannot be emptied, we have
\(u_i^{1,*}=1\). Thus, define,
\begin{equation}
    \hat x_i^1(\delta)
    :=
    x_1^1+(i-1)\delta(\lambda^1-\mu^1).
\end{equation}
If $\hat x_i^1(\delta) /(\Delta_i(\delta)\mu^1)
    +
    (\lambda^1/\mu^1)
    >1$ for all $i \in \{1,\ldots,n(\delta)\}$,
then class \(1\) cannot be cleared in any period, even under full allocation.
Hence $u_i^{1,*}=1,
u_i^{2,*}=0$ for all $i\in\{1,\ldots,n(\delta)\}$
and the optimal policy is fully characterized independent of $\delta$. Hence, we assume $\mu^1>\lambda^1$ throughout this section. Define the first period in which class \(1\) can be cleared by,
\begin{equation}\label{eq:first_clearable_period}
    m(\delta)
    :=
    \min
    \left\{
    i\in\{1,\ldots,n(\delta)\}:
    \bar u_i^1(\delta)\le 1
    \right\}.
\end{equation}
For the fixed review length under consideration, write
\(m=m(\delta)\).
For every period \(i<m\), class \(1\) cannot be cleared, so $u_i^{1,*}=1,
u_i^{2,*}=0$.
Therefore, for \(i<m\), the states evolve according to,
\begin{align}
    x_{i+1}^1
    &=
    x_i^1+\delta(\lambda^1-\mu^1),\\
    x_{i+1}^2
    &=
    x_i^2+\delta\lambda^2.
\end{align}
In period \(m\), class \(1\) can be cleared. Theorem~\ref{thm:K_class optimal policy}
implies that the optimal allocation satisfies $\bar u_m^1
    \le
    u_m^{1,*}
    \le
    1.$
The exact value of \(u_m^{1,*}\) is the solution of a
one-dimensional optimization problem over the interval
\([\bar u_m^1,1]\). Fix \(u\in[\bar u_m^1,1]\). Since
\(u\ge \bar u_m^1\), class \(1\) clears by the end of period \(m\), and the
next-period class-\(1\) state is zero. In the following we consider the
nonterminal case \(m<n(\delta)\), if \(m=n(\delta)\), the same
minimization applies with the continuation value omitted. The class-\(2\) state after period \(m\)
is,
\begin{equation}\label{eq:two_class_x_mplus1_2}
    x_{m+1}^2(u)
    =
    \left[
    x_m^2+\Delta_m(\delta)\{\lambda^2-\mu^2(1-u)\}
    \right]^+.
\end{equation}
Thus,
\begin{align}\label{eq:two_class_one_dim_problem}
    u_m^{1,*}
    \in
    \arg\min_{u\in[\bar u_m^1,1]}
    \bigg\{
    &h^1
    \int_0^{\sigma_m^1(u)}
    \left[
    x_m^1+s(\lambda^1-\mu^1u)
    \right]ds \nonumber\\
    &+
    h^2
    \int_0^{\Delta_m(\delta)\wedge\sigma_m^2(u)}
    \left[
    x_m^2+s\left(\lambda^2-\mu^2(1-u)\right)
    \right]ds \nonumber\\
    &+
    v_{m+1}
    \left(
    0,\,
    x_{m+1}^2(u),\,
    \delta
    \right)
    \bigg\},
\end{align}
where,
\[
\sigma_m^1(u)
=
\frac{x_m^1}{\mu^1u-\lambda^1},
\qquad
\sigma_m^2(u)
=
\begin{cases}
\dfrac{x_m^2}{\mu^2(1-u)-\lambda^2},
& \text{if } \mu^2(1-u)>\lambda^2,\\[0.8em]
\infty,
& \text{otherwise}.
\end{cases}
\]

It remains to characterize the continuation after class \(1\) has cleared.
Suppose a period \(i>m\) begins with \(x_i^1=0\). Then allocating $u_i^1=\lambda^1/\mu^1$ keeps class \(1\) empty throughout the period. %
Therefore, by the
priority structure in Theorem~\ref{thm:K_class optimal policy}, there exists an optimal
allocation with,
$
u_i^{1,*}=\frac{\lambda^1}{\mu^1}$, and 
$
u_i^{2,*}=1-\frac{\lambda^1}{\mu^1},
$
in every subsequent period that begins with \(x_i^1=0\), provided
\(\lambda^1/\mu^1\le 1\).

\subsection{Definitions and Overview of the Results}
 To characterize the relationship between \(v_1\) and \(\delta\), define the thresholds,
\begin{align}
    \tilde{\delta} & := \frac{x_1^1}{\mu^1 - \lambda^1}\label{def: delta large enough}, \\
    \hat{\delta} & := \sup \{ \delta >0:  u_1^{1,*}(\delta) = 1 \} . \label{def: delta hat}
\end{align}
Here, $\tilde{\delta}$ is the largest $\delta$ for which assigning full first-period capacity to class $1$ does not clear class $1$ before the end of the period. Thus, for $0<\delta \leq \tilde{\delta}$, class $1$ either does not clear during the first period, or it clears exactly at the end of the period. 
\(\hat{\delta}\) is the largest value of \(\delta\) for which the optimal first-period allocation still gives full capacity to class \(1\), even if class \(1\) clears before the end of the period. These thresholds induce up to three regions: \emph{Region 1}: $(0,\tilde{\delta}]$, \emph{Region 2}: $(\tilde{\delta},\hat{\delta}]$, and \emph{Region 3}: $(\hat{\delta},T]$.
The second region is empty if \(\hat{\delta}=\tilde{\delta}\).

Figure \ref{fig: three regions overview} illustrates the three regions for an instance of the problem. As we show in the remainder of this section, the relationship between \(v_1\) and \(\delta\) can be non-monotone throughout Region 1, eventually becoming concave in Region 3, though it may be initially linear in Region 3. $v_1$ can be either convex or concave with respect to $\delta$ in Region 2, depending on whether the system empties by time $T$. We analyze each of the three regions separately next.

\begin{figure}
    \FIGURE{
    \includegraphics[width=0.65\textwidth]{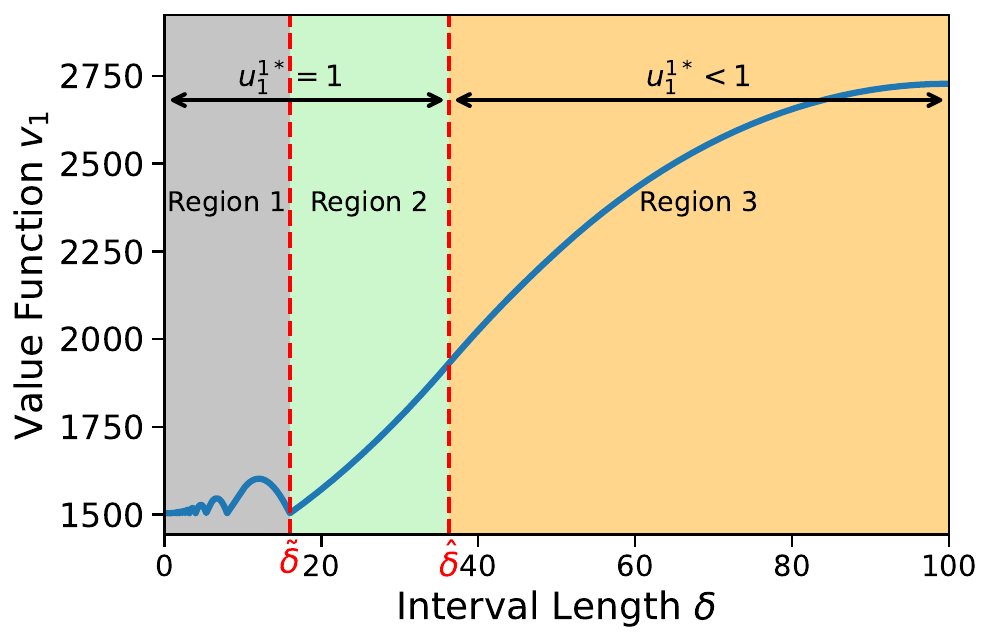}}
      {The relationship between the value function \(v_1\) versus interval length \(\delta\) for the two-class problem. \label{fig: three regions overview}}
      {The dotted red lines represent the boundaries between regions, defined by \(\tilde{\delta}\) and \(\hat{\delta} \) (see \eqref{def: delta large enough} and \eqref{def: delta hat}). The system parameters are \((\lambda^1, \lambda^2) = (0.5, 0.25)\), \(\mu^1 = \mu^2 = 1\), \((h^1, h^2) = (20, 1)\), \(T = 100\), and \((x_1^1, x_1^2 ) = (8, 4)\).}
\end{figure}

\subsection{Region 1}\label{section: region one}
  This region is defined by \(\delta \in (0, \tilde{\delta}]\), i.e., assigning full capacity to class 1 is not sufficient to empty it during the first period. In general, \(v_1\) is not monotone with respect to \(\delta\) in this region. We illustrate this and provide intuition using an example for the same example from Figure \ref{fig: three regions overview}. Figure \ref{fig:optimal traj grid} illustrates the optimal trajectories under continuous-time control as well as for \(\delta \in \{4, 6, 8, 14, 16\}\). For continuous-time control, the \(c\mu\) policy initially assigns full capacity to class \(1\) until its trajectory hits zero at time \(16\). However, it is not possible to achieve the same trajectories in discrete-time for all \(\delta>0\). We need to reassign service capacity at time \(16\) to attain the optimal trajectories of the continuous-time control problem. This is possible when \(\delta\) divides \(16\); for example when \(\delta \in \{4, 8, 16\}\), as seen in Figure \ref{fig:optimal traj grid}. For \(\delta\in \{6, 14\}\), the optimal trajectories differ from those of continuous-time control, as also seen in Figure \ref{fig:optimal traj grid}. As such, not merely frequency of control, but the specific timing of control is relevant for minimizing cost.

  Figure \ref{fig: region 1 zoom} depicts the value function in Region 1. Values of \(\delta\) where the continuous-time control optimal trajectories can be achieved result in the lowest possible cost; these are the valleys seen in the figure. At such \(\delta\) values, the optimal trajectories always have a period where all capacity is given to class \(1\), with it clearing at the end of that period. This can be observed in Figure \ref{fig:optimal traj grid} for \(\delta \in\{4, 8, 16\}\). At such points, \(v_1\) is not differentiable with respect to \(\delta\), as evident from Figure \ref{fig: region 1 zoom}. %

\begin{remark} As noted above, the nonmonotonicity in Region 1 arises from the periodic-review structure: changing \(\delta\) changes both the frequency and the locations of review epochs, so the feasible control sets for different \(\delta\) are generally not nested. Under an alternative nested formulation, in which more frequent review retains all existing decision epochs and adds new ones, monotonicity would be recovered. We focus on periodic review because it reflects the motivating settings in which decisions are revisited on a fixed cycle and the review-period length itself is the operational constraint of interest.
\end{remark}

\begin{figure}
    \FIGURE{
    \includegraphics[scale=0.25]{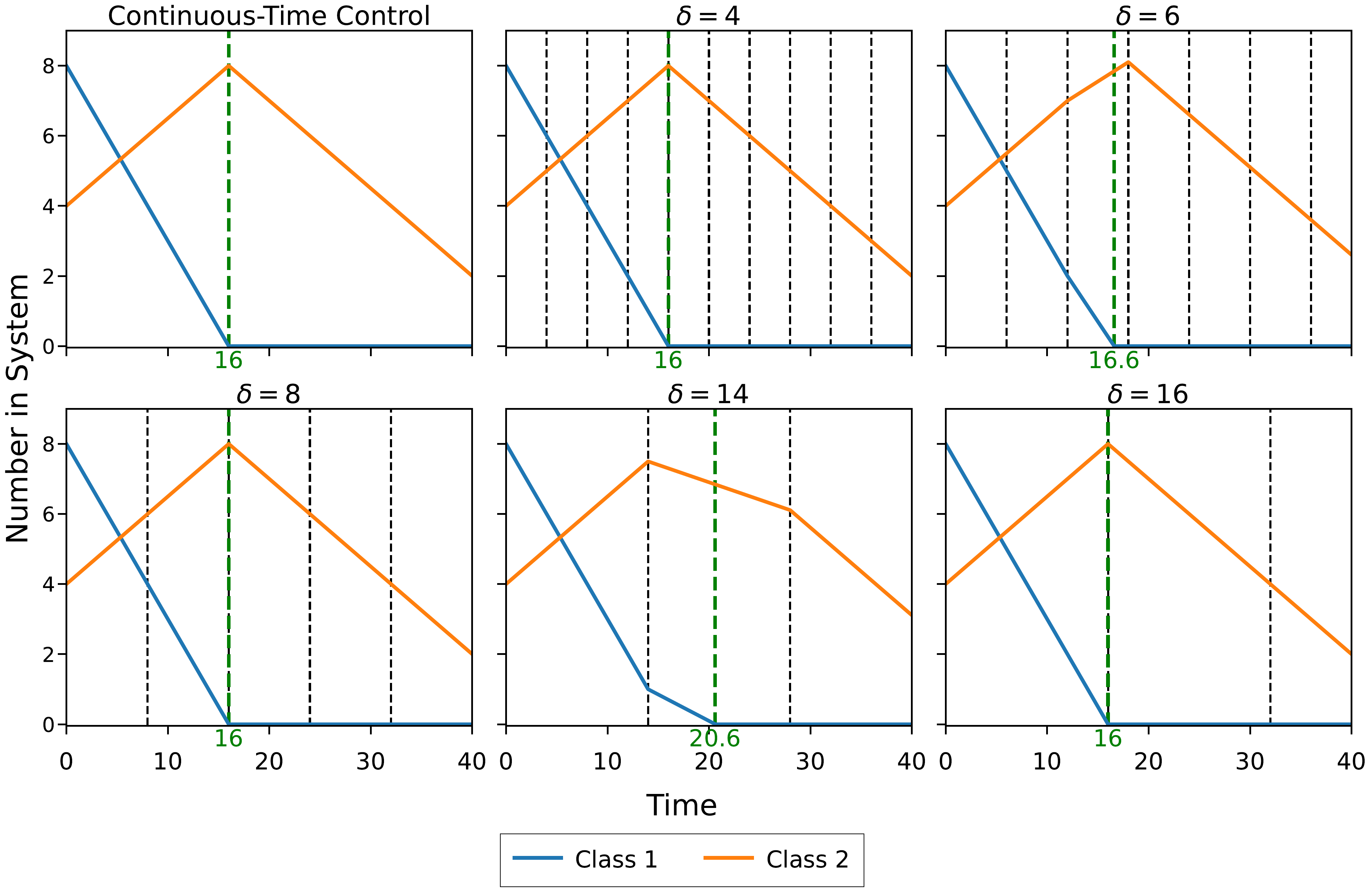}}
    {Comparison of optimal trajectories for various $\delta$ values in Region 1. \label{fig:optimal traj grid}}
      { The black dotted lines represent decision epochs where capacity can be reassigned, while the green dotted line and number indicate the time when class \(1\) empties out. The system parameters are the same as in Figure \ref{fig: three regions overview}.}
\end{figure}

\begin{figure}
    \FIGURE{
    \includegraphics[width=0.55\linewidth]{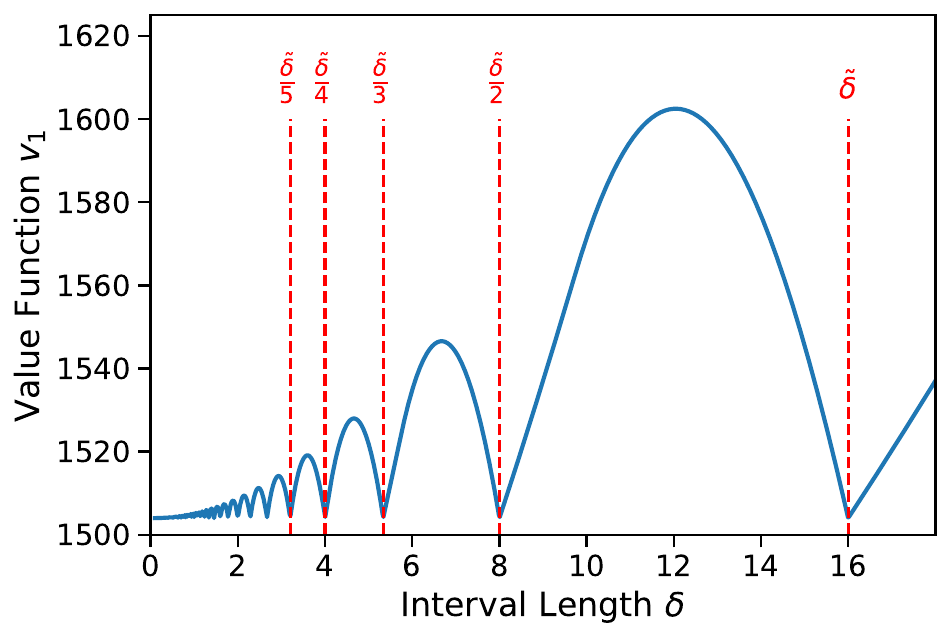}}
      {Value function \(v_1\) versus \(\delta\) for the two-class problem in Region 1. \label{fig: region 1 zoom}}
      { The red dotted lines show the first few divisors of \(\tilde{\delta}\); see \eqref{def: delta large enough}. The parameters are the same as in Figure \ref{fig: three regions overview}.} %
\end{figure}

\subsection{Regions 2 and 3: Preliminaries}\label{section: regions 2 and 3 overview}
  Before analyzing each region, we provide some general results for Regions 2 and 3. We show that in these regions, \(v_1\) is monotone increasing with respect to \(\delta\) and provide expressions for the derivatives of \(v_1\) with respect to \(\delta\). We focus our analysis on the nontrivial regime in which class 2 remains positive at the end of the first period. If both classes clear strictly before the end of the first period, then the value function is independent of $\delta$. 
  
  Region 2 consists of \(\delta\) values where \(u_1^{1,*}=1\) and \(u_1^{2,*}=0\), while Region 3 consists of \(\delta\) values where \(u_1^{1,*}<1\) and \(u_1^{2,*}>0\). Thus, \(\hat\delta\) can be viewed as a constraint-activity cutoff: the class-\(2\) nonnegativity constraint is active in Region 2 and inactive in Region 3. At this cutoff, the nonnegativity multiplier for class \(2\) is zero, or equivalently \(\mathcal M_1^1=\mathcal M_1^2\) and so SCS fails; see \eqref{eq:dynamic_marginal_value}. This cutoff tie does not create a first-order kink in \(v_1\). The optimizer derivative is zero in the interior of Region 2 and is given by the interior KKT derivative in Region 3. However, the second derivatives need not agree at \(\hat\delta\), where the optimal allocation may not be differentiable.

Recall that our sensitivity results do not apply when a class starting with positive queue clears exactly at the end of a period. However, we show that first-period endpoint clearing can be optimal throughout an open interval of review lengths. In such cases, if $\delta$ is large enough that class $1$ \emph{can} clear in the first period, it is optimal to assign maximal allocation to class $2$, without violating Theorem \ref{thm:K_class optimal policy}. Thus, in the first period, we assign the minimal capacity to clear class $1$ (i.e. at the end of the period), giving the rest of the allocation to class $2$. In this case, the clearing equality explicitly determines the optimizer on
that interval: \begin{equation}\label{eq:two_class_endpoint_allocation}
  u_1^{1,*}(\delta)
  =\bar u_1^1(\delta)
  =\frac{x_1^1}{\delta\mu^1}+\frac{\lambda^1}{\mu^1}.
\end{equation}
Consequently,
\begin{equation}\label{eq:two_class_endpoint_allocation_derivative}
  \frac{d u_1^{1,*}(\delta)}{d\delta}
  =-\frac{x_1^1}{\mu^1\delta^2},
  \qquad
  \frac{d u_1^{2,*}(\delta)}{d\delta}
  =\frac{x_1^1}{\mu^1\delta^2}.
\end{equation}
The value function can then be differentiated directly after
substituting \eqref{eq:two_class_endpoint_allocation} into the
cost expression. In the following, we denote by $\sigma^k_i$ the clearing time of class $k$ in period $i$ under the \emph{optimal policy}. %

  \begin{proposition}\label{thm: 2 class scheduling v_n der wrt delta}
    Assume $\sigma_1^2>\delta$. For \(\delta>\tilde{\delta}\) we have,
  \begin{equation}\label{eq: two class first derivative general}
      \pd{}{\delta}v_1 = 
  \begin{cases}
        h^2 \mu^2 (T-\delta)\left(u_1^{1,*}-\frac{\lambda^1}{\mu^1}\right), & \sigma_1^1 \ne \delta,\ \sigma_1^2>\delta,\ \sigma_2^2>T-\delta, \\
            h^2 x_2^2 \left(1+\dfrac{\lambda^2-\mu^2 u_1^{2,*}}{\mu^2 u_2^{2,*}-\lambda^2}\right), & \sigma_1^1 \ne \delta,\ \sigma_1^2>\delta,\ \sigma_2^2\le T-\delta, \\
          \dfrac{x_1^1 (h^1 \mu^1-h^2 \mu^2)}{2 \mu^1}, & \sigma_1^1=\delta.
  \end{cases}    
  \end{equation}
  \end{proposition}
  On regular subregions with \(\sigma_1^1<\delta\), Theorem \ref{thm:delta_der_value_function} gives the first derivative. The endpoint-clearing formula follows by direct differentiation using \eqref{eq:two_class_endpoint_allocation}. The expression depends on whether class \(2\) clears by the terminal time and on when class \(1\) clears. The conditions \(\sigma_1^2>\delta\) and \(\sigma_2^2>T-\delta\) imply that class \(2\) remains positive at the end of period \(1\), whereas conditions \(\sigma_1^2>\delta\) and \(\sigma_2^2\le T-\delta\) imply that class \(2\) remains positive at the end of period \(1\) and then clears later in the horizon. For the case with $\sigma_1^1=\delta$, we show that if endpoint clearing is optimal for any \(\delta>\tilde{\delta}\), then it is optimal in an open interval. In this case, the first derivative is constant and so \(v_1\) is locally linear. 

  The next result establishes the monotonicity of \(v_1\) with respect to \(\delta\) in Regions 2 and 3. 

\begin{theorem}\label{prop: 2 class first derivative non-negative}
      Assume $\sigma^2_1>\delta$. Then, $\pd{v_1}{\delta}>0$ wherever it exists for \(\delta \in (\tilde{\delta}, T)\). 
  \end{theorem}
Thus, once the review period is sufficiently long that class 1 can clear in the first period, longer review periods strictly increase cost. In this regime, more frequent control is therefore always valuable. The remaining question is how this marginal value changes with the review-period length, which is determined by the curvature of \(v_1\).

We next turn to the curvature. Unlike the first derivative, the second derivative depends on the local behavior of the optimal allocation. In particular, \(d u_1^{2,*}(\delta)/d\delta\) need not agree across the Region 2--Region 3 cutoff, and the continuation value changes form when class \(2\) clears exactly at the terminal time. The next result gives the second derivative away from these boundaries.

  \begin{proposition}\label{thm: 2 class scheduling v_n second der wrt delta}
    Let $\delta\in(\tilde{\delta},T)$ with $\delta\ne\hat{\delta}$,
and assume that $\sigma_1^2>\delta$ and
$\sigma_2^2\ne T-\delta$. If $\sigma_1^1=\delta$, assume
additionally that $\delta$ lies in the interior of an interval on which the optimal policy clears class~1
exactly at the end of the first period. Then the second derivative
of $v_1$ with respect to $\delta$ is given by,
  \begin{equation}
  \begin{aligned}
      \pdd{}{\delta}v_1
      &=
      \begin{cases}
          h^2 \mu^2 \left((\delta-T)\frac{d u_1^{2,*}(\delta)}{d\delta}
          -\left(u_1^{1,*}-\dfrac{\lambda^1}{\mu^1}\right)\right),
          & \sigma_1^1 \ne \delta,\ \sigma_2^2>T-\delta, \\
          \begin{aligned}
          h^2\Bigg[&
          (\lambda^2-\mu^2 u_1^{2,*})
          \left( 1+\dfrac{\lambda^2-\mu^2 u_1^{2,*}}{\mu^2 u_2^{2,*}-\lambda^2} \right) \\
          &-\frac{d u_1^{2,*}(\delta)}{d\delta}\mu^2
          \left( \delta+\dfrac{x_2^2+\delta(\lambda^2-\mu^2 u_1^{2,*})}
          {\mu^2 u_2^{2,*}-\lambda^2} \right)
          \Bigg],
          \end{aligned}
          & \sigma_1^1 \ne \delta,\ \sigma_2^2<T-\delta, \\
          0,
          & \sigma_1^1=\delta .
      \end{cases}
  \end{aligned}
  \end{equation}
  \end{proposition}
The cases in Proposition \ref{thm: 2 class scheduling v_n second der wrt delta} are the same as those in \ref{thm: 2 class scheduling v_n der wrt delta}. The additional condition in the endpoint-clearing case excludes
transitions between clearing class~1 exactly at the first review
epoch and clearing it strictly earlier, where the second
derivative need not exist. The boundary \(\hat{\delta}\) is excluded because \(u_1^{2,*}\) changes from zero in Region \(2\) to positive in Region \(3\), and the optimizer mapping need not be differentiable there. The boundary \(\sigma_2^2=T-\delta\) is the switch between the first two clearing cases.
  
\subsection{Region 2}\label{section: region two}
This region is defined by \(\delta \in \left(\tilde{\delta}, \hat{\delta}\right]\). Here, the optimal allocation during the first period assigns full capacity to class \(1\), regardless of any idleness incurred in class \(1\) after it clears. Thus \(u_1^{1,*}=1\), \(u_1^{2,*}=0\), and \(d u_1^{2,*}(\delta)/d\delta=0\) at interior points of Region 2. We note that \(\delta=\sigma_1^1\) is not possible in this region: since \(\delta>\tilde{\delta}\) and \(u_1^{1,*}=1\), class \(1\) clears strictly before the end of period \(1\).
  
Region 2 does not always exist. It requires \(h^1\mu^1\) to be sufficiently larger than \(h^2\mu^2\), because assigning full capacity to class \(1\) in the first period can cause class \(1\) to idle after it clears. In the following proposition, we show that on intervals in Region 2 where class \(2\) clears after period \(1\) and strictly before \(T\), \(v_1\) is strictly convex with respect to \(\delta\). On intervals in Region 2 where class \(2\) does not clear by \(T\), \(v_1\) is strictly concave with respect to \(\delta\).

\begin{theorem}\label{prop: region 2 second der}
      For \(\tilde{\delta}<\delta<\hat{\delta}\), assume $\sigma_1^2>\delta$ and \(\sigma_2^2\ne T-\delta\). We have
     \begin{equation}
      \pdd{}{\delta}v_1 = 
      \begin{cases}
                - h^2 \mu^2  \left(1-\frac{\lambda^1}{\mu^1} \right) < 0 , & \sigma_2^2>T-\delta, \\
             h^2 \lambda^2 \left(1+\frac{\lambda^2}{\mu^2 u_2^{2,*}-\lambda^2} \right)>0 , & \sigma_2^2<T-\delta.
      \end{cases}    
  \end{equation}
  \end{theorem}
   Theorem \ref{prop: region 2 second der} follows immediately from Proposition \ref{thm: 2 class scheduling v_n second der wrt delta}, plugging in \(u_1^{1,*}=1\), \(u_1^{2,*}=0\), and \(d u_1^{2,*}(\delta)/d\delta=0\). The endpoint \(\hat\delta\) is excluded from the second-derivative statement because it is a cutoff-tie point, and \(\sigma_2^2=T-\delta\) is excluded because it is the boundary between the two clearing cases. 
   
The two curvature regimes have different implications for the marginal value of more frequent control. If class 2 remains positive through the horizon, \(v_1\) is concave, so the marginal value of reducing \(\delta\) increases as control becomes more frequent. If class 2 clears before the end of the horizon, \(v_1\) is convex, so the marginal value of reducing \(\delta\) instead decreases as control becomes more frequent. The intuition is tied to the idleness created by the discrete-review restriction. In Region 2, all first-period capacity is allocated to class 1 even after it clears, while class 2 builds up. Reducing \(\delta\) allows this idle capacity to be reallocated to class 2 sooner. When class 2 remains positive through \(T\), this earlier reallocation is beneficial over the entire remaining horizon, so further reductions in \(\delta\) become increasingly valuable. When class 2 eventually clears, by contrast, the benefit of earlier reallocation ends once class 2 empties, leading to diminishing marginal value from further reductions in \(\delta\).

\begin{figure}
    \FIGURE{
    \includegraphics[scale=0.55]{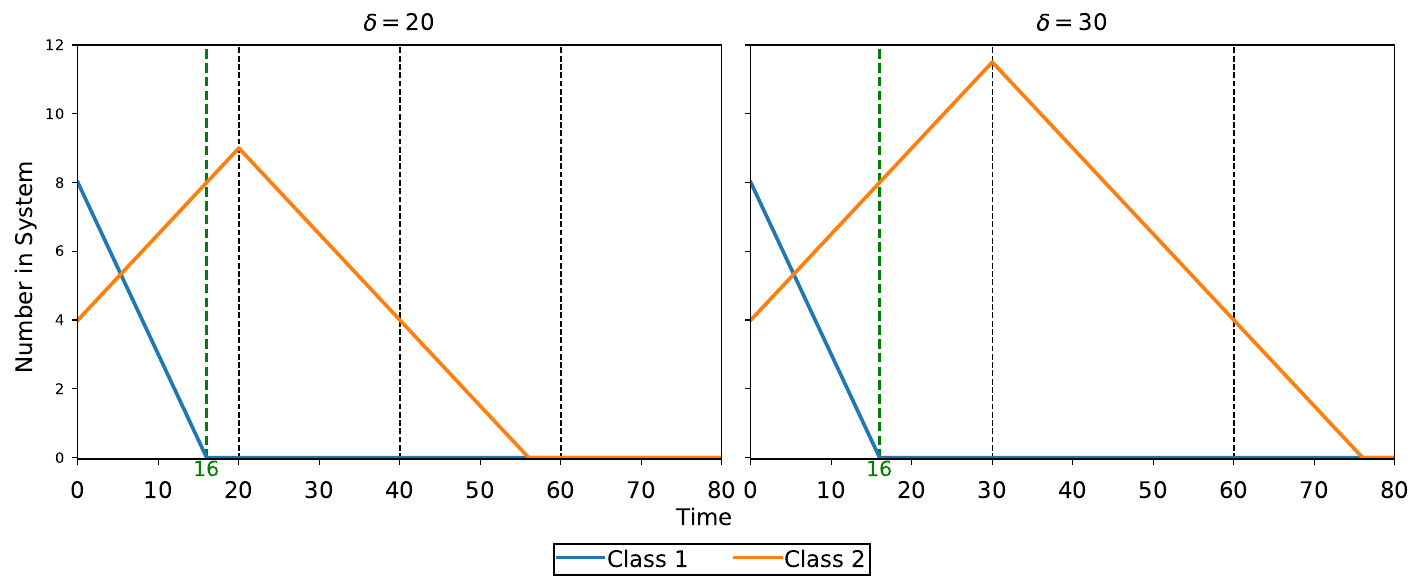}}
    {Comparison of optimal trajectories for two $\delta$ values in Region 2. \label{fig: region 2 eg optimal traj}}
      { The black dotted lines represent decision epochs where capacity can be reassigned, while the green dotted line and number indicate the time when class \(1\) empties out. The system parameters are the same as Figure \ref{fig: three regions overview}.}
\end{figure}

\subsection{Region 3}\label{section: region three}
This region is defined by \(\delta \in \left(\hat{\delta}, T \right]\). Here, the optimal allocation during the first period assigns less than full capacity to class \(1\), and hence class \(2\) receives positive capacity. First, we show that on intervals where class \(1\) clears at the end of the first period, \(\sigma_1^1=\delta\), \(v_1\) is linear with respect to \(\delta\), independent of whether class \(2\) clears by \(T\). 
\begin{theorem}\label{prop: Region 3 Linearity}
  On any interval contained in \(\{\delta:\delta>\hat{\delta},\ \sigma_1^1=\delta\}\), \(v_1\) is linear in \(\delta\). 
\end{theorem} Equivalently, \(\pdd{}{\delta}v_1=0\) at every point in such an interval. Theorem \ref{prop: Region 3 Linearity} follows immediately from Proposition \ref{thm: 2 class scheduling v_n der wrt delta}, where \(\pd{}{\delta}v_1\) is constant for \(\delta=\sigma_1^1\). Clearing class \(1\) right at the end of the first period, even though it could clear earlier with more capacity, is only optimal if class \(2\) has a \(c\mu\) index close to that of class \(1\), making it suboptimal to incur class-\(1\) idleness mid-period. 
Otherwise, if class \(1\) clears before the end of the first period, \(\sigma_1^1<\delta\), \(v_1\) is strictly concave with respect to \(\delta\) on intervals that do not cross the terminal-clearing boundary for class \(2\).

Here, the optimal allocation adjusts with \(\delta\) so that class 1 clears exactly at the end of the first review period, avoiding within-period idleness. As the review period changes, capacity is rebalanced between the two classes to preserve this endpoint-clearing structure. Consequently, the marginal value of reducing the review period remains constant. Thus, linearity corresponds to a regime in which the system can absorb changes in review frequency without creating additional idleness.

\begin{theorem}\label{prop: Region 3 concavity}
      Assume \(\delta>\hat{\delta}\), \(\sigma_1^1<\delta\), and \(\sigma_2^2\ne T-\delta\). We have \(\pdd{}{\delta}v_1<0\).
\end{theorem}
If \(\delta>\tilde{\delta}\) and \(u_1^{1,*}=1\), class \(1\) clears before the end of the first period. Thus, if there exists a local neighborhood of \(\delta\) values where it is optimal to clear class \(1\) at the end of the first period, then \(\tilde{\delta}=\hat{\delta}\) and Region 2 is empty. From theorems \ref{prop: region 2 second der} and \ref{prop: Region 3 Linearity}, \(v_1\) may either be convex or linear for some \(\delta\) values, but not both. It is further eventually concave for large enough $\delta$. The concavity in Region 3 also reflects the role of idleness, but now the optimal allocation can adjust to mitigate it. Because class 1 clears strictly before the end of the first period, discrete review creates within-period idleness. However, as \(\delta\) increases, the optimal policy shifts more first-period capacity toward class 2, delaying the clearing of class 1 and reducing both idleness and class-2 buildup. This endogenous adjustment partially offsets the cost of a longer review period, so the marginal cost of increasing \(\delta\) decreases with \(\delta\). Equivalently, the marginal value of reducing the review period increases as control becomes more frequent. Thus, modest reductions in a relatively long review period may yield limited gains, and substantial reductions may be needed before the benefits of more frequent control become pronounced. We illustrate these results further numerically in Section \ref{section: fluid model experiments}.

\section{More than Two Classes}\label{section: more than two classes}

We next discuss the general \(K\)-class fluid scheduling problem. For \(K>2\), the value function can similarly be non-monotone in \(\delta\). %
For each \(k\in\{1,\ldots,K-1\}\) satisfying
\(\sum_{l=1}^k\lambda^l/\mu^l<1\), define
\(\tilde{\delta}^k\) as the period length at which classes
\(1,\ldots,k\) can all be cleared exactly at the end of the first
period using all available capacity, while classes \(k+1,\ldots,K\)
receive no capacity. Formally,
\begin{equation}\label{eq: tilde delta general def}
      \tilde{\delta}^k :=
      \frac{\sum_{l=1}^k x_1^l/\mu^l}
      {1-\sum_{l=1}^k\lambda^l/\mu^l}.
  \end{equation}
Note that \(\tilde{\delta}\) given by \eqref{def: delta large enough} for the two-class problem is a special case of \eqref{eq: tilde delta general def}, with \(k=1\). We further define
\begin{equation}\label{eq:exceptional_review_set}
    \mathcal D:=\left\{\frac{\tilde{\delta}^k}{q}: k=1,\dots,K-1,\ q\in\mathbb N\right\}.
\end{equation}
Thus, \(\delta\in\mathcal D\) are review periods for which some \(\tilde{\delta}^k\) is an integer multiple of \(\delta\). At such \(\delta\) values, the discrete review epochs can coincide with a class clearing time, and \(v_1\) may fail to be differentiable with respect to \(\delta\).

Figure \ref{fig: multiclass example} illustrates the value function \(v_1\) as a function of \(\delta\), for a \(K = 3\) class fluid scheduling problem with parameters \((\lambda^1, \lamda^2, \lambda^3) = (0.5, 0.15, 0.12)\), \((\mu^1, \mu^2, \mu^3) = (1, 1, 1)\), \((h^1, h^2, h^3) = (8, 6, 4)\), \((x_1^1, x_1^2, x_1^3 ) = (8, 7, 5)\), \(T = 100\). Although similar oscillations in the value function were observed in Region \(1\) of the \(2\)-class case (as discussed in Section \ref{section: region one} and in Figure \ref{fig: region 1 zoom}), the oscillations in Figure \ref{fig: multiclass example} differ in that the valleys are no longer equal, and the peaks no longer increase successively. From Figure \ref{fig: multiclass example} we also observe that \(v_1(\vecx_1, \delta)\) is not differentiable at the plotted values of \(\delta \in \mathcal D\).

\begin{figure}
    \FIGURE{
    \includegraphics[scale=0.6]{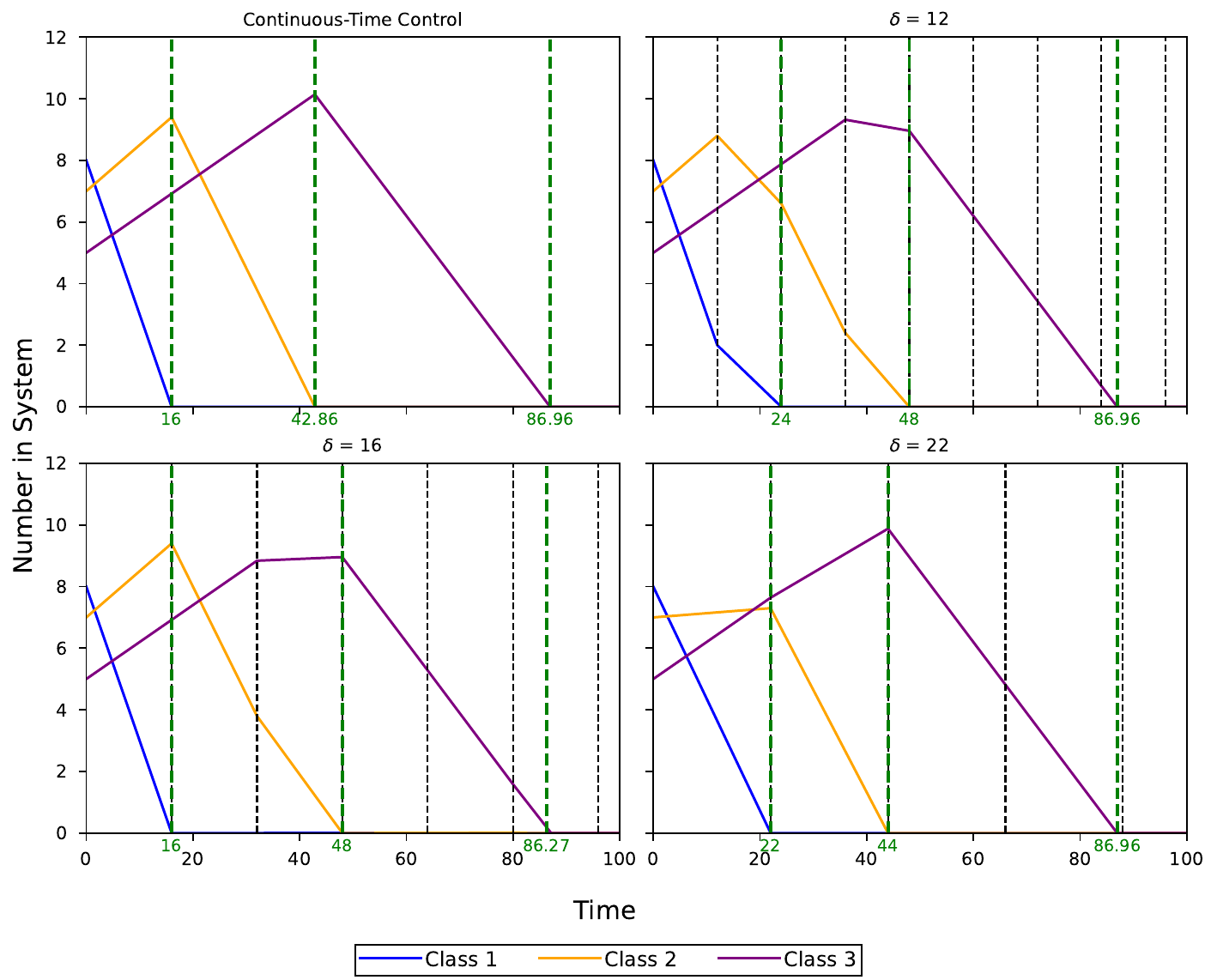}}
    {Optimal trajectories for various $\delta$ values for the three-class fluid scheduling problem. \label{fig:optimal traj grid 3 class}}
    { The black dotted lines represent decision epochs where capacity can be reassigned, while the green dotted lines and numbers indicate when each of the three classes clear out.  The system parameters are \((\lambda^1, \lamda^2, \lambda^3) = (0.5, 0.15, 0.12)\), \((\mu^1, \mu^2, \mu^3) = (1, 1, 1)\), \((h^1, h^2, h^3) = (8, 6, 4)\), \((x_1^1, x_1^2, x_1^3 ) = (8, 7, 5)\), \(T = 100\).}
\end{figure}

To gain insights, we next examine the optimal trajectories for the same example. Figure \ref{fig:optimal traj grid 3 class} depicts the optimal trajectories for certain \(\delta\) values. In the \(2\)-class problem, we observed that at \(\delta\) values that were divisors of \(\tilde{\delta}\), \(v_1\) was the same as the continuous control value function and the optimal trajectories perfectly matched the continuous-time control trajectories. In the general \(K\)-class problem, however, matching the continuous-time trajectories may require the review epochs to align with multiple class-clearing times. We first observe the optimal trajectories under continuous-time control in Figure \ref{fig:optimal traj grid 3 class}. In order to obtain such trajectories for a value of \(\delta>0\), we would need to be able to reassign capacity when class \(1\) clears at time \(16\) and again when class \(2\) clears at time \(42.86\). Thus, a review period that aligns with one of these clearing times need not align with the other. Exact matching is possible when both clearing times are review epochs, but different \(\delta\) values may align with different subsets of these times, resulting in the observed oscillations of \(v_1\).

\begin{figure}
    \FIGURE{
    \includegraphics[width=0.6\textwidth]{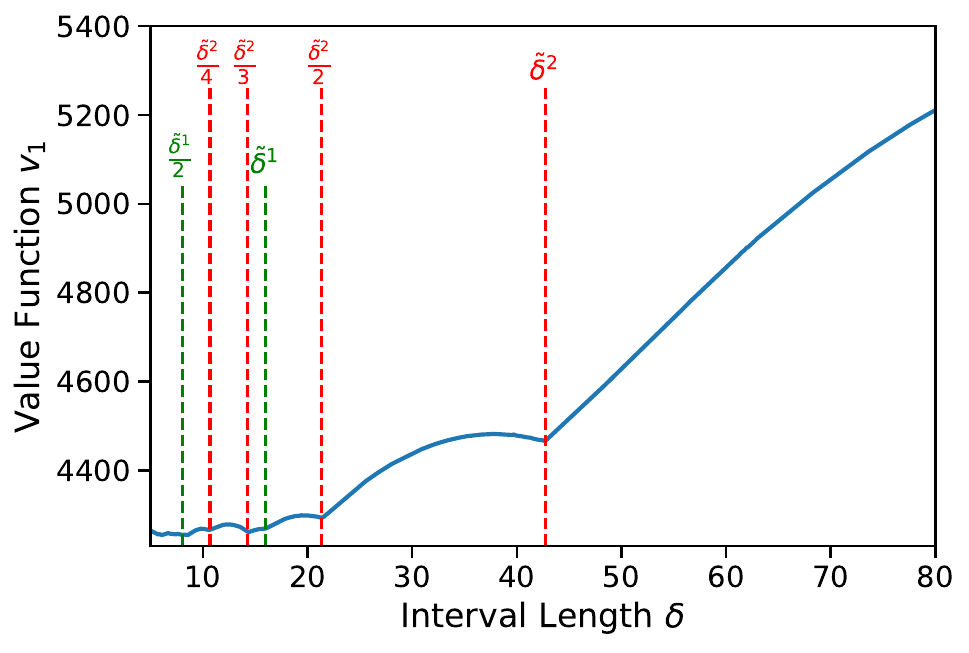}}
    {Value function \(v_1\) versus interval length \(\delta\) for the \(3\) class fluid scheduling problem. \label{fig: multiclass example}}
    { The green and red dotted lines show the first few divisors of \(\tilde{\delta}^1\) and \(\tilde{\delta}^2\) respectively; see \eqref{eq: tilde delta general def}. The system parameters are \((\lambda^1, \lamda^2, \lambda^3) = (0.5, 0.15, 0.12)\), \((\mu^1, \mu^2, \mu^3) = (1, 1, 1)\), \((h^1, h^2, h^3) = (8, 6, 4)\), \((x_1^1, x_1^2, x_1^3 ) = (8, 7, 5)\), \(T = 100\).}
\end{figure}

We consider the regime in which classes $1,\ldots,K-1$, clear strictly before the end of the first period while class $K$ remains positive. On every smooth interval in this regime, we show that $v_1$ is increasing and that the convex–concave structure of the two-class problem reappears. 

\begin{theorem}
\label{prop:K_class_eventual_shape}
Let \(\delta\) be a regular point and assume that
$\sigma_1^k<\delta$ for $k \in \{1,\dots,K-1\}$ and 
$\sigma_1^K>\delta$. Then,
\begin{enumerate}
\item[(i)]
\(\pd{}{\delta}v_1>0\).

\item[(ii)]
If \(\sigma_2^K>T-\delta\), then
\(\pdd{}{\delta}v_1<0\).

\item[(iii)]
Suppose \(\sigma_2^K<T-\delta\). If \(u_1^{K,*}=0\), then $\pdd{}{\delta}v_1>0$. If \(u_1^{K,*}>0\), then
\(\pdd{}{\delta}v_1<0\).
\end{enumerate}
\end{theorem}
The result implies that \(v_1\) is strictly increasing on every interval satisfying the assumptions of Proposition \ref{thm: K class scheduling first der}. On such an interval $v_1$ is strictly concave if class \(K\) does not clear by
\(T\); if class \(K\) clears by \(T\), it is strictly convex when class \(K\) receives no first-period capacity and strictly concave when it receives positive first-period capacity. In Appendix \ref{ap:4 class} we provide numerical examples of the value function \(v_1\) versus \(\delta\) for a four class problem, where we observe similar behavior in the monotonicity and differentiability of \(v_1\) with respect to \(\delta\).

\section{Numerical Study}\label{sec:numerics}
In this section, we conduct numerical studies to investigate (1) the impact of different parameters on the structure of the value function as $\delta$ varies, in particular its curvature in Regions 2 and 3; and (2) the generalization of the results for the original stochastic  scheduling problem of Section \ref{sec:prob}.

\subsection{Fluid Model Experiments}\label{section: fluid model experiments}
We investigate the value of more frequent control under different parameter regimes. To quantify this value, we compute the relative cost increase compared to continuous-time control, defined as, 
\begin{equation}
    \frac{v_1(\vecx_1, \delta) - v_1(\vecx_1, 0)}{v_1(\vecx_1 , 0)}.
\end{equation}
We consider a two-class system and a horizon of length \(T = 100\). All systems are initiated with the same initial condition \(\vecx_1^\intercal = (8,4)\). We vary the utilization \(\rho := \lambda^1/\mu^1+\lambda^2/\mu^2\); the load ratio \(r := (\lambda^1/\mu^1)/(\lambda^2/\mu^2)\); and the \(c\mu\)-index ratio \(\nu := (h^1\mu^1)/(h^2\mu^2)\). As we are particularly interested in studying the curvature of the value function in Regions 2 and 3, the horizon length was chosen to be long enough that the size of Region 1 across \(\delta\) was minimized and the queues to empty in cases with the lowest ratio between their \(c\mu\) indexes. 

Figure \ref{fig: fluid numerics 1} plots the relative cost increase as a function of \(\delta\) for \(\rho \in \{0.7, 0.9\}\), \(r \in \{0.5, 1, 2\}\), and \(\nu \in \{2,5,20\}\), assuming \(\mu^1=\mu^2=1\). A key observation is that systems whose optimal policy clears class 1 quickly, thereby creating class-1 idleness while class 2 builds up, are generally more sensitive to the review-period length. In contrast, when the optimal allocation shifts capacity toward class 2 and thereby reduces or avoids such idleness and buildup, the cost is less sensitive to \(\delta\).

Larger values of \(\nu\) place greater priority on class 1, making it optimal to clear class 1 more quickly and resulting in idleness. Hence, the resulting loss of reallocation flexibility becomes more costly as \(\delta\) increases. In these experiments, systems with smaller \(\nu\) tend to empty sooner because their more balanced allocations generate less class-1 idleness and class-2 buildup. A horizontal plateau occurs once the first review period is sufficiently long to empty the system, after which further increases in \(\delta\) do not affect cost.

The experiments also illustrate the curvature regimes characterized in Section \ref{section: two class fluid problem}. For large \(\delta\), the relative cost increase is concave, consistent with Theorem \ref{prop: Region 3 concavity}. At intermediate review periods, the \(\nu=2\) systems exhibit the linear regime of Theorem \ref{prop: Region 3 Linearity}, in which the allocation adjusts so that class 1 clears at the review epoch and avoids within-period idleness. The \(\nu=20\), \(\rho=0.7\) systems instead exhibit the convex regime of Theorem \ref{prop: region 2 second der}, where all first-period capacity continues to be assigned to class 1 and both classes eventually clear within the horizon. The remaining examples are predominantly concave. 

These curvature differences have direct implications for the value of more frequent control. Within the convex regime, the marginal benefit of reducing the review period is greater at longer review periods and diminishes as the review period becomes shorter. Within the concave regime, the marginal benefit instead increases as the review period becomes shorter. The linear regime corresponds to a constant marginal benefit. These differences are visible in the \(\rho=0.7\) systems. Although the \(\nu=20\) systems incur the largest relative cost increase under very infrequent control, their relative cost increase declines rapidly as the review period is reduced and eventually falls below that of \(\nu=5\), and for \(r=2\), below that of \(\nu=2\).

One exception to the general ordering by \(\nu\) occurs when both utilization and the relative load of class 1 are high. For example, when \((\rho,r)=(0.9,2)\), the \(\nu=20\) system has a smaller relative cost increase than the \(\nu=5\) system over much of the range. This comparison is in relative terms: with high class-1 load and a large \(c\mu\)-index ratio, class 1 accounts for a larger share of the continuous-control baseline cost. Consequently, the additional cost caused by infrequent reallocation can represent a smaller fraction of the baseline cost, even though the policy may generate more within-period idleness.


In the above experiments we assumed \(\mu^1=\mu^2\). The same qualitative observations also hold when the service rates are not equal; see Figure \ref{fig: fluid numerics 2} in Appendix \ref{ap:numerics}.

\begin{figure}
    \FIGURE{
    \includegraphics[scale=0.5]{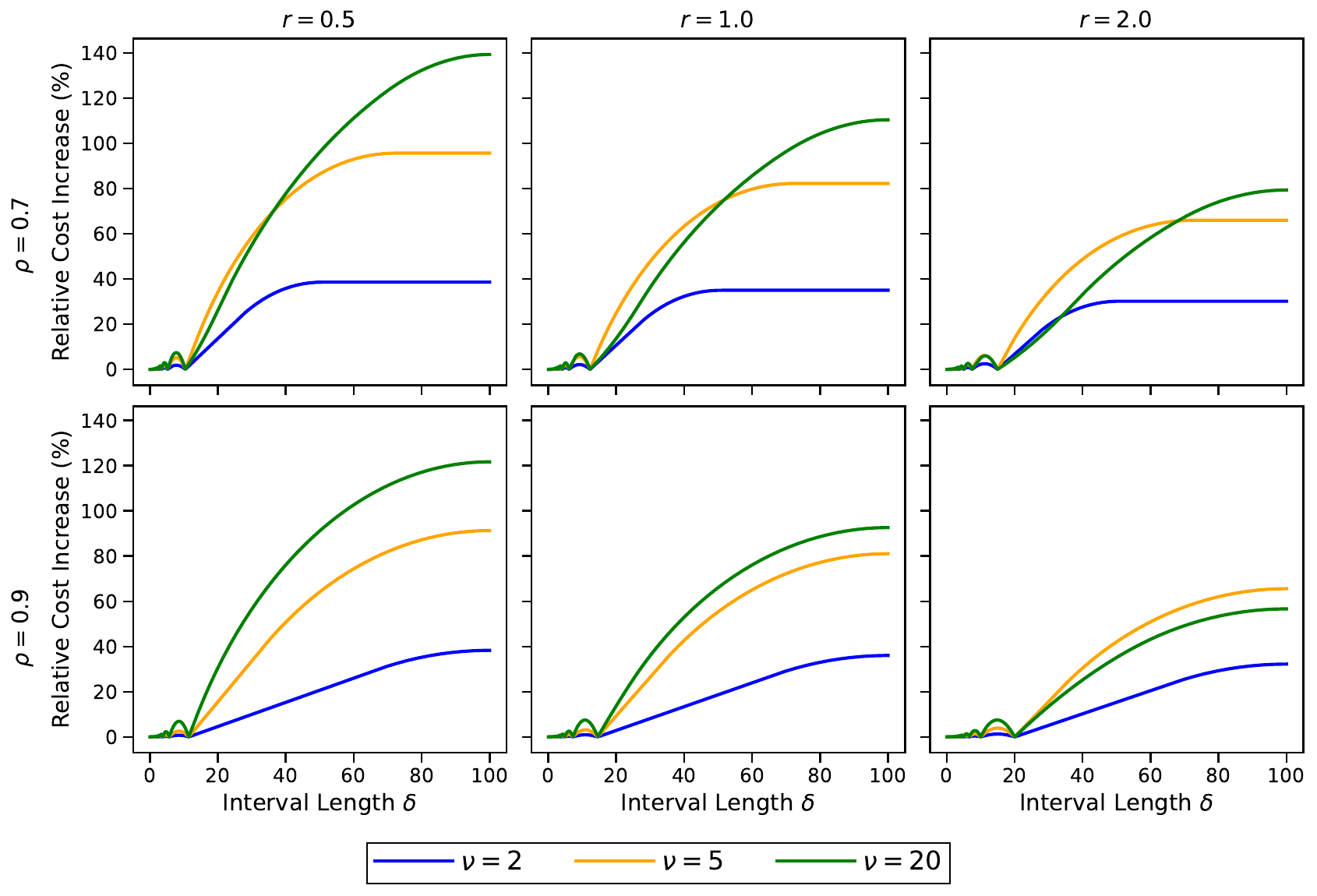}}
    {Relative cost increase compared to continuous-time control versus $\delta$, for various cost ratios $\nu$, load ratios $r$, and utilizations $\rho$. \label{fig: fluid numerics 1}}
    { We keep $\rho$ constant across columns and $r$ constant across rows. We fix $\mu^1 = \mu^2 = 1$, and \(\vecx_1 = (8,4)\) fixed for all systems. The system parameters that vary are as follows: For the $\rho = 0.7$ systems (from left to right) , $(\lamda^1, \lamda^2) = (0.23, 0.47), (0.35, 0.35), (0.47, 0.23)$. For the $\rho = 0.9$ systems, (from left to right), $(\lamda^1, \lamda^2) = (0.3, 0.6),(0.45, 0.45),  (0.6, 0.3)$}
\end{figure}

\subsection{Stochastic System Experiments}\label{sec: stochastic system experiments}
In this section, we numerically investigate the generalization of our results for the stochastic problem. %
To this end, we first formulate the stochastic problem as a discrete-time finite-horizon MDP with a bounded state-space. Recall that $\textbf{X}(t):=(X^1(t),\ldots, X^K(t))^\intercal$ keep track of the number of customers of each class in system at time $t$. We truncate the state-space by introducing a maximum number-in-system $X_{\max}$ such that the arrivals are rejected once the total number of customers in system reaches $X_{\max}$. Hence, the state-space is given by, 
$
    \mathcal{S}=\{(X^1,X^2)^\intercal\in \mathbb{Z}^+\times \mathbb{Z}^+;X^1+X^2\leq X_{\max}\}.
$ Let $\textbf{Z}^\intercal:=(Z^1,Z^2)$ denote the server allocation vector corresponding to the number of servers allocated to each class.  The set of feasible server allocations is, 
$
    \mathcal{Z} = \{(Z^1,Z^2)^\intercal\in \mathbb{Z}^+ \times \mathbb{Z}^+; Z^1+Z^2=N \}.
$
The value function in period \(i\) starting from state \(\textbf{X}\), when period lengths are equal to \(\delta\), then satisfies the optimality equations,
\begin{align}\label{eq:MDP}
     V_i(\textbf{X},\delta) = \min_{\textbf{Z}\in\mathcal{Z}}  \mathbb{E} \left[\int_0^{\Delta_i(\delta)} \textbf{h}^\intercal \textbf{X}(s)ds | \textbf{X}(0)=\textbf{X}\right] + \nonumber \\ 
     \sum_{\textbf{X}'\in\mathcal{S}} \mathbb{P}(\textbf{X}(\Delta_i(\delta))=\textbf{X}'|\textbf{X}(0)=\textbf{X}) V_{i+1}(\textbf{X}',\delta),
\end{align}
with terminal condition \(V_{n(\delta)+1}(\cdot,\delta)=0\).
Solving this MDP exactly is difficult, because the transition probabilities, as well as the expectation in \eqref{eq:MDP} rely on transient dynamics and hence cannot be obtained explicitly. Instead, we simulate sample paths of the queueing dynamics over each period and estimate the expectations. We then use backward induction and full enumeration of all feasible allocations to compute the optimal value function for different values of $\delta$.

 We examine a sequence of stochastic systems indexed by \(\eta\). For class \(k\) in the \(\eta\)-th system, the arrival and service rates satisfy \(\lambda_\eta^k := \lambda^k \eta\) and \(\mu_\eta^k := \mu^k \eta\), respectively, with initial conditions \((\eta X^1, \eta X^2)\). The scaled value function for the \(\eta\)-th system is defined as \(V_1/\eta\), where \(V_1\) is the value function for the \(\eta\)-th system. For a given sequence of systems, the ``base'' system refers to the system with \(\eta = 1\). The number of servers $N$ remains fixed as the scaling parameter $\eta$ increases. Thus, feasible server allocations remain restricted to fractions of $1/N$. The scaled stochastic control problem therefore need not converge to the continuous-allocation fluid problem. We interpret the experiments instead as a robustness check of whether the qualitative sensitivity patterns predicted by the fluid model persist in the stochastic system.

\textbf{Experiments}. We consider two sequences of systems and present the results for \(\eta \in \{1,5, 10\}\). Throughout all experiments, we keep the total number of servers fixed at \(N = 10\).  For each experiment, the service rates for the base system are \(\mu_1^1=\mu_1^2 = 0.1\). We select horizon lengths large enough that the queues empty with high probability during the horizon. To verify this, we compare the average queue lengths at the end of horizon with the initial conditions. The first base system has parameters \((\lambda^1, \lambda^2) = (0.35, 0.3)\), \((h^1 , h^2) = (3, 1)\), \((X^1, X^2) = (9, 1)\), \(T = 40\). The second base system has parameters \((\lambda^1, \lambda^2) = (0.35, 0.25)\), \((h^1 , h^2) = (20, 1)\), \((X^1, X^2) = (9, 1)\), \(T = 40\). 

In our experiments, we choose a large \(X_{\max}\) such that almost no arrivals are rejected. For the first system we set \(X_{\max} \in \{30, 75 ,130\}\), for \(\eta \in \{1, 5, 10\}\). For the second system we set \(X_{\max} \in \{35, 80, 135\}\), for \(\eta \in \{1, 5, 10\}\). To ensure this is sufficiently large, we repeat the experiments with \(X_{\max}\in\{35, 86, 150\}\) for the first system and \(X_{\max}\in\{40, 92, 155\}\) for the second system (\(15\%\) larger) to make sure the value functions remain approximately the same (see Appendix \ref{ap:numerics}). We use \(400, 250\), and \(100\) samples for \(\eta \in \{1, 5,10\}\), respectively, when estimating the expectations. To make sure the sample sizes do not affect the choice of optimal policies, we repeat the experiments after doubling the number of samples. We use common random numbers (CRN) in all experiments to reduce the variance. We also run additional replications for the selected optimal policy to obtain more accurate estimates of the corresponding costs when presenting the results.

\begin{figure}
    \FIGURE{
    \includegraphics[scale=0.4]{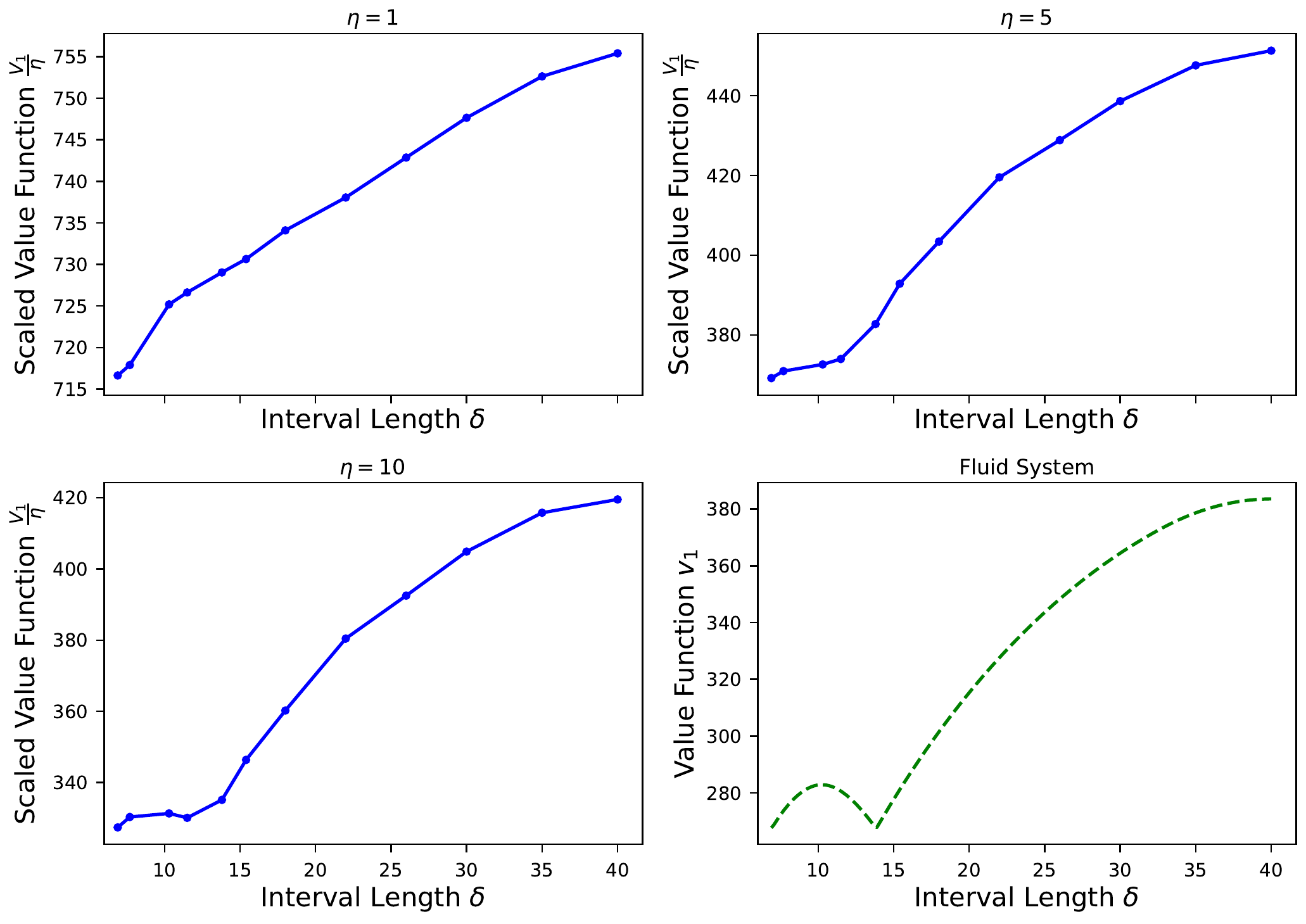}}
    {Scaled value function \(V_1 /\eta\) versus \(\delta\) for a sequence of stochastic systems along with a plot of the fluid value function \(v_1\) with the same parameters as the base stochastic system.\label{fig: stoch sys_1}}
    { The base $\eta=1$ parameters are: $(\lambda^1, \lambda^2) = (0.35, 0.3)$, $(h^1 , h^2) = (3, 1)$, $(X^1, X^2) = (9, 1)$, $T = 40$.}
\end{figure}

\begin{figure}
    \FIGURE{
    \includegraphics[scale=0.4]{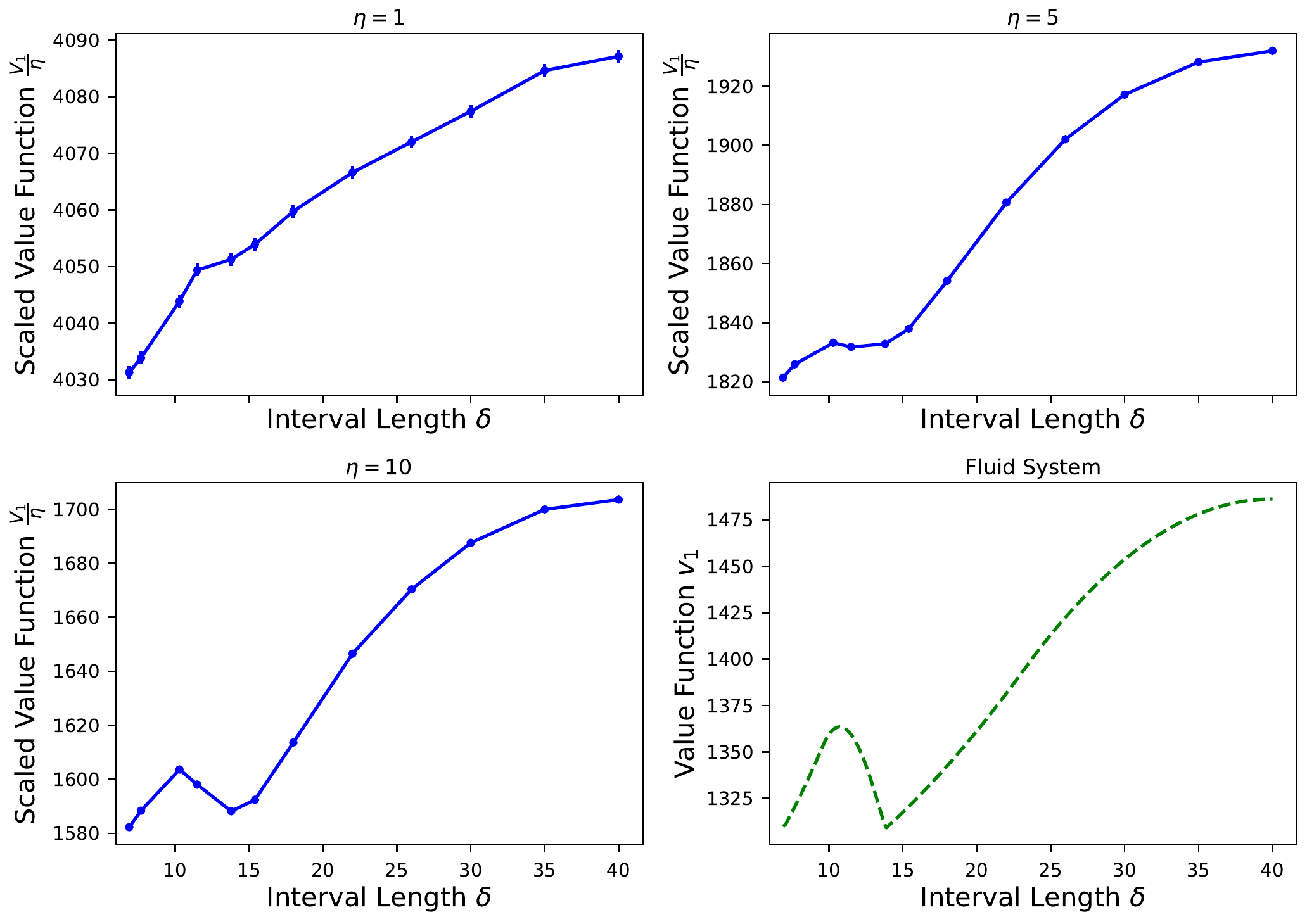}}
    {Scaled value function \(V_1 /\eta\) versus \(\delta\) for a sequence of stochastic systems along with a plot of the fluid value function \(v_1\) with the same parameters as the base stochastic system.\label{fig: stoch sys_2}}
    { The base $\eta=1$ parameters are: $(\lambda^1, \lambda^2) = (0.35, 0.25)$, $(h^1 , h^2) = (20, 1)$, $(X^1, X^2) = (9, 1)$, $T = 40$.}
\end{figure}

\textbf{Results}. Figure \ref{fig: stoch sys_1} presents value function estimates (with 95\% confidence intervals) for the first sequence of systems for \(\eta \in \{1, 5, 10\}\), along with a plot of the associated fluid value function \(v_1\). We observe that the non-monotonicity predicted by the fluid model only begins to appear in the stochastic system as \(\eta\) is increased. When \(\eta = 1\), however, the value function appears monotone with respect to \(\delta\), with the oscillations smoothed out. Intuitively, the oscillations in the fluid model arise from the exact alignment of review epochs with deterministic queue-clearing times. In the stochastic system, clearing times are random, weakening this alignment effect and thereby smoothing the oscillations. For \(\eta = 1\) the curve is concave with respect to \(\delta\) everywhere, and the systems with \(\eta \in \{5, 10\}\) appear concave in the same \(\delta\) regions as in the fluid model. %

Similarly, Figure \ref{fig: stoch sys_2} plots the second sequence of systems for \(\eta \in \{1, 5, 10\}\), with the associated fluid problem. As with the previous system, we observe a monotone stochastic value function when \(\eta = 1\), with non-monotonicity appearing for \(\eta \in \{5, 10\}\) in the same regions predicted by the fluid problem. We also observe that the stochastic systems are initially convex before becoming concave with respect to \(\delta\), consistent with the associated fluid value function. As \(\eta\) increases, the scaled stochastic value functions also become closer in magnitude to the continuous-allocation fluid value function in these instances.

\section{Conclusions}\label{sec:conc}
We investigate the sensitivity of the value function of a multiclass scheduling problem with respect to the review-period \(\delta\). We do so by analyzing the sensitivity of the associated fluid dynamic programming equations. We characterize first- and second-order derivatives of the value function and use them to identify monotonicity and curvature regimes.

We show that the value function need not be monotone in \(\delta\) for small review periods. This arises because the timing of review epochs rather than merely the frequency of control matter in this regime. In the two-class case, once \(\delta\) is large enough that class \(1\) clears in the first period, the value function is increasing in \(\delta\). Its curvature depends on the clearing regime: it can be linear, strictly convex, or strictly concave before eventually becoming concave near the no-control regime.

Idleness is central to these effects. When the high-priority class has a much larger \(c\mu\) index, the optimal policy may clear that class quickly, creating class-\(1\) idleness within a review period while the lower-priority queue builds up. This loss of flexibility increases the value of more frequent control. The stochastic experiments exhibit the same broad structure as the fluid model, but the initial non-monotonicity is smoothed out. Intuitively, deterministic fluid dynamics allow exact alignment of review epochs with clearing times, whereas stochastic dynamics make such exact timing less persistent.

Our study appears to be the first to examine the value of more frequent control in queueing control, and hence suggests several directions for future work. Similar sensitivity questions arise in other queueing control problems where continuous intervention is infeasible or undesirable, such as dynamic service-rate control, staffing, routing, and pricing. More broadly, developing sensitivity theory for general MDPs would be an interesting direction for future work. The finite-horizon stochastic DP is challenging because the transition probability itself depends on \(\delta\); discounted infinite-horizon formulations may offer a more tractable setting for implicit differentiation.

\bibliographystyle{informs2014} %
\bibliography{refs} %

\ECTYPE{Supplementary material to\enspace}
\ECAUpunct{}
\ECSwitch
\ECHead{Supplementary Material}
\medskip
 \section{Proofs} 
 \subsection{Proofs of Propositions \ref{prop: convexity of g and f} and  \ref{prop: convex and differentiable structure}}
 \begin{proof}{Proof of Proposition \ref{prop: convexity of g and f}.} Fix $\delta>0$. $\textbf{(i)}$ For each $k\in\{1,\dots,K\}$, the map $(x_i^k,u_i^k)\mapsto x_i^k+\delta(\lambda^k-\mu^k u_i^k)$ is affine. Since $r\mapsto r^+=\max\{r,0\}$ is convex and nondecreasing, it follows that,
\[
f^k(x_i^k,u_i^k,\delta)=\bigl(x_i^k+\delta(\lambda^k-\mu^k u_i^k)\bigr)^+,
\]
is jointly convex in $(x_i^k,u_i^k)$. Hence $\vecf(\vecx_i,\vecu_i,\delta)$ is componentwise convex in $(\vecx_i,\vecu_i)$.

\noindent\textbf{(ii)}
Joint convexity follows from standard arguments: for each $k$ and each $s\in[0,\delta]$, the integrand,
\[
(\vecx_i,\vecu_i)\mapsto \bigl(x_i^k+s(\lambda^k-\mu^k u_i^k)\bigr)^+,
\]
is convex (as the positive-part of an affine function). Convexity is preserved under integration over $s$, multiplying by $h^k>0$, and summation over $k$, hence $g(\vecx_i,\vecu_i,\delta)$ is jointly convex in $(\vecx_i,\vecu_i)$.

To prove that $g(\vecx_i,\vecu_i,\delta)$ is continuously differentiable with respect to each component of $(\vecx_i,\vecu_i)$, consider the $k$th summand,
\[
G^k(x_i^k,u_i^k,\delta):=h^k\int_0^\delta \bigl(x_i^k+s(\lambda^k-\mu^k u_i^k)\bigr)^+\,ds,
\qquad 
g(\vecx_i,\vecu_i,\delta)=\sum_{k=1}^K G^k(x_i^k,u_i^k,\delta).
\]
First assume $x^k_i>0$ for all $k$. Using the definition of $\sigma_i^k$ in \eqref{eq:sigma_cleartime} as the clearing time,
\begin{equation}\label{eq:Gk_piecewise_with_sigma}
G^k(x_i^k,u_i^k,\delta)
=
h^k\Bigl(x_i^k(\delta\wedge \sigma_i^k)+\frac{(\delta\wedge \sigma_i^k)^2}{2}\,(\lambda^k-\mu^k u_i^k)\Bigr),
\end{equation}
and in particular the $k$th term admits the piecewise representation,
\begin{equation}\label{eq:Gk_piecewise_cases}
G^k(x_i^k,u_i^k,\delta)=
\begin{cases}
\displaystyle h^k\Bigl(x_i^k\delta+\frac{\delta^2}{2}(\lambda^k-\mu^k u_i^k)\Bigr), & \delta<\sigma_i^k,\\[0.9em]
\displaystyle h^k\,\frac{(x_i^k)^2}{2(\mu^k u_i^k-\lambda^k)}, & \delta>\sigma_i^k.
\end{cases}
\end{equation}
On each open region $\{\delta<\sigma_i^k\}$ and $\{\delta>\sigma_i^k\}$, the expression in \eqref{eq:Gk_piecewise_cases} is smooth in $(x_i^k,u_i^k)$. It remains to verify that the partial derivatives with respect to $x_i^k$ and $u_i^k$ match at boundary points where $\delta=\sigma_i^k$, which establishes differentiability at kink points. For $\delta<\sigma_i^k$,
\[
\frac{\partial G^k}{\partial x_i^k}(x_i^k,u_i^k,\delta)=h^k\delta,
\]
and for $\delta>\sigma_i^k$,
\[
\frac{\partial G^k}{\partial x_i^k}(x_i^k,u_i^k,\delta)
=h^k\frac{x_i^k}{\mu^k u_i^k-\lambda^k}
=h^k\sigma_i^k.
\]
At $\delta=\sigma_i^k$ these agree:
\[
\frac{\partial^- G^k}{\partial x_i^{k}}(x_i^k,u_i^k,\delta)
=h^k\sigma_i^k
=h^k\delta
=\frac{\partial^+ G^k}{\partial x_i^{k}}(x_i^k,u_i^k,\delta).
\]
Similarly, for $\delta<\sigma_i^k$,
\[
\frac{\partial G^k}{\partial u_i^k}(x_i^k,u_i^k,\delta)
=h^k\cdot \frac{\delta^2}{2}\cdot\frac{\partial}{\partial u_i^k}(\lambda^k-\mu^k u_i^k)
=-h^k\mu^k\frac{\delta^2}{2},
\]
and for $\delta>\sigma_i^k$,
\[
\frac{\partial G^k}{\partial u_i^k}(x_i^k,u_i^k,\delta)
=h^k\cdot \frac{(x_i^k)^2}{2}\cdot\frac{\partial}{\partial u_i^k}\bigl((\mu^k u_i^k-\lambda^k)^{-1}\bigr)
=-h^k\mu^k\frac{(x_i^k)^2}{2(\mu^k u_i^k-\lambda^k)^2}
=-h^k\mu^k\frac{(\sigma_i^k)^2}{2}.
\]
At $\delta=\sigma_i^k$ these agree:
\[
\frac{\partial^- G^k}{\partial u_i^{k}}(x_i^k,u_i^k,\delta)
=-h^k\mu^k\frac{(\sigma_i^k)^2}{2}
=-h^k\mu^k\frac{\delta^2}{2}
=\frac{\partial^+ G^k}{\partial u_i^{k}}(x_i^k,u_i^k,\delta).
\]
Therefore $G^k(\cdot,\cdot,\delta)$ is differentiable everywhere in $(x_i^k,u_i^k)$ for $x_i^k>0$. Moreover, on each side of the boundary the derivative is continuous,
and since the left/right derivatives agree on the boundary, $G^k$ is continuously differentiable. Summing over $k$ implies that
$g(\vecx_i,\vecu_i,\delta)=\sum_{k=1}^K G^k(x_i^k,u_i^k,\delta)$ is continuously differentiable with respect to each component
of $(\vecx_i,\vecu_i)$. It remains to consider the case with $x_i^k=0$. In this case,
\[
G^k(0,u_i^k,\delta)
=
\frac{h^k\delta^2}{2}
\left(\lambda^k-\mu^k u_i^k\right)^+.
\]
Hence, $G^k$ is continuously differentiable in $(x_i^k,u_i^k)$ when
$u_i^k\neq\lambda^k/\mu^k$. At the maintained-empty point $
x_i^k=0,
u_i^k=\lambda^k/\mu^k,
$
however, the one-sided derivatives with respect to $u_i^k$ are
$-\frac{h^k\mu^k\delta^2}{2}$ and $0$, respectively, 
so $G^k$ is not differentiable there. Therefore,
$g(\mathbf{x}_i,\mathbf{u}_i,\delta)=\sum_{k=1}^K G^k(x_i^k,u_i^k,\delta)$
is continuously differentiable in $(\mathbf{x}_i,\mathbf{u}_i)$ except at points where
$x_i^k=0$ and $u_i^k=\lambda^k/\mu^k$ for some $k$.

Next, we establish differentiability of $g$ with respect to $\delta$ for fixed $(\vecx_i,\vecu_i)$. For each $k$, the integrand,
\[
s\mapsto \bigl(x_i^k+s(\lambda^k-\mu^k u_i^k)\bigr)^+,
\]
is continuous in $s$. Hence, by the Fundamental Theorem of Calculus, the mapping,
\[
\delta \mapsto G^k(x_i^k,u_i^k,\delta)=h^k\int_0^\delta \bigl(x_i^k+s(\lambda^k-\mu^k u_i^k)\bigr)^+\,ds,
\]
is differentiable for $\delta\ge 0$ with derivative,
\[
\frac{\partial}{\partial \delta}G^k(x_i^k,u_i^k,\delta)
= h^k\bigl(x_i^k+\delta(\lambda^k-\mu^k u_i^k)\bigr)^+.
\]
Summing over $k$ and using linearity of differentiation yields,
\[
\frac{\partial}{\partial \delta} g(\vecx_i,\vecu_i,\delta)
=\sum_{k=1}^K h^k\bigl(x_i^k+\delta(\lambda^k-\mu^k u_i^k)\bigr)^+.
\]
Since each term on the right-hand side is continuous in $\delta$, it follows that $\partial g/\partial \delta$ is continuous, i.e.,
$g(\vecx_i,\vecu_i,\delta)$ is continuously differentiable in $\delta$.

\noindent\textbf{(iii)}
Let $(\vecx_i,\vecu_i,\delta)$ lie in an open set that does not intersect the excluded set. Then for every $k$,
$x_i^k+\delta(\lambda^k-\mu^k u_i^k)\neq 0$ throughout the set. Hence the sign of each
$x_i^k+\delta(\lambda^k-\mu^k u_i^k)$ is locally constant, so each positive-part expression is either always active
or always inactive locally. Consequently, for each $k$, $f^k(x_i^k,u_i^k,\delta)$ reduces locally to either an affine function
or the constant zero function, and thus is twice continuously differentiable on the set. Likewise, on such a set each summand $G^k$
is given locally by one of the smooth expressions in \eqref{eq:Gk_piecewise_cases}, hence is twice continuously differentiable, and therefore
so is $g(\vecx_i,\vecu_i,\delta)$. This completes the proof. \Halmos
\end{proof}
 \begin{proof}{Proof of Proposition \ref{prop: convex and differentiable structure}.}
Fix $\delta>0$, and write \(n=n(\delta)\) and
\(r=r(\delta)=T-(n-1)\delta\in(0,\delta]\). The proof is by backward induction
on the period index.

\noindent\textbf{Base case ($i=n$).}
The last-period objective is,
\[
q_n(\vecx_n,\vecu_n,\delta)=g(\vecx_n,\vecu_n,r).
\]
By Proposition \ref{prop: convexity of g and f}, the function
\((\vecx_n,\vecu_n)\mapsto q_n(\vecx_n,\vecu_n,\delta)\) is jointly convex.
Since \(\mathcal U=\{\vecu\in\mathbb R_+^K:\vecu^\top\mathbf e=1\}\) is convex,
partial minimization preserves convexity, so,
\[
v_n(\vecx_n,\delta)=\min_{\vecu_n\in\mathcal U}q_n(\vecx_n,\vecu_n,\delta),
\]
is convex in \(\vecx_n\). Moreover, for each fixed \(\vecu_n\),
\(g(\cdot,\vecu_n,r)\) is componentwise nondecreasing in \(\vecx_n\) because
increasing an initial queue cannot decrease any integrand in \eqref{eq:g_sched}.
Taking the pointwise minimum over the common feasible set \(\mathcal U\)
preserves componentwise monotonicity, so \(v_n(\cdot,\delta)\) is
componentwise nondecreasing.

\medskip
\noindent\textbf{Induction step.}
Fix \(i<n\) and assume that \(v_{i+1}(\cdot,\delta)\) is convex and
componentwise nondecreasing.
Consider
\[
q_i(\vecx_i,\vecu_i,\delta)
= g(\vecx_i,\vecu_i,\delta)+v_{i+1}\!\bigl(\vecf(\vecx_i,\vecu_i,\delta),\delta\bigr).
\]
By Proposition \ref{prop: convexity of g and f}, $(\vecx_i,\vecu_i)\mapsto g(\vecx_i,\vecu_i,\delta)$ is convex,
and the mapping $(\vecx_i,\vecu_i)\mapsto \vecf(\vecx_i,\vecu_i,\delta)$ is componentwise convex.
Because $v_{i+1}(\cdot,\delta)$ is convex and componentwise nondecreasing and $\vecf(\cdot,\cdot,\delta)$ has convex
components, the monotone composition 
$(\vecx_i,\vecu_i)\mapsto v_{i+1}(\vecf(\vecx_i,\vecu_i,\delta),\delta)$ is convex in $(\vecx_i,\vecu_i)$.
Therefore $q_i(\vecx_i,\vecu_i,\delta)$, as a sum of convex functions, is convex in $(\vecx_i,\vecu_i)$.

Since $\mathcal U$ is convex and partial minimization preserves convexity, $v_i(\vecx_i,\delta)$ is convex in $\vecx_i$. Moreover, for each fixed $\vecu_i$, the map $\vecx_i\mapsto g(\vecx_i,\vecu_i,\delta)$ is componentwise nondecreasing and $\vecx_i\mapsto \vecf(\vecx_i,\vecu_i,\delta)$ is componentwise nondecreasing. Because $v_{i+1}(\cdot,\delta)$ is componentwise nondecreasing by the induction hypothesis, $\vecx_i\mapsto q_i(\vecx_i,\vecu_i,\delta)$ is componentwise nondecreasing. Taking the pointwise minimum over the common feasible set $\mathcal U$ preserves componentwise monotonicity: if $\vecx\le \vecy$, then $q_i(\vecx,\vecu,\delta)\le q_i(\vecy,\vecu,\delta)$ for every $\vecu\in\mathcal U$, and hence $v_i(\vecx,\delta)\le v_i(\vecy,\delta)$. This completes the backward induction. \Halmos \end{proof}

\subsection{Proof of Theorem \ref{thm:K_class optimal policy}}\label{app:proofs}
We start by establishing the lower bound for the single-period problem. 
\begin{lemma}[Single-period lower bound]\label{lem:single_period_priority_lower_bound}
Fix a single period of length \(\delta>0\). Every optimal allocation for
\(\min_{\vecu\in\mathcal U}g(\vecx,\vecu,\delta)\) satisfies,
\[
u^k
\geq
\min\left(
1-\sum_{l=1}^{k-1}u^l,
\frac{x^k}{\delta\mu^k}+\frac{\lambda^k}{\mu^k}
\right),
\qquad 1\leq k\leq K-1.
\]
\end{lemma}

\proof{Proof of Lemma \ref{lem:single_period_priority_lower_bound}.}
We prove the result by contradiction. Let $\vecu$ be an optimal
allocation and suppose that the bound is violated for some class. Let
$j$ be the smallest index such that,
\[
u^j<
\min\left\{
1-\sum_{l=1}^{j-1}u^l,
\frac{x^j}{\delta\mu^j}+\frac{\lambda^j}{\mu^j}
\right\}.
\]
The second inequality implies that
$x^j+\delta(\lambda^j-\mu^j u^j)>0$, and hence class $j$ remains positive
throughout the period. Therefore, assigning additional capacity to class
$j$ reduces its holding cost over the entire period.

The first inequality implies that
$u^j<1-\sum_{l=1}^{j-1}u^l$. Since $\sum_{k=1}^K u^k=1$, there
exists some $m>j$ such that $u^m>0$. Choose $\epsilon>0$ sufficiently
small and define a feasible allocation $\widetilde{\vecu}$ by increasing
the allocation to class $j$ by $\epsilon$ and decreasing the allocation to
class $m$ by $\epsilon$, i.e.,
\[
\tilde{u}^j=u^j+\epsilon,\qquad
\tilde{u}^m=u^m-\epsilon,
\]
with all other components unchanged. Because class $j$ remains positive throughout the period, the reduction in
its holding cost is $\frac{h^j\mu^j\delta^2}{2}\epsilon$. For class $m$,
the increase in holding cost resulting from reducing its allocation by
$\epsilon$ is bounded above by
$\frac{h^m\mu^m\delta^2}{2}\epsilon$. Therefore,
\[
g(\vecx,\vecu,\delta)
-
g(\vecx,\tilde{\vecu},\delta)
\geq
\frac{\delta^2\epsilon}{2}
\left(h^j\mu^j-h^m\mu^m\right)>0,
\]
where the strict inequality follows from $j<m$ and the assumed ordering of
the $c\mu$ indices. Thus, $\tilde{\vecu}$ achieves a strictly lower cost than the optimal
allocation $\vecu$, which is a contradiction. Hence, the lower bound must
hold for every class $k\in\{1,\ldots,K-1\}$. \halmos \endproof

\begin{lemma}[Exchange improvement]\label{lem:exchange_improvement}
Consider an \(n\)-period problem and let \(\pi\) be a feasible policy whose allocations satisfy the lower bound in \eqref{eq: k class op inequality} from
period \(2\) onward. If the first-period allocation of \(\pi\) violates \eqref{eq: k class op inequality} for the smallest such class \(j\), then \(\pi\) cannot be optimal.
\end{lemma}

The proof of Lemma \ref{lem:exchange_improvement} follows the same idea as in the proof of Lemma \ref{lem:single_period_priority_lower_bound}, but requires constructing a feasible policy
over subsequent periods to account for the impact of the capacity shift on
future states. In particular, the period-2 allocation must be adjusted to
preserve either the clearing time of the higher-priority class or the
ordering of future states.

\begin{proof}{Proof of Lemma \ref{lem:exchange_improvement}.}
Let \(\pi\) be a feasible policy for the \(n\)-period problem whose allocations
from period \(2\) onward satisfy,
\begin{equation}\label{eq:inductive_hypothesis}
u_i^k
\geq
\min\left(
1-\sum_{l=1}^{k-1}u_i^l,
\frac{x_i^k}{\Delta_i(\delta)\mu^k}+\frac{\lambda^k}{\mu^k}
\right),
\quad
1\leq k\leq K-1,
\end{equation}
for each \(i\in\{2,\dots,n\}\). We show that if the first-period allocation
violates,
\begin{equation}\label{eq:inductive_want_to_show}
u_1^k
\geq
\min\left(
1-\sum_{l=1}^{k-1}u_1^l,
\frac{x_1^k}{\Delta_1(\delta)\mu^k}+\frac{\lambda^k}{\mu^k}
\right),
\quad
1\leq k\leq K-1.
\end{equation}
then \(\pi\) cannot be optimal.

Suppose, for contradiction, that \(\pi\) is optimal and violates
\eqref{eq:inductive_want_to_show} in period \(i=1\) for at least one class. Let
\(j\in\{1,\dots,K-1\}\) be the smallest index for which the bound fails, i.e.,
\begin{equation}\label{eq:violation_1}
u_1^j
<
\min\!\left\{
1-\sum_{l=1}^{j-1}u_1^l,
\frac{x_1^j}{\Delta_1(\delta)\mu^j}+\frac{\lambda^j}{\mu^j}
\right\}.
\end{equation}
Hence, \(1-\sum_{l=1}^{j-1}u_1^l>u_1^j\). Since
\(\sum_{k=1}^K u_1^k=1\), it follows that
\(\sum_{l=j+1}^K u_1^l>0\). Therefore there exists some \(m>j\) such that
\(u_1^m>0\). By assumption, \(h^m\mu^m<h^j\mu^j\).

\medskip
\noindent\emph{Construction of the policy \(\tilde\pi\).}
We proceed by constructing a policy \(\tilde\pi\) that achieves lower cost than
\(\pi\). Write \(\Delta_i=\Delta_i(\delta)\) for brevity. Choose
\(\varepsilon>0\) small enough that,
\begin{equation}\label{eq:eps_small}
0<\varepsilon\leq u_1^m,
\qquad
u_1^j+\varepsilon< 1-\sum_{l=1}^{j-1}u_1^l,
\qquad
u_1^j+\varepsilon
<
\frac{x_1^j}{\Delta_1\mu^j}+\frac{\lambda^j}{\mu^j}.
\end{equation}
Define \(\tilde\pi\) to modify the period-1 allocation only for classes \(j\)
and \(m\), by setting,
\[
\tilde u_1^j:=u_1^j+\varepsilon,\qquad
\tilde u_1^m:=u_1^m-\varepsilon.
\]
This preserves feasibility, so \(\tilde{\vecu}_1\in\mathcal U\), by
\eqref{eq:eps_small}. Let \(\vecx_2\) and \(\tilde{\vecx}_2\) denote the states at the start of period
\(2\) under \(\pi\) and \(\tilde\pi\), respectively, i.e.,
\[
\vecx_2=\vecf(\vecx_1,\vecu_1,\delta),
\qquad
\tilde{\vecx}_2=\vecf(\vecx_1,\tilde{\vecu}_1,\delta).
\]
Only classes \(j\) and \(m\) may differ between \(\vecx_2\) and
\(\tilde{\vecx}_2\). We consider two cases depending on whether class \(j\) clears
within period \(2\) under \(\pi\).

\medskip
\noindent\textbf{Case 1: Under \(\pi\), class \(j\) clears in period \(2\).} See Figure \ref{fig: induction proof case 1} for an illustration. 
Let \(\sigma\in(0,\Delta_2]\) denote the clearing time of class \(j\) under
policy \(\pi\) in period \(2\). Then, by \eqref{eq:sigma_cleartime},
\(u_2^j\) satisfies,
\[
u_2^j
=
\frac{x_2^j}{\sigma\mu^j}
+
\frac{\lambda^j}{\mu^j}.
\]
Because \(\tilde u_1^j>u_1^j\), and because class \(j\) does not clear in period
\(1\) under either policy by \eqref{eq:eps_small}, we have,
\[
\tilde x_2^j=x_2^j-\Delta_1\mu^j\varepsilon.
\]
By choosing \(\varepsilon>0\) sufficiently small, we may assume
\(\tilde x_2^j>0\). Hence policy \(\tilde\pi\) can be constructed to clear
class \(j\) at the same time \(\sigma\). This allocation is given by,
\[
\tilde u_2^j
=
\frac{\tilde x_2^j}{\sigma\mu^j}
+
\frac{\lambda^j}{\mu^j}.
\]
Therefore,
\begin{equation}\label{eq:induction_case1_period2_j_policy_difference}
u_2^j-\tilde u_2^j
=
\frac{x_2^j-\tilde x_2^j}{\sigma\mu^j}
=
\frac{\Delta_1\varepsilon}{\sigma}.
\end{equation}
We assign the freed capacity to class \(m\). Specifically, define,
\begin{equation}\label{eq:repair_case1}
\tilde u_2^j
:=
u_2^j-\frac{\Delta_1}{\sigma}\varepsilon,
\qquad
\tilde u_2^m
:=
u_2^m+\frac{\Delta_1}{\sigma}\varepsilon,
\end{equation}
and set \(\tilde u_2^k:=u_2^k\) for \(k\notin\{j,m\}\). For \(\varepsilon>0\)
small enough, feasibility \(\tilde{\vecu}_2\in\mathcal U\) and nonnegativity are
preserved.

We now compare the costs in this case. Let \(D=J(\pi)-J(\tilde\pi)\), so
that \(D>0\) means that \(\tilde\pi\) achieves a lower cost than \(\pi\).
During period \(1\), class \(j\) receives \(\varepsilon\) more capacity under
\(\tilde\pi\). Since class \(j\) does not clear in period \(1\), the reduction in
class-\(j\) cost is,
\begin{equation}\label{eq:cost_difference_class_j_period1_case1}
D_1^j
=
h^j\int_0^{\Delta_1} \mu^j\varepsilon s\,ds
=
\frac{h^j\mu^j\Delta_1^2\varepsilon}{2}.
\end{equation}
Class \(m\) receives \(\varepsilon\) less capacity under \(\tilde\pi\) in period
\(1\). Hence the class-\(m\) cost increase in period \(1\) is at most,
\begin{equation}\label{eq:cost_difference_class_m_period1_case1}
-D_1^m
\leq
h^m\int_0^{\Delta_1} \mu^m\varepsilon s\,ds
=
\frac{h^m\mu^m\Delta_1^2\varepsilon}{2},
\end{equation}
or equivalently,
\begin{equation}\label{eq:cost_difference_class_m_period1_case1_lower}
D_1^m
\geq
-\frac{h^m\mu^m\Delta_1^2\varepsilon}{2}.
\end{equation}

Next, consider period \(2\). By construction, both policies clear class \(j\) at
time \(\sigma\). At the start of period \(2\), class \(j\) under \(\tilde\pi\)
is lower by \(\Delta_1\mu^j\varepsilon\), and this difference decreases linearly
to zero at time \(\sigma\). Therefore the period-2 class-\(j\) cost reduction is,
\begin{equation}\label{eq:cost_difference_class_j_period2_case1}
D_2^j
=
h^j\frac{\Delta_1\mu^j\varepsilon\cdot\sigma}{2}
=
\frac{h^j\mu^j\Delta_1\sigma\varepsilon}{2}.
\end{equation}
For class \(m\), we have,
\[
0\leq \tilde x_2^m-x_2^m
\leq \Delta_1\mu^m\varepsilon.
\]
Since
\(\tilde u_2^m=u_2^m+\Delta_1\varepsilon/\sigma\), we have, for
\(s\in[0,\sigma]\),
\[
\tilde x_2^m(s)-x_2^m(s)
\leq
\Delta_1\mu^m\varepsilon
\left(1-\frac{s}{\sigma}\right).
\]
In particular, \(\tilde x_2^m(\sigma)\leq x_2^m(\sigma)\). Since
\(\tilde u_2^m\geq u_2^m\), this ordering is preserved on
\([\sigma,\Delta_2]\). Therefore, the class-\(m\) cost increase in period
\(2\) is at most
\[
h^m\int_0^\sigma
\Delta_1\mu^m\varepsilon\left(1-\frac{s}{\sigma}\right)ds
=
\frac{h^m\mu^m\Delta_1\sigma\varepsilon}{2},
\]
and hence
\begin{equation}\label{eq:cost_difference_class_m_period2_case1}
D_2^m
\geq
-\frac{h^m\mu^m\Delta_1\sigma\varepsilon}{2}.
\end{equation}

Finally, since class \(j\) is empty under both policies at the end of period
\(2\), class \(m\) is weakly smaller under \(\tilde\pi\), and all other classes
are unchanged, we have \(\tilde{\vecx}_3\leq \vecx_3\). From period \(3\)
onward, let \(\tilde\pi\) follow an optimal continuation from
\(\tilde{\vecx}_3\). By monotonicity of the value function, this continuation
cost is no larger than the optimal continuation cost from \(\vecx_3\), and hence
no larger than the continuation cost incurred by \(\pi\).
Combining the period-1 and period-2 cost differences with the weakly smaller
continuation cost gives,
\begin{align*}
D
&\geq
\frac{h^j\mu^j\Delta_1^2\varepsilon}{2}
-\frac{h^m\mu^m\Delta_1^2\varepsilon}{2}
+\frac{h^j\mu^j\Delta_1\sigma\varepsilon}{2}
-\frac{h^m\mu^m\Delta_1\sigma\varepsilon}{2} \\
&=
\frac{\varepsilon}{2}(\Delta_1^2+\Delta_1\sigma)
\left(h^j\mu^j-h^m\mu^m\right)
>0.
\end{align*}
The strict inequality follows from \(h^j\mu^j>h^m\mu^m\). Therefore
\(\tilde\pi\) achieves strictly lower cost than \(\pi\), contradicting the
optimality of \(\pi\).

\begin{figure}
    \FIGURE{
    \includegraphics[width=0.65\textwidth]{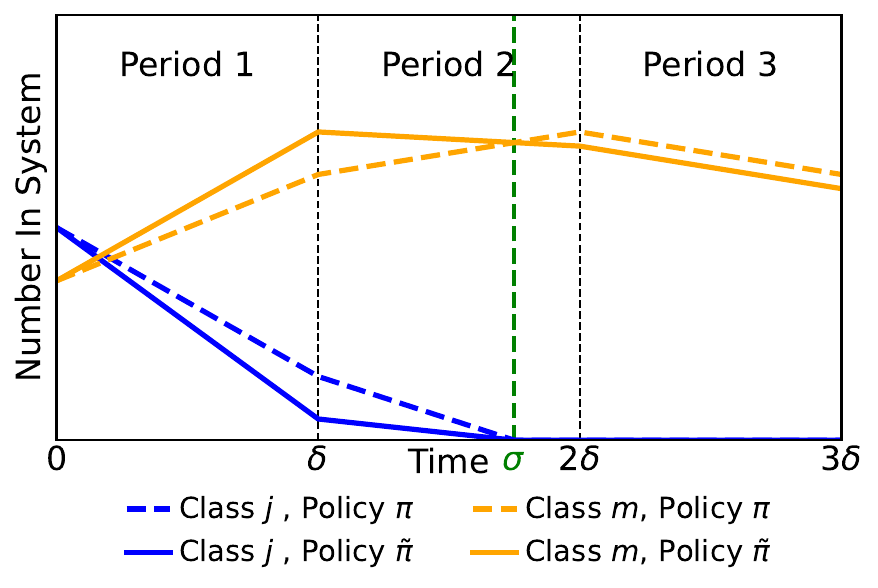}}
    {Illustrative example of policies $\pi$ and $\tPi$ in Case 1. \label{fig: induction proof case 1}}
    { Here, policies $\pi$ and $\tPi$ clear class $j$ at the same time in period $2$. At this same elapsed time, class $m$ is no larger under $\tPi$ than under $\pi$, which lets us bound the cost difference using periods $1$ and $2$ and the monotonicity of the continuation value.}
\end{figure}

\medskip
\noindent\textbf{Case 2: Under \(\pi\), class \(j\) does not clear in period \(2\).}
See Figure \ref{fig: induction proof case 2} for an illustration. By \eqref{eq:inductive_hypothesis} applied to period
\(2\), policy \(\pi\) satisfies the lower-bound condition in period \(2\). Since
class \(j\) does not clear in period \(2\), we have,
\[
u_2^j
<
\frac{x_2^j}{\Delta_2\mu^j}
+
\frac{\lambda^j}{\mu^j}.
\]
Therefore, the lower bound in \eqref{eq:inductive_hypothesis} can hold only if,
\[
u_2^j
\geq
1-\sum_{l=1}^{j-1}u_2^l.
\]
Since feasibility implies \(u_2^j\leq 1-\sum_{l=1}^{j-1}u_2^l\), it follows that,
\[
u_2^j
=
1-\sum_{l=1}^{j-1}u_2^l,
\qquad
u_2^m=0,
\]
where the second equality follows because \(m>j\).

We first show that the period-\(2\) allocation satisfies \(u_2^j>0\). Since class \(m>j\) receives positive capacity in period \(1\), and
\(j\) is the smallest class whose period-\(1\) allocation violates the
lower bound, every \(l<j\) must receive at least its period-\(1\)
clearing allocation. Hence \(x_2^l=0\) for every \(l<j\), and,
\[
\sum_{l=1}^{j-1}\frac{\lambda^l}{\mu^l}
\leq
\sum_{l=1}^{j-1}u_1^l<1.
\]
Suppose instead that \(u_2^j=0\). Since \(x_2^j>0\), applying the
induction hypothesis to class \(j\) in period \(2\) gives,
\[
\sum_{l=1}^{j-1}u_2^l=1,
\]
which is impossible if \(j=1\). Otherwise, some \(l<j\) satisfies
\(u_2^l>\lambda^l/\mu^l\). Shifting a sufficiently small amount of
period-\(2\) capacity from class \(l\) to class \(j\) keeps class \(l\)
empty throughout period \(2\), strictly reduces the period-\(2\) cost,
and weakly decreases the end-of-period-\(2\) state. Monotonicity of the
continuation value then contradicts the optimality of \(\pi\).
Therefore \(u_2^j>0\). Let,
\[
\alpha:=\frac{\Delta_1}{\Delta_2}\varepsilon.
\]
We construct \(\tilde\pi\) in period \(2\) by shifting \(\alpha\) units of
capacity from class \(j\) back to class \(m\). Specifically, define,
\begin{equation}\label{eq:repair_case2}
\tilde u_2^j:=u_2^j-\alpha,
\qquad
\tilde u_2^m:=\alpha,
\end{equation}
and set \(\tilde u_2^k:=u_2^k\) for \(k\notin\{j,m\}\). For \(\varepsilon>0\)
small enough, feasibility is preserved.

We first verify that the perturbed state at the start of period \(3\) is
componentwise no larger. By construction, \(\tilde u_i^k=u_i^k\) for all
\(k\notin\{j,m\}\) and \(i\in\{1,2\}\), so \(\tilde x_3^k=x_3^k\) for all
\(k\notin\{j,m\}\). It remains to check classes \(j\) and \(m\). Since class
\(j\) does not clear in period \(1\) under either policy by \eqref{eq:eps_small},
and does not clear in period \(2\) under \(\pi\), we have, for \(\varepsilon\)
small enough,
\begin{align*}
\tilde x_3^j
&=
x_1^j+\Delta_1(\lambda^j-\mu^j(u_1^j+\varepsilon))
+\Delta_2(\lambda^j-\mu^j(u_2^j-\alpha)) \\
&=
x_1^j+\Delta_1(\lambda^j-\mu^j u_1^j)
+\Delta_2(\lambda^j-\mu^j u_2^j)
=
x_3^j.
\end{align*}
For class \(m\), the period-1 perturbation increases the state at the start of
period \(2\) by at most \(\Delta_1\mu^m\varepsilon\), while the period-2 repair
provides exactly \(\Delta_2\mu^m\alpha=\Delta_1\mu^m\varepsilon\) additional service relative to
\(\pi\) and hence $\tilde x_3^m\leq x_3^m.$
Therefore, \(\tilde{\vecx}_3\leq\vecx_3\). From period \(3\) onward, let
\(\tilde\pi\) follow an optimal continuation from \(\tilde{\vecx}_3\). By
monotonicity of the value function, this continuation cost is no larger than the
optimal continuation cost from \(\vecx_3\), and hence no larger than the
continuation cost incurred by \(\pi\).

We now compare the costs in periods \(1\) and \(2\). Let
\(D=J(\pi)-J(\tilde\pi)\). During period \(1\), class \(j\) receives
\(\varepsilon\) more capacity under \(\tilde\pi\). Since class \(j\) does not
clear in period \(1\), the reduction in class-\(j\) cost is,
\begin{equation}\label{eq:cost_difference_class_j_period1_case2}
D_1^j
=
h^j\int_0^{\Delta_1} \mu^j\varepsilon s\,ds
=
\frac{h^j\mu^j\Delta_1^2\varepsilon}{2}.
\end{equation}
Class \(m\) receives \(\varepsilon\) less capacity under \(\tilde\pi\) in period
\(1\). Hence the class-\(m\) cost increase in period \(1\) is at most
\begin{equation}\label{eq:cost_difference_class_m_period1_case2}
-D_1^m
\leq
h^m\int_0^{\Delta_1} \mu^m\varepsilon s\,ds
=
\frac{h^m\mu^m\Delta_1^2\varepsilon}{2},
\end{equation}
or equivalently,
\[
D_1^m
\geq
-\frac{h^m\mu^m\Delta_1^2\varepsilon}{2}.
\]

Next consider period \(2\). Since the period-2 repair shifts \(\alpha\)
units of capacity from class \(j\) to class \(m\), and since the period-1
perturbation reduced class \(j\)'s period-2 initial state by
\(\Delta_1\mu^j\varepsilon\), the class-\(j\) fluid under \(\tilde\pi\) is lower
than under \(\pi\) by \(\mu^j\Delta_1\varepsilon(1-s/\Delta_2)\) at elapsed time
\(s\in[0,\Delta_2]\). Hence,
\begin{equation}\label{eq:cost_difference_class_j_period2_case2}
D_2^j
=
h^j\int_0^{\Delta_2} \mu^j\Delta_1\varepsilon\left(1-\frac{s}{\Delta_2}\right)\,ds
=
\frac{h^j\mu^j\Delta_1\Delta_2\varepsilon}{2}.
\end{equation}
For class \(m\), the period-1 perturbation increases the period-2 initial state
by \(\Delta_1\mu^m\varepsilon\), while the period-2 repair assigns
\(\alpha\) units of capacity to class \(m\). Therefore the class-\(m\) fluid
under \(\tilde\pi\) is higher than under \(\pi\) by at most
\(\mu^m\Delta_1\varepsilon(1-s/\Delta_2)\) at elapsed time \(s\in[0,\Delta_2]\). Thus,
\begin{equation}\label{eq:cost_difference_class_m_period2_case2}
-D_2^m
\leq
h^m\int_0^{\Delta_2} \mu^m\Delta_1\varepsilon\left(1-\frac{s}{\Delta_2}\right)\,ds
=
\frac{h^m\mu^m\Delta_1\Delta_2\varepsilon}{2},
\end{equation}
or equivalently,
\[
D_2^m
\geq
-\frac{h^m\mu^m\Delta_1\Delta_2\varepsilon}{2}.
\]

Combining the period-1 and period-2 cost differences with the weakly smaller continuation cost, we obtain,
\begin{align*}
D
&\geq
\frac{h^j\mu^j\Delta_1^2\varepsilon}{2}
-\frac{h^m\mu^m\Delta_1^2\varepsilon}{2}
+\frac{h^j\mu^j\Delta_1\Delta_2\varepsilon}{2}
-\frac{h^m\mu^m\Delta_1\Delta_2\varepsilon}{2} \\
&=
\frac{\varepsilon}{2}(\Delta_1^2+\Delta_1\Delta_2)
\left(h^j\mu^j-h^m\mu^m\right)
>0.
\end{align*}
The strict inequality follows from \(h^j\mu^j>h^m\mu^m\). Therefore
\(\tilde\pi\) achieves strictly lower cost than \(\pi\), contradicting the
optimality of \(\pi\).

\begin{figure}
    \FIGURE{
    \includegraphics[width=0.65\textwidth]{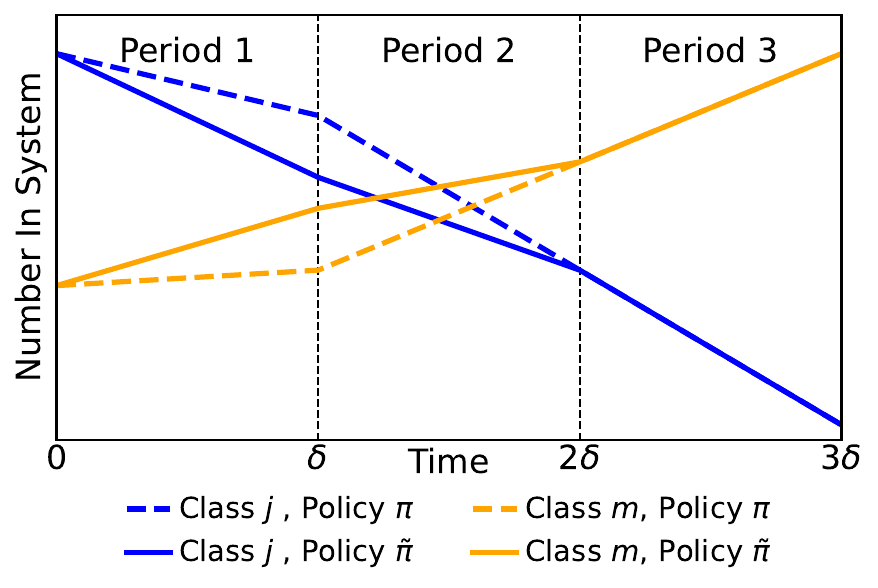}}
    {Illustrative example of policies $\pi$ and $\tPi$ in Case 2. \label{fig: induction proof case 2}}
    { Here, policy $\pi$ does not clear class $j$ in period $2$. The repair makes the states of classes $j$ and $m$ at the start of period $3$ weakly no larger under $\tPi$ than under $\pi$, so the remaining comparison uses periods $1$ and $2$ and the monotonicity of the continuation value.}
\end{figure}
The two cases exhaust all possibilities for period 2. Hence any first-period allocation that violates \eqref{eq:inductive_want_to_show} cannot be optimal. \Halmos \end{proof}

We are now ready to present the proof of Theorem \ref{thm:K_class optimal policy}.

\begin{proof}{Proof of Theorem \ref{thm:K_class optimal policy}.}
We argue by induction on the number of periods. For a single-period problem,
Lemma \ref{lem:single_period_priority_lower_bound} gives the claim. Now assume
that the result holds for every problem with \(n-1\) periods and consider an
\(n\)-period problem. The induction hypothesis applies to the continuation
problem beginning in period \(2\), so an optimal continuation satisfies the
lower bound from period \(2\) onward. If the first-period allocation violated
the lower bound, Lemma
\ref{lem:exchange_improvement} would construct a feasible policy with strictly
lower total cost, contradicting optimality. Hence the first-period allocation
also satisfies the bound. This completes the induction. \Halmos \end{proof}

\subsection{Section \ref{sec:monotone} Proofs}

\begin{proof}{Proof of Proposition \ref{prop:regular_point_smoothness}.}
We prove part (i) by backward induction in reduced coordinates.
Suppose that the
reduced objective \(q_i\) is \(C^2\) on the fixed smooth regime and that
the continuation value is convex and \(C^2\), with the continuation term
omitted when \(i=n\). Let \(\mathbf d\) be a nonzero reduced critical
direction. Under SCS, \(d^k=0\) for every
\(k\in A_i^*\setminus E_i\) and
\(\mathbf d^\intercal\mathbf e=0\), so at least two coordinates in
\(B_i^*\setminus E_i\) have nonzero components of \(\mathbf d\).

At an all-empty state, the priority structure and maintained-empty
reduction leave at most one positive free allocation, with the
convention selecting that coordinate in the underloaded case. The
capacity equality therefore makes the reduced critical cone trivial.
Thus a nonzero \(\mathbf d\) can occur only when the state is not all
empty. In that case, an empty class cannot receive excess capacity at
optimality while another queue is positive or growing. Theorem
\ref{thm:K_class optimal policy} also implies that at most one class
receiving positive capacity is nonclearing. Hence, for some \(k\) with
\(d^k\ne0\), we have
\(x_i^{k,*}>0\), and class \(k\) clears strictly before the period
endpoint.

On the fixed smooth piece, \(\vecf\) is affine in \(\vecu_i\).
Therefore, with all derivatives evaluated at the base point,
\[
\begin{aligned}
\mathbf d^\intercal
\nabla_{\vecu_i\vecu_i}^2 q_i\,\mathbf d
&=
\mathbf d^\intercal
\nabla_{\vecu_i\vecu_i}^2 g\,\mathbf d
+
\left[
    \left(
        \frac{\partial\vecf}{\partial\vecu_i}
    \right)^\intercal
    \mathbf d
\right]^\intercal
\nabla_{\vecx_{i+1}\vecx_{i+1}}^2 v_{i+1}
\left[
    \left(
        \frac{\partial\vecf}{\partial\vecu_i}
    \right)^\intercal
    \mathbf d
\right].
\end{aligned}
\]
The second term is nonnegative by convexity. The first is strictly
positive because \(g\) is separable and convex and
\[
\mathbf d^\intercal
\nabla_{\vecu_i\vecu_i}^2g\mathbf d
\geq
\frac{h^k(x_i^{k,*})^2(\mu^k)^2}
{(\mu^k u_i^{k,*}-\lambda^k)^3}(d^k)^2
>0.
\]
The reduced constraints are affine, so their Lagrangian terms
are zero when taking second derivatives with respect to $\vecu_i$. This proves SOSC.

For \(i=n\), the constant-period-count condition makes
\(r(\delta)\) affine on \(I\). Conditions 1 and 2 of Definition
\ref{def:regular_point}, together with Proposition
\ref{prop: convexity of g and f}, imply that
\(q_n(\vecx_n,\vecu_n,\delta)
=g(\vecx_n,\vecu_n,r(\delta))\) is \(C^2\) near the base point and
convex. The reduced simplex satisfies LICQ, and Condition 3 gives SCS;
the preceding argument gives SOSC. Proposition
\ref{prop:parametric_nlp_sensitivity} therefore yields a locally unique
\(C^1\) optimizer mapping, a locally fixed active
nonnegativity-constraint set, and a \(C^2\)
local value. Convexity makes this optimizer the unique global optimizer
of the reduced problem: a second optimizer would generate a segment of
optimizers contradicting local uniqueness. Condition 2 identifies the
reduced value with \(v_n\).

Now suppose that the conclusions hold for period \(i+1\). Since
\(\Delta_i(\delta)\) is affine on \(I\), Conditions 1 and 2 keep
\(g\) and \(\vecf\) on fixed \(C^2\) pieces. Together with the \(C^2\)
continuation value, this makes
$
q_i
=g(\vecx_i,\vecu_i,\Delta_i(\delta))
+v_{i+1}\bigl(
\vecf(\vecx_i,\vecu_i,\Delta_i(\delta)),\delta
\bigr)
$
\(C^2\) near the base point; it is convex by Proposition
\ref{prop: convex and differentiable structure}. The reduced simplex
satisfies LICQ, Condition 3 gives SCS, and the preceding argument gives
SOSC. Proposition \ref{prop:parametric_nlp_sensitivity}, convexity, and
Condition 2 therefore give all the conclusions in part (i) for period
\(i\). This completes the backward induction.

For part (ii), start from the fixed state \(\vecx_1\) and
recursively compose the optimizer mappings from part (i) with the state
transitions. After intersecting the finitely many local parameter
neighborhoods, this recursion gives the unique canonical \(C^1\) reduced
optimal trajectory, which agrees with the base trajectory at
\(\delta=\delta_0\).
\Halmos
\end{proof}

\begin{proof}{Proof of Theorem \ref{thm:delta_der_value_function}.}
By \eqref{eq: q_t}, for fixed \(\vecx_1\),
$
q_1(\vecx_1,\vecu,\delta)
=g(\vecx_1,\vecu,\delta)+
v_{2}(\vecf(\vecx_1,\vecu,\delta),\delta).
$
The reduced feasible set is compact and does not depend on \(\delta\). By Proposition \ref{prop:regular_point_smoothness}, \(q_1\) is differentiable at the base optimal allocation and the continuation value \(v_2\) is differentiable at the induced next state. The envelope theorem \citep{milgrom2002envelope} therefore gives, at \(\delta=\delta_0\),
\[
\pd{}{\delta}v_1(\vecx_1,\delta)
=\pd{}{\delta}q_1(\vecx_1,\vecu_1^*,\delta),
\]
for the base optimal allocation \(\vecu_1^*=\vecu_1^*(\delta_0)\). Applying the chain rule to \(q_1\), while holding the initial state \(\vecx_1\) fixed yields, \(\delta=\delta_0\),
\[
\pd{}{\delta}q_1(\vecx_1,\vecu_1^*,\delta)
=
\pd{g}{\delta}(\vecx_1,\vecu_1^*,\delta_0)
+\pd{\vecf}{\delta}(\vecx_1,\vecu_1^*,\delta_0)
\pd{v_2}{\vecx_2}(\vecx_2^*,\delta_0)
+\pd{v_2}{\delta}(\vecx_2^*,\delta_0),
\]
where \(\vecx_2^*=\vecf(\vecx_1,\vecu_1^*,\delta_0)\). Suppressing the base-point arguments gives \eqref{eq:delta_der_value_function}.  \Halmos \end{proof}

\begin{proof}{Proof of Theorem \ref{thm:delta_second_der_value_function}.}
Proposition \ref{prop:regular_point_smoothness} gives a differentiable optimal trajectory and the required second derivatives on a neighborhood of \(\delta_0\). The first-order envelope argument therefore applies throughout a sufficiently small neighborhood and gives, \[ \pd{}{\delta}v_1 = \pd{g}{\delta} + \pd{\vecf}{\delta}  \pd{v_{2}}{\vecx_{2}} + \pd{v_{2}}{\delta}.
\] Differentiating the right-hand side of this identity as a composed function of \(\delta\) gives,
\begin{equation}\label{eq: second_der_prod_rule}
    \pdd{}{\delta} v_1  = \frac{d}{d\delta} \left[ \pd{g}{\delta} \right] + \left( \frac{d}{d\delta} \left[ \pd{\vecf}{\delta} \right] \right) \pd{v_{2}}{\vecx_{2}} + 
 \pd{\vecf}{\delta}
\left(
    \frac{d}{d\delta}
    \left[\pd{v_2}{\vecx_2}\right]
\right)^\intercal + \frac{d}{d\delta} \left[  \pd{v_{2}}{\delta} \right].
\end{equation}
We expand each term separately. By the definition of the single-period cost function \eqref{eq:g_sched} we have,
\[
\frac{d}{d\delta}  \left[ \pd{g}{\delta} (\vecx_1, \vecu_1^*(\delta), \delta) \right]
= \frac{d\vecu_1^*(\delta)}{d\delta} \mpd{g}{\delta}{\vecu_1} + \pdd{g}{\delta}.
\]
Similarly, by definition of the state-evolution function \eqref{eq:f_component_interval},
\[
\frac{d}{d\delta} \left[ \pd{\vecf}{\delta}(\vecx_1, \vecu_1^*(\delta), \delta)\right]
= \frac{d\vecu_1^*(\delta)}{d\delta} \mpd{\vecf}{\delta}{\vecu_1} + \pdd{\vecf}{\delta},
\]
where we have used that \(\vecx_1\) is a given initial condition. Using \eqref{eq:fluid_opt} we have,
\begin{equation*}
    \begin{aligned}
         \frac{d}{d\delta} \left[ \pd{v_{2}}{\vecx_{2}} (\vecx_{2}(\delta),\delta)\right] 
         & = \frac{d}{d\delta} \left[ \vecf(\vecx_1, \vecu_1^*(\delta), \delta) \right] \pdd{v_{2}}{\vecx_{2}} + \mpd{v_{2}}{\vecx_{2}}{\delta}  \\
         & = \left( \frac{d\vecu_1^*(\delta)}{d\delta} \pd{\vecf}{\vecu_1} + \pd{\vecf}{\delta}\right)\pdd{v_{2}}{\vecx_{2}} + \mpd{v_{2}}{\vecx_{2}}{\delta},
    \end{aligned}
\end{equation*}
 and, \begin{equation*}
    \begin{aligned}
        \frac{d}{d\delta} \left[ \pd{v_{2}}{\delta} (\vecx_{2}(\delta), \delta) \right] 
        & = \frac{d}{d\delta} \left[ \vecf (\vecx_1, \vecu_1^*(\delta), \delta) \right] \mpd{v_{2}}{\delta}{\vecx_{2}} + \pdd{v_{2}}{\delta}\\
        & = \left( \frac{d\vecu_1^*(\delta)}{d\delta} \pd{\vecf}{\vecu_1} + \pd{\vecf}{\delta} \right) \mpd{v_{2}}{\delta}{\vecx_{2}}  + \pdd{v_{2}}{\delta},
    \end{aligned}
 \end{equation*}
where we have made repeated use of the fact that \(\vecf\) itself depends on \(\vecx_1\), \(\vecu_1^*(\delta)\), and \(\delta\). Plugging these equations into \eqref{eq: second_der_prod_rule} gives \eqref{eq:delta_second_der_value_function}.
\Halmos \end{proof}

\begin{proof}{Proof of Lemma \ref{lem:kkt_cancellation_regular_regime}.}
By Proposition \ref{prop:regular_point_smoothness}, \(\vecu_1^*(\delta)\) is differentiable along the local regular regime. For maintained-empty classes \(k\in E_1\), Condition 2 of Definition \ref{def:regular_point} gives \(u_1^{k,*}(\delta)=\lambda^k/\mu^k\), so \(d u_1^{k,*}(\delta)/d\delta=0\). Differentiating the reduced simplex equality constraint, and then embedding the derivative in the full allocation vector, gives \eqref{eq: KKT Derivative of Equality Constraint}. The active nonnegativity-constraint set of the reduced problem is fixed on this regime. If \(u_1^{k,*}>0\), then complementary slackness gives \(\gamma^k=0\). If \(u_1^{k,*}=0\), then \(u_1^{k,*}(\delta)=0\) throughout the local regime, so \(d u_1^{k,*}(\delta)/d\delta=0\) at the base point. Hence \(\gamma^k d u_1^{k,*}(\delta)/d\delta=0\) for every \(k\), which gives \eqref{eq:KKT Complentary Slack with u delta derivatives}. Combining these two identities yields
\[
     \frac{d\vecu_1^*(\delta)}{d\delta}(\beta \mathbf{e} + \vecgamma)=0.
\]
Multiplying the identity \eqref{eq:KKT Stationarity Chain Rule Applied} by \(d\vecu_1^*(\delta)/d\delta\) gives \eqref{eq: Vector FONC}. \Halmos \end{proof}

\begin{proof}{Proof of Lemma \ref{lem:kkt_differentiated_stationarity}.}
By Proposition \ref{prop:regular_point_smoothness}, the functions \(g\), \(\vecf\), \(v_2\), and \(\vecu_1^*(\delta)\) have the required derivatives on the local regular regime after maintained-empty components are fixed. Thus, the second derivatives in \eqref{eq: der_of_fonc_interior} are well defined for the free variables of the reduced problem. Differentiating the reduced stationarity identity with respect to
$\delta$ and writing the resulting identity in column form gives
\[
\frac{d\beta(\delta)}{d\delta}\mathbf{e}
+
\left(
    \frac{d\boldsymbol{\gamma}(\delta)}{d\delta}
\right)^\top
=
\left(
    \frac{d}{d\delta}
    \left[
        \frac{\partial g}{\partial\mathbf{u}_1}
    \right]
\right)^\top
+
\left(
    \frac{d}{d\delta}
    \left[
        \frac{\partial\mathbf{f}}{\partial\mathbf{u}_1}
        \frac{\partial v_2}{\partial\mathbf{x}_2}
    \right]
\right)^\top.
\]
Here, total derivatives are taken along the local optimal trajectory,
holding $\mathbf{x}_1$ fixed. All expressions below are evaluated at
$\delta=\delta_0$ after differentiation.

For the first term, the chain rule gives
\[
\left(
    \frac{d}{d\delta}
    \left[
        \frac{\partial g}{\partial\mathbf{u}_1}
    \right]
\right)^\top
=
\frac{\partial^2 g}{\partial\mathbf{u}_1^2}
\left(
    \frac{d\mathbf{u}_1^*(\delta)}{d\delta}
\right)^\top
+
\frac{\partial^2 g}
     {\partial\mathbf{u}_1\,\partial\delta}.
\]

For the second term, the product and chain rules give
\[
\begin{aligned}
\left(
    \frac{d}{d\delta}
    \left[
        \frac{\partial\mathbf{f}}{\partial\mathbf{u}_1}
        \frac{\partial v_2}{\partial\mathbf{x}_2}
    \right]
\right)^\top
&=
\frac{\partial^2\mathbf{f}}
     {\partial\delta\,\partial\mathbf{u}_1}
\frac{\partial v_2}{\partial\mathbf{x}_2}
\\
&\quad+
\frac{\partial\mathbf{f}}{\partial\mathbf{u}_1}
\left[
    \frac{\partial^2 v_2}{\partial\mathbf{x}_2^2}
    \left(
        \frac{d\mathbf{u}_1^*(\delta)}{d\delta}
        \frac{\partial\mathbf{f}}{\partial\mathbf{u}_1}
        +
        \frac{\partial\mathbf{f}}{\partial\delta}
    \right)^\top
    +
    \frac{\partial^2 v_2}
         {\partial\mathbf{x}_2\,\partial\delta}
\right].
\end{aligned}
\]
Here,
$\partial^2\mathbf{f}/(\partial\delta\,\partial\mathbf{u}_1)$
denotes the entrywise partial derivative of
$\partial\mathbf{f}/\partial\mathbf{u}_1$ with respect to $\delta$.
Terms involving second derivatives of $\mathbf{f}$ with respect to
$\mathbf{u}_1$ vanish because $\mathbf{f}$ is affine in
$\mathbf{u}_1$ on each fixed smooth piece. Combining these identities
yields \eqref{eq: der_of_fonc_interior}. \Halmos \end{proof}

\subsection{Section \ref{section: two class fluid problem} Proofs}
 
\begin{proof}{Proof of Proposition \ref{thm: 2 class scheduling v_n der wrt delta}.}

By Theorem \ref{thm:K_class optimal policy}, we know that we always assign enough capacity to clear class $1$ in the first period where this is possible. Thus, because $\delta > \tilde{\delta}$ for regions 2 and 3, we know that class $1$ always clears in the first period. By Theorem \ref{thm:K_class optimal policy}, we also know that the optimal policy for all subsequent periods is to assign just enough capacity  to class $1$ to keep it empty i.e. $u_i^{1,*} = \frac{\lambda^1}{\mu^1}$, giving the rest to class $2$. Thus, when class \(2\) remains positive after period \(1\), we have a closed form expression for \(v_2\):
\begin{equation}\label{eq: v_n-1 expression delta large enough}
    v_2 = h^2 \int_0^{\sigma_2^2 \wedge (T - \delta)} x_2^2 + t (\lambda^2 - \mu^2 u_2^{2,*})dt.
\end{equation} This expression can also be written out piecewise, getting rid of the minimum, depending on whether class \(2\) clears by \(T\):
\begin{equation}\label{eq: v_{n-1} piecewise delta large enough}
    v_2 = \begin{cases}
        h^2 \int_0^{T - \delta} x_2^2 + t (\lambda^2 - \mu^2 u_2^{2,*})dt  & \text{class \(2\) does not clear by \(T\)},\\
        h^2 \int_0^{\sigma_2^2} x_2^2 + t (\lambda^2 - \mu^2 u_2^{2,*})dt & \text{class \(2\) clears by \(T\)}.
    \end{cases}
\end{equation}

    First, we consider the regular cases where \(\sigma_1^1 \ne \delta\) and no positive queue clears exactly at a period endpoint. On these subregions, we use Theorem \ref{thm:delta_der_value_function} to evaluate this derivative. We deal with the two scenarios in which class \(2\) clears strictly before \(T\) or does not clear by \(T\) separately. Assume class \(2\) clears after period \(1\) and strictly before \(T\). Then using the definition of $\sigma_1^1$ from \eqref{eq:sigma_cleartime} we have,
        \begin{align*}
            g (\vecx_1, \vecu_1^*, \delta) 
            & = h^1\int_0^{\sigma_1^1} x_1^1 + t (\lambda^1 - \mu^1 u_1^{1,*}) dt + h^2\int_0^{\delta} x_1^2 + t (\lambda^2 - \mu^2 u_1^{2,*})dt \\
            & = h^1 \left( \frac{(x_1^1)^2}{2(\mu^1 u_1^{1,*} - \lambda^1) } \right) + h^2 \left( x_1^2 \delta + \frac{\delta^2}{2} (\lambda^2 - \mu^2 u_1^{2,*})\right), \\
            f (\vecx_1, \vecu_1^*, \delta) &= \left(
        0 ,  x_1^2 + \delta (\lambda^2 - \mu^2 u_1^{2,*})
    \right) ^\intercal , \\
    v_2(\vecx_2, \delta) 
         & = h^2 \int_0^{\sigma_2^2} x_2^2 + t (\lambda^2 - \mu^2 u_2^{2,*})dt
         = h^2 \left( \frac{(x_2^2)^2}{2(\mu^2 u_2^{2,*} - \lambda^2)} \right).
        \end{align*}
    Now, evaluating derivatives,
    \begin{align}
        \pd{g}{\delta} & = h^2 \left( x_1^2 + \delta(\lambda^2 - \mu^2 u_1^{2,*}) \right),\label{eq: dg/d delta, delta large enough}\\
        \pd{\vecf}{\delta} &  = \left( 
        0 , \lambda^2 - \mu^2 u_1^{2,*}
    \right), \label{eq: df/d delta, delta large enough} \\
    \pd{v_2}{\vecx_2} & = \left(
        0 ,  \frac{h^2 x_2^2}{\mu^2 u_2^{2,*} - \lambda^2 }
    \right)^\intercal, \label{eq: dv_{n-1}/x_{n-1}, class 2 clear} \\
        \pd{v_2}{\delta} & = 0 , \label{eq: dv_{n-1}/d delta, class 2 clear}
    \end{align}
    and so, \begin{equation*}
        \begin{aligned}
            \pd{}{\delta}v_1 
            & = h^2 \left( x_1^2 + \delta(\lambda^2 - \mu^2 u_1^{2,*}) \right) + \left( \lambda^2 - \mu^2 u_1^{2,*} \right)\left( h^2 \left( \frac{x_2^2}{\mu^2 u_2^{2,*} - \lambda^2 }\right) \right) \\
            & = h^2  x_2^2 + h^2 \left( \frac{x_2^2 \left( \lambda^2 - \mu^2 u_1^{2,*} \right)}{\mu^2 u_2^{2,*} - \lambda^2} \right) = h^2 x_2^2 \left( 1 + \frac{\lambda^2 - \mu^2 u_1^{2,*}}{\mu^2 u_2^{2,*} - \lambda^2}\right).
        \end{aligned}
    \end{equation*}
    Next, we deal with the case where class \(2\) does not clear by \(T\). Then \(g\) and \(\vecf\) remain unchanged, but by \eqref{eq: v_n-1 expression delta large enough}, \(v_2\) is given by,
    \begin{equation*}
        \begin{aligned}
            v_2 
            & = h^2 \int_0^{T - \delta} x_2^2 + t (\lambda^2 - \mu^2 u_2^{2,*})dt
            & = h^2 \left( x_2^2 (T - \delta) + \frac{(T - \delta)^2}{2} (\lambda^2 - \mu^2 u_2^{2,*}) \right).
        \end{aligned}
    \end{equation*}
    Thus, \begin{align}
        \pd{v_2}{\vecx_2} & = \left(
            0 , h^2 ( T - \delta)
        \right) ^ \intercal, \label{eq: dv_{n-1}/x_{n-1}, class 2 no clear} \\
        \pd{v_2}{\delta} & = h^2 \left(- x_2^2 - (T - \delta) (\lambda^2 - \mu^2 u_2^{2,*}) \right), \label{eq: dv_{n-1}/d delta, class 2 no clear} 
    \end{align}
    and so,
    \begin{equation*}
        \begin{aligned}
            \pd{}{\delta}v_1 
            & = h^2 \left( \left( x_1^2 + \delta(\lambda^2 - \mu^2 u_1^{2,*}) \right) + \left( \lambda^2 - \mu^2 u_1^{2,*} \right) \left(   T - \delta \right) +  \left(- x_2^2 - (T - \delta) (\lambda^2 - \mu^2 u_2^{2,*}) \right) \right) \\
            & = h^2 \left( x_2^2 - x_2^2 + (T - \delta)\left( \lambda^2 - \mu^2 u_1^{2,*} - \left( \lambda^2 - \mu^2 u_2^{2,*} \right) \right) \right) \\
            & = h^2 \left( (T - \delta)\mu^2(u_2^{2,*} - u_1^{2,*}) \right)\\
            & = h^2 \mu^2 (T - \delta) (u_1^{1,*} - \frac{\lambda^1}{\mu^1}).
        \end{aligned}
    \end{equation*}
Finally, suppose that \(\sigma_1^1=\delta\). We first show that
endpoint clearing cannot occur only at an isolated review length. By
Theorem \ref{thm:K_class optimal policy},
\begin{equation}\label{eq: endpoint_allocation_identity}
  u_1^{1,*}(\delta)
  =
  \bar u_1^1(\delta)
  =
  \frac{x_1^1}{\delta\mu^1}
  +
  \frac{\lambda^1}{\mu^1}.
\end{equation}
Let
\[
U:=1-\frac{\lambda^1}{\mu^1},
\qquad
B:=\mu^2U-\lambda^2.
\]
After class \(1\) clears, its maintenance allocation is
\(\lambda^1/\mu^1\), so class \(2\) receives \(U\). Under the
endpoint-clearing allocation in
\eqref{eq: endpoint_allocation_identity},
\[
x_2^2
=
x_1^2-\delta B+\frac{\mu^2}{\mu^1}x_1^1.
\]
Hence, when \(B>0\), the absolute time at which class \(2\) clears is
\[
\tau_2
:=
\delta+\frac{x_2^2}{B}
=
\frac{x_1^2+(\mu^2/\mu^1)x_1^1}{B},
\]
which is independent of \(\delta\). If \(B\leq0\), set
\(\tau_2:=\infty\), and define $\tau:=T\wedge\tau_2$.

Define the first-period Bellman objective along the simplex by
\[
\psi(u,\delta)
:=
q_1\!\left(
\vecx_1,(u,1-u)^\intercal,\delta
\right).
\]
At \(u=\bar u_1^1(\delta)\), a marginal increase in \(u\) makes class
\(1\) clear earlier but removes the same amount of capacity from class
\(2\). The resulting right derivative is
\[
\begin{aligned}
\partial_u^+\psi\!
\left(\bar u_1^1(\delta),\delta\right)
&=
h^2\mu^2
\left(
\int_0^\delta s\,ds
+
\int_\delta^\tau \delta\,ds
\right)
-
h^1\mu^1\int_0^\delta s\,ds\\
&=
h^2\mu^2\delta
\left(\tau-\frac{\delta}{2}\right)
-\frac{h^1\mu^1\delta^2}{2}\\
&=
\delta
\left(
h^2\mu^2\tau
-\frac{h^1\mu^1+h^2\mu^2}{2}\delta
\right).
\end{aligned}
\]
Because \(\psi(\cdot,\delta)\) is convex and every optimal allocation
satisfies \(u_1^{1,*}\geq\bar u_1^1(\delta)\), endpoint clearing is
optimal exactly when this right derivative is nonnegative. Equivalently,
\[
\delta
\leq
\delta_E
:=
\frac{2h^2\mu^2\tau}
{h^1\mu^1+h^2\mu^2}.
\]
Therefore, if endpoint clearing is optimal for any
\(\delta>\tilde{\delta}\), then it is optimal throughout $(\tilde{\delta},\delta_E]$.
In particular, it cannot occur only at an isolated review length.

Because for \(\delta<\delta_E\) the endpoint-clearing identity holds on a
neighborhood, we differentiate the allocation directly. At
\(\delta=\delta_E\), the marginal condition is an equality. The
endpoint and interior envelope derivatives agree at this cutoff, so the
same first-derivative formula extends to \(\delta_E\).We have,
\begin{align}\label{eq: u_n^1 delta = sigma1}
    u_1^{1,*}(\delta) = \frac{x_1^1}{\delta \mu^1} + \frac{\lambda^1}{\mu^1}.
\end{align}
Assume class \(2\) does not clear by \(T\). Then we have, 
\begin{equation*}
    \begin{aligned}
        g (\vecx_1, \vecu_1^*, \delta) 
        & = h^1\int_0^{\sigma_1^1} x_1^1 + t (\lambda^1 - \mu^1 u_1^{1,*}) dt + h^2\int_0^{\delta} x_1^2 + t (\lambda^2 - \mu^2 u_1^{2,*})dt \\
        & = h^1 \left( \frac{(x_1^1)^2}{2(\mu^1 u_1^{1,*} - \lambda^1) } \right) + h^2 \left( x_1^2 \delta + \frac{\delta^2}{2} (\lambda^2 - \mu^2 u_1^{2,*})\right),
    \end{aligned}
\end{equation*}
where we can use \eqref{eq: u_n^1 delta = sigma1} along with the fact that $u_1^{1,*} = 1 - u_1^{2,*}$ and $u_2^{2,*} = 1 - \frac{\lambda^1}{\mu^1}$ from Theorem \ref{thm:K_class optimal policy} to rewrite $g$ as,
\begin{equation*}
    \begin{aligned}
        g(\vecx_1 , \vecu_1^*, \delta)
        & = h^1 \left( \frac{x_1^1 \delta}{2} \right) + h^2\left( x_1^2 \delta + \frac{\delta^2}{2}\left(\lambda^2 - \mu^2 (1 - \frac{x_1^1}{\delta \mu^1} - \frac{\lambda^1}{\mu^1})\right) \right) \\
        & = h^1 \left( \frac{x_1^1 \delta}{2} \right) + h^2\left( x_1^2 \delta + \frac{\delta^2}{2}\left(\lambda^2 - \mu^2 u_2^{2,*} \right) + x_1^1\frac{\delta}{2}\left( \frac{ \mu^2}{\mu^1}\right)  \right) .
    \end{aligned}
\end{equation*} Using \eqref{eq: v_n-1 expression delta large enough}, we have that,
\begin{equation*}
        \begin{aligned}
            v_2 (\vecx_2, \delta)
            & = h^2 \int_0^{T - \delta} x_2^2 + t (\lambda^2 - \mu^2 u_2^{2,*})dt \\
            & = h^2 \left( x_2^2 (T - \delta) + \frac{(T - \delta)^2}{2} (\lambda^2 - \mu^2 u_2^{2,*}) \right),
        \end{aligned}
    \end{equation*} i.e.,
    \begin{equation}
        \begin{aligned}
            v_1 = 
            &  h^1 \left( \frac{x_1^1 \delta}{2} \right) + h^2\left( x_1^2 \delta + \frac{\delta^2}{2}\left(\lambda^2 - \mu^2 u_2^{2,*} \right) + x_1^1\frac{\delta}{2}\left( \frac{ \mu^2}{\mu^1}\right)  \right) \\
            & + h^2 \left( x_2^2 (T - \delta) + \frac{(T - \delta)^2}{2} (\lambda^2 - \mu^2 u_2^{2,*}) \right).
        \end{aligned}
    \end{equation} We also note that by \eqref{eq:f_component_+} that $x_2^2 = x_1^2  + \delta (\lambda^2 - \mu^2 u_1^{2,*}) $, and so similar to above, using \eqref{eq: u_n^1 delta = sigma1} along with the fact that $u_1^{1,*} = 1 - u_1^{2,*}$ and $u_2^{2,*} = 1 - \frac{\lambda^1}{\mu^1}$, we see that  $x_2^2 = x_1^2 + \delta \left( \lambda^2 + \frac{x_1^1}{\delta}\frac{\mu^2}{\mu^1} - \mu^2 u_2^{2,*} \right) $. We also observe that,
    \begin{equation}\label{eq: d/d delta x_n-1 delta = sigma}
        \frac{d x_2^2(\delta)}{d\delta} = \left( \lambda^2 + \frac{x_1^1}{\delta}\frac{\mu^2}{\mu^1} - \mu^2 u_2^{2,*} \right) + \delta \left( -\frac{x_1^1}{\delta^2}\frac{\mu^2}{\mu^1} \right) = \lambda^2 - \mu^2 u_2^{2,*}.
    \end{equation} 
    Thus, differentiating $v_1$ with respect to $\delta$ we see,
\begin{equation*}
    \begin{aligned}
        \pd{}{\delta}v_1 =
        & \frac{h^1 x_1^1}{2} + h^2 \left( x_1^2 + \delta(\lambda^2 - \mu^2 u_2^{2,*}) + \frac{x_1^1}{2}\frac{\mu^2}{\mu^1} \right)\\
        & + h^2 \left( \left(\lambda^2 - \mu^2 u_2^{2,*} \right) (T - \delta) - x_2^2 - (T - \delta)(\lambda^2 - \mu^2 u_2^{2,*}) \right) \\
        & = \frac{h^1 x_1^1}{2} + h^2\left( x_1^2 + \delta(\lambda^2 - \mu^2 u_2^{2,*}) + \frac{x_1^1}{2} \frac{\mu^2}{\mu^1} - x_2^2 \right)\\
        & = \frac{h^1 x_1^1}{2} + h^2\left( x_1^2 + \delta(\lambda^2 - \mu^2 u_2^{2,*}) + \frac{x_1^1}{2} \frac{\mu^2}{\mu^1} - x_1^2 - \delta \left( \lambda^2 + \frac{x_1^1}{\delta}\frac{\mu^2}{\mu^1} - \mu^2 u_2^{2,*} \right) \right)\\
        & = \frac{h^1 x_1^1}{2} + h^2\left(   \frac{x_1^1}{2} \frac{\mu^2}{\mu^1}  - \left(  x_1^1\frac{\mu^2}{\mu^1} \right) \right) = \frac{h^1 x_1^1}{2} - \frac{h^2 x_1^1}{2} \frac{\mu^2}{\mu^1}\\
          & = \frac{x_1^1(h^1 \mu^1 - h^2\mu^2)}{2\mu^1}.
    \end{aligned} 
\end{equation*}Next, we consider the case where class \(2\) clears after period \(1\) and by \(T\). \(g\) is unchanged, with \(v_2\) given by \eqref{eq: v_n-1 expression delta large enough} as,
\begin{equation*}
            v_2(\vecx_2, \delta) 
          = h^2 \int_0^{\sigma_2^2} x_2^2 + t (\lambda^2 - \mu^2 u_2^{2,*})dt
          = h^2 \left( \frac{(x_2^2)^2}{2(\mu^2 u_2^{2,*} - \lambda^2)} \right),
    \end{equation*} and so we have,
    \begin{equation*}
        \begin{aligned}
            v_1 = 
            &  h^1 \left( \frac{x_1^1 \delta}{2} \right) + h^2\left( x_1^2 \delta + \frac{\delta^2}{2}\left(\lambda^2 - \mu^2 u_2^{2,*} \right) + x_1^1\frac{\delta}{2}\left( \frac{ \mu^2}{\mu^1}\right)  \right) 
            & + h^2 \left( \frac{(x_2^2)^2}{2(\mu^2 u_2^{2,*} - \lambda^2)} \right).
        \end{aligned}
    \end{equation*}  Analogous to above, using \eqref{eq: u_n^1 delta = sigma1} along with the fact that $u_1^{1,*} = 1 - u_1^{2,*}$, $u_2^{2,*} = 1 - \frac{\lambda^1}{\mu^1}$, $x_2^2 = x_1^2 + \delta \left( \lambda^2 + \frac{x_1^1}{\delta}\frac{\mu^2}{\mu^1} - \mu^2 u_2^{2,*} \right) $ and \eqref{eq: d/d delta x_n-1 delta = sigma}, we can differentiate $v_1$ with respect to $\delta$ to get,
    \begin{equation*}
        \begin{aligned}
            \pd{}{\delta}v_1 = & \frac{h^1 x_1^1}{2} + h^2 \left( x_1^2 + \delta(\lambda^2 - \mu^2 u_2^{2,*}) + \frac{x_1^1}{2}\frac{\mu^2}{\mu^1} + \frac{x_2^2}{(\mu^2 u_2^{2,*} - \lambda^2)} ( \lambda^2 - \mu^2 u_2^{2,*} )\right)\\
            & =  \frac{h^1 x_1^1}{2} + h^2 \left( x_1^2 + \delta(\lambda^2 - \mu^2 u_2^{2,*}) + \frac{x_1^1}{2}\frac{\mu^2}{\mu^1} - \left( x_1^2 + \delta \left( \lambda^2 + \frac{x_1^1}{\delta}\frac{\mu^2}{\mu^1} - \mu^2 u_2^{2,*} \right) \right) \right) \\
            & =  \frac{h^1 x_1^1}{2} + h^2 \left( \frac{x_1^1}{2}\frac{\mu^2}{\mu^1} -  x_1^1\frac{\mu^2}{\mu^1}   \right) = \frac{h^1 x_1^1 }{2} - \frac{h^2 x_1^1}{2}\frac{\mu^2}{\mu^1}\\
            &  = \frac{x_1^1 (h^1 \mu^1 - h^2 \mu^2)}{2\mu^1}.
        \end{aligned}
    \end{equation*} The proof is complete. \Halmos \end{proof}

\begin{proof}{Proof of Theorem \ref{prop: 2 class first derivative non-negative}.}
    Using Proposition \ref{thm: 2 class scheduling v_n der wrt delta}, we first consider the case where class \(2\) does not clear by \(T\), with \(\sigma_1^1 \ne \delta\), i.e.,
     \begin{equation}
     \pd{}{\delta}v_1 = h^2 \mu^2 (T - \delta)\left(u_1^{1,*} - \frac{\lambda^1}{\mu^1}\right) .
     \end{equation} 
     Because $\delta > \tilde{\delta}$, we know that $u_1^{1,*}$ is enough to empty class $1$ before the end of the first period, i.e.,
    \begin{equation}
      x_1^1 + \delta(\lambda^1 - \mu^1 u_1^{1,*}) \leq 0 \Longrightarrow  \mu^1 u_1^{1,*} -\lambda^1    \geq  \frac{x_1^1}{\delta} \Longrightarrow u_1^{1,*} - \frac{\lambda^1}{\mu^1} \geq \frac{x_1^1}{\delta \mu^1} \geq 0,
   \end{equation}
     and, \begin{equation}
     h^2 \mu^2 (T - \delta)\left(u_1^{1,*} - \frac{\lambda^1}{\mu^1}\right) > 0. 
     \end{equation}
      We now consider the case where class \(2\) clears by \(T\), with \(\sigma_1^1 \ne \delta\), i.e.,
      \begin{equation}
          \pd{}{\delta}v_1 = h^2 x_2^2 \left(1 + \frac{\lambda^2 - \mu^2 u_1^{2,*}}{\mu^2 u_2^{2,*} - \lambda^2} \right).
      \end{equation}
     Because \(u_1^{1,*}>\lambda^1/\mu^1\), we know,
     \begin{equation}
        u_2^{2,*}=1-\frac{\lambda^1}{\mu^1}>1-u_1^{1,*}=u_1^{2,*}.
     \end{equation}
    Since class \(2\) clears after period \(1\), \(\mu^2u_2^{2,*}-\lambda^2>0\). Thus,
    \begin{equation}
        1 + \frac{\lambda^2-\mu^2u_1^{2,*}}{\mu^2u_2^{2,*}-\lambda^2}
        =
        \frac{\mu^2(u_2^{2,*}-u_1^{2,*})}{\mu^2u_2^{2,*}-\lambda^2}>0,
    \end{equation}
    i.e., \begin{equation}
          h^2 x_2^2 \left(1 + \frac{\lambda^2 - \mu^2 u_1^{2,*}}{\mu^2 u_2^{2,*} - \lambda^2} \right) > 0 .
    \end{equation}
    Finally if $\delta = \sigma_1^1$, by Proposition \ref{thm: 2 class scheduling v_n der wrt delta}, $\pd{}{\delta}v_1 = \frac{x_1^1(h^1 \mu^1 - h^2\mu^2)}{2\mu^1}$; which is a constant. By assumption $h^1\mu^1 > h^2 \mu^2$ and so this constant term is also positive. \Halmos \end{proof} %

\begin{proof}{Proof of Proposition \ref{thm: 2 class scheduling v_n second der wrt delta}.}
At an interior point of an interval on which \(\delta=\sigma_1^1\),
Proposition \ref{thm: 2 class scheduling v_n der wrt delta} gives a constant
first derivative, and hence \(\pdd{}{\delta}v_1=0\). We thus only consider regular subregions where \(\delta \ne \sigma_1^1\), \(\sigma_1^2>\delta\), \(\delta\ne\hat{\delta}\), and class \(2\) does not clear exactly at the terminal time. In such cases, we can evaluate \(\pdd{}{\delta}v_1\) using Theorem \ref{thm:delta_second_der_value_function}.
We thus first obtain each of the terms. Analogous to the first derivatives, because \(\delta > \tilde{\delta}\), we know that class \(1\) empties during period \(1\), and so we have exact expressions for \(g\), \(\vecf\) and \(v_2\), where the expression for \(v_2\) depends on whether class \(2\) clears by \(T\); see \eqref{eq: v_{n-1} piecewise delta large enough}. First we deal with the case where class \(2\) clears after period \(1\) and strictly before \(T\). From \eqref{eq: dg/d delta, delta large enough}, \eqref{eq: df/d delta, delta large enough}, \eqref{eq: dv_{n-1}/x_{n-1}, class 2 clear} and \eqref{eq: dv_{n-1}/d delta, class 2 clear} in the proof of Proposition \ref{thm: 2 class scheduling v_n der wrt delta} we know,
\[ \pd{g}{\delta} = h^2 \left( x_1^2 + \delta(\lambda^2 - \mu^2 u_1^{2,*}) \right), \]
\[ \pd{\vecf}{\delta} = \left(
        0 , \lambda^2 - \mu^2 u_1^{2,*}
    \right), \]
\[ \pd{v_2}{\vecx_2} = \left(
        0 ,   \frac{h^2 x_2^2}{\mu^2 u_2^{2,*} - \lambda^2 }
    \right)^\intercal,  \] 
\[ \pd{v_2}{\delta} = 0 . \] We then observe that, 

\begin{equation}
    \mpd{g}{\delta}{\vecu_1} = \left(
        0 , -h^2 \delta \mu^2 
    \right) ^\intercal ,
\end{equation}
\begin{equation}
    \pdd{g}{\delta} = h^2 (\lambda^2 - \mu^2 u_1^{2,*})  ,
\end{equation}
\begin{equation}
    \mpd{\vecf}{\delta}{\vecu_1} = \begin{bmatrix}
        0 & 0 \\
        0 & -\mu^2 
    \end{bmatrix} ,
\end{equation}
\begin{equation}
    \pdd{\vecf}{\delta} = \left(
        0 , 0 
    \right) ,
\end{equation}
\begin{equation}
    \pdd{v_2}{\vecx_2} = \begin{bmatrix}
        0 & 0 \\
        0 & \frac{h^2}{\mu^2 u_2^{2,*} - \lambda^2}
    \end{bmatrix},
\end{equation}
\begin{equation}
    \mpd{v_2}{\vecx_2}{\delta} = \left(
        0 , 0 
    \right) ,
\end{equation}
\begin{equation}
    \mpd{v_2}{\delta} {\vecx_2}= \left(
        0 , 0 
    \right) ^\intercal,
\end{equation} 
\begin{equation}
    \pdd{v_2}{\delta} = 0 .
\end{equation}
We thus have,
\begin{align}
    \pdd{}{\delta} v_1 
    & = -h^2 \delta \mu^2 \frac{d u_1^{2,*}(\delta)}{d\delta} + h^2 (\lambda^2 - \mu^2 u_1^{2,*}) -h^2 \mu^2 \frac{d u_1^{2,*}(\delta)}{d\delta} \left( \frac{x_2^2}{\mu^2 u_2^{2,*} - \lambda^2} \right) \nonumber \\
    & \quad  + h^2(\lambda^2 - \mu^2 u_1^{2,*} ) \left( - \frac{d u_1^{2,*}(\delta)}{d\delta} \delta \mu^2 \frac{1}{\mu^2 u_2^{2,*} - \lambda^2} + (\lambda^2 - \mu^2 u_1^{2,*})\frac{1}{\mu^2 u_2^{2,*} - \lambda^2}\right) \nonumber  \\
    & = h^2 \frac{d u_1^{2,*}(\delta)}{d\delta} \left( -\delta \mu^2 - \mu^2 \left( \frac{x_2^2}{\mu^2 u_2^{2,*} - \lambda^2}\right) - \delta \mu^2 (\lambda^2 - \mu^2 u_1^{2,*}) \frac{1}{\mu^2 u_2^{2,*} - \lambda^2}\right) \nonumber  \\
    & \quad + h^2 \left( (\lambda^2 - \mu^2 u_1^{2,*}) + \frac{(\lambda^2 - \mu^2 u_1^{2,*})^2}{\mu^2 u_2^{2,*} - \lambda^2} \right) \nonumber  \\
    & = - h^2 \frac{d u_1^{2,*}(\delta)}{d\delta} \mu^2 \left( \delta  +   \frac{x_2^2 + \delta  (\lambda^2 - \mu^2 u_1^{2,*})}{\mu^2 u_2^{2,*} - \lambda^2} \right)
     + h^2 (\lambda^2 - \mu^2 u_1^{2,*}) \left( 1 + \frac{\lambda^2 - \mu^2 u_1^{2,*}}{\mu^2 u_2^{2,*} - \lambda^2} \right), 
    \end{align}
as needed. Next, we consider the case where class \(2\) does not clear by \(T\). Although \(g\) and \(\vecf\) are the same as above, now \(v_2\) explicitly depends on \(\delta\). Once again, from the proof of Proposition \ref{thm: 2 class scheduling v_n der wrt delta} we have \eqref{eq: dv_{n-1}/x_{n-1}, class 2 no clear} and \eqref{eq: dv_{n-1}/d delta, class 2 no clear},
\[ \pd{v_2}{\vecx_2} = \left(
            0 , h^2 ( T - \delta)
        \right) ^ \intercal, \]
\[ \pd{v_2}{\delta} = h^2 \left(- x_2^2 - (T - \delta) (\lambda^2 - \mu^2 u_2^{2,*}) \right),\] 
and so we have,
\begin{equation}
    \pdd{v_2}{\vecx_2} = \begin{bmatrix}
        0 & 0 \\
        0 & 0
    \end{bmatrix},
\end{equation}

\begin{align}
\frac{\partial^2 v_2}{\partial \delta\,\partial \mathbf{x}_2}
&=
\begin{bmatrix}
0 & -h^2
\end{bmatrix},\\
\frac{\partial^2 v_2}{\partial \mathbf{x}_2\,\partial \delta}
&=
\begin{bmatrix}
0\\
-h^2
\end{bmatrix},\\
\pdd{v_2}{\delta} &= h^2 (\lambda^2 - \mu^2 u_2^{2,*}) .
\end{align}
 Once again plugging everything into \eqref{eq:delta_second_der_value_function} we have,
\begin{align}
        \pdd{}{\delta}v_1 
        & = -h^2 \delta \mu^2 \frac{d u_1^{2,*}(\delta)}{d\delta} + h^2 (\lambda^2 - \mu^2 u_1^{2,*}) \nonumber \\
        & \quad + \left(- \frac{d u_1^{2,*}(\delta)}{d\delta} \mu^2 \right) \left( h^2 (T - \delta)\right) - h^2 (\lambda^2 - \mu^2 u_1^{2,*}) \nonumber \\
        & \quad + \left( -\frac{d u_1^{2,*}(\delta)}{d\delta}  \delta \mu^2 + (\lambda^2 - \mu^2 u_1^{2,*}) \right)(-h^2) + h^2 (\lambda^2 - \mu^2 u_2^{2,*}) \nonumber  \\
        & = h^2 \frac{d u_1^{2,*}(\delta)}{d\delta} \left( -\delta \mu^2 -\mu^2 (T - \delta) + \delta \mu^2  \right) \nonumber  \\
        & \quad + h^2 \left( (\lambda^2 - \mu^2 u_1^{2,*} - (\lambda^2 - \mu^2 u_1^{2,*})) - (\lambda^2 - \mu^2 u_1^{2,*} ) + (\lambda^2 - \mu^2 u_2^{2,*})\right) \nonumber \\
        & = -h^2\mu^2\frac{d u_1^{2,*}(\delta)}{d\delta} (T - \delta) + h^2 (\mu^2 u_1^{2,*} - \mu^2 u_2^{2,*}),
    \end{align} as claimed. \Halmos \end{proof}

\begin{proof}{Proof of Theorem \ref{prop: Region 3 concavity}.}
    The general expression for $\pdd{}{\delta}v_1$ in Proposition \ref{thm: 2 class scheduling v_n second der wrt delta} depends explicitly on $d u_1^{2,*}(\delta)/d\delta$. Because $\delta > \hat{\delta}$, by definition of this region we have \(u_1^{1,*},u_1^{2,*}\in(0,1)\). Hence \(d u_1^{1,*}(\delta)/d\delta+d u_1^{2,*}(\delta)/d\delta=0\), and we can use Lemma \ref{lem:kkt_differentiated_stationarity}, specifically \eqref{eq: der_of_fonc_interior}, to solve for $d u_1^{2,*}(\delta)/d\delta$. Because we are only working in regimes with $u_1^{1,*}, u_1^{2,*} \in (0, 1)$, by \eqref{eq:KKT Complementary Slackness} we know $\vecgamma = \veczero$, for all $\delta$ values considered. Thus, \eqref{eq: der_of_fonc_interior} reduces to,
\[
\begin{aligned}
&\pdd{g}{\vecu_1}
\left(
    \frac{d\vecu_1^*(\delta)}{d\delta}
\right)^\intercal
+
\mpd{g}{\delta}{\vecu_1}
+
\mpd{\vecf}{\vecu_1}{\delta}\pd{v_2}{\vecx_2}
\\
&\quad+
\pd{\vecf}{\vecu_1}
\left[
    \pdd{v_2}{\vecx_2}
    \left(
        \frac{d\vecu_1^*(\delta)}{d\delta}
        \pd{\vecf}{\vecu_1}
        +
        \pd{\vecf}{\delta}
    \right)^\intercal
    +
    \mpd{v_2}{\delta}{\vecx_2}
\right]
=
\frac{d\beta(\delta)}{d\delta}\mathbf{e}.
\end{aligned}
\] Analogous to the proofs of Propositions \ref{thm: 2 class scheduling v_n der wrt delta} and \ref{thm: 2 class scheduling v_n second der wrt delta}, we consider the two cases of class \(2\) clearing separately. First we deal with the case where class \(2\) clears after period \(1\) and strictly before \(T\). 
    From the proof of Proposition \ref{thm: 2 class scheduling v_n second der wrt delta}, we have,
    \begin{align}
        \frac{d\vecu_1^{*}(\delta)}{d\delta} & = \left(
            \frac{d u_1^{1,*}(\delta)}{d\delta} , \frac{d u_1^{2,*}(\delta)}{d\delta}
        \right), \\
        \pdd{g}{\vecu_1} & = \begin{bmatrix}
\dfrac{
h^1 (x_1^1)^2 (\mu^1)^2
}{
\left(\mu^1 u_1^{1,*}-\lambda^1\right)^3
}
& 0
\\
0 & 0
\end{bmatrix}, \\
        \mpd{g}{\delta}{\vecu_1} & = \left(
            0 , -h^2 \delta \mu^2
        \right)^\intercal, \\
        \pd{\vecf}{\vecu_1} & = \begin{bmatrix}
            0 & 0 \\ 0 & -\mu^2 \delta
        \end{bmatrix} , \\
        \mpd{\vecf}{\vecu_1}{\delta} & = \begin{bmatrix}
            0 & 0 \\ 0 & -\mu^2
        \end{bmatrix}, \\
        \pd{v_2}{\vecx_2} & = \begin{bmatrix}
            0 & h^2 \frac{x_2^2}{\mu^2 u_2^{2,*} - \lambda^2}
        \end{bmatrix}^\intercal,\\
        \pdd{v_2}{\vecx_2} & = \begin{bmatrix}
            0 & 0 \\ 0 & h^2\frac{1}{\mu^2 u_2^{2,*} - \lambda^2}
        \end{bmatrix}, \\
        \mpd{v_2}{\delta}{\vecx_2} & = \left(
            0 , 0
        \right) ^\intercal ,
    \end{align}
     Thus,
     \small
    \begin{align}\label{eq: two class u derivative unsimplified}
            \frac{d\beta(\delta)}{d\delta}\mathbf{e}
            = \left( 
                \frac{h^1 (x_1^1)^2 (\mu^1)^2 \left(\frac{d u_1^{1,*}(\delta)}{d\delta}\right)}{(\mu^1 u_1^{1,*} - \lambda^1)^3} , -h^2 \delta \mu^2 - \frac{h^2 \mu^2}{\mu^2 u_2^{2,*} - \lambda^2}\left( x_2^2 + \delta \left( \lambda^2 - \mu^2 u_1^{2,*} - \delta \mu^2 \frac{d u_1^{2,*}(\delta)}{d\delta}\right) \right)
            \right) ^ \intercal .
    \end{align}
    \normalsize
    Because,
    \(\frac{d\beta(\delta)}{d\delta}\mathbf{e}
    = \left(
        \frac{d\beta(\delta)}{d\delta},\,
        \frac{d\beta(\delta)}{d\delta}
    \right)^\intercal\), the two components of \eqref{eq: two class u derivative unsimplified} must be equal. Thus,
    \small
    \begin{align}
            \frac{h^1 (x_1^1)^2 (\mu^1)^2 \left(\frac{d u_1^{1,*}(\delta)}{d\delta}\right)}{(\mu^1 u_1^{1,*} - \lambda^1)^3}
            & = -h^2 \delta \mu^2 - \frac{h^2 \mu^2}{\mu^2 u_2^{2,*} - \lambda^2}\left( x_2^2 + \delta \left( \lambda^2 - \mu^2 u_1^{2,*} - \delta \mu^2 \frac{d u_1^{2,*}(\delta)}{d\delta}\right) \right)\nonumber  \\
            & = -h^2 \mu^2 \left( \delta + \frac{x_2^2 + \delta(\lambda^2 - \mu^2 u_1^{2,*}) - \delta^2 \mu^2 \frac{d u_1^{2,*}(\delta)}{d\delta}}{\mu^2 u_2^{2,*} - \lambda^2}
            \right) \nonumber  \\
            & = -h^2 \mu^2 \left( \frac{\delta(\mu^2 u_2^{2,*} - \lambda^2) + x_2^2 + \delta (\lambda^2 - \mu^2 u_1^{2,*})}{\mu^2 u_2^{2,*} - \lambda^2} \right) + h^2 \left(\frac{\delta^2 (\mu^2)^2 \frac{d u_1^{2,*}(\delta)}{d\delta}}{\mu^2 u_2^{2,*} - \lambda^2}\right).
         \end{align}
    \normalsize
         Rearranging we see,
    \small
    \begin{equation}
        \frac{h^1 (x_1^1)^2 (\mu^1)^2 \left(\frac{d u_1^{1,*}(\delta)}{d\delta}\right)}{(\mu^1 u_1^{1,*} - \lambda^1)^3}  - h^2 \left(\frac{\delta^2 (\mu^2)^2 \frac{d u_1^{2,*}(\delta)}{d\delta}}{\mu^2 u_2^{2,*} - \lambda^2}\right) = -h^2 \mu^2 \left( \frac{\delta(\mu^2 u_2^{2,*} - \lambda^2) + x_2^2 + \delta (\lambda^2 - \mu^2 u_1^{2,*})}{\mu^2 u_2^{2,*} - \lambda^2} \right).
    \end{equation} 
    \normalsize Using $u_1^{2,*} = 1 - u_1^{1,*}$ and $d u_1^{2,*}(\delta)/d\delta=-d u_1^{1,*}(\delta)/d\delta$, we can simplify our expression,
    \begin{equation}
        \frac{h^1 (x_1^1)^2 (\mu^1)^2 \left(\frac{d u_1^{1,*}(\delta)}{d\delta}\right)}{(\mu^1 u_1^{1,*} - \lambda^1)^3}  + h^2 \left(\frac{\delta^2 (\mu^2)^2 \frac{d u_1^{1,*}(\delta)}{d\delta}}{\mu^2 u_2^{2,*}- \lambda^2}\right) = -h^2 \mu^2 \left( \frac{ \delta \mu^2 (u_2^{2,*} - u_1^{2,*}) + x_2^2}{\mu^2 u_2^{2,*} - \lambda^2} \right),
    \end{equation}
    where we recognize that the left hand side can be written out as,
    \[\left( \frac{h^1 (x_1^1)^2 (\mu^1)^2 (\mu^2 u_2^{2,*} - \lambda^2) + h^2 \delta^2 \mu^2 (\mu^1 u_1^{1,*} - \lambda^1)^3  }{(\mu^1 u_1^{1,*} - \lambda^1)^3 (\mu^2 u_2^{2,*} - \lambda^2)} \right) \left(\frac{d u_1^{1,*}(\delta)}{d\delta}\right) , \] i.e. our full expression can be written as,
    \begin{equation}\label{eq: 2 class; du / d delta}
        \frac{d u_1^{1,*}(\delta)}{d\delta} = \frac{-h^2 \mu^2 (\mu^1 u_1^{1,*} - \lamda^1)^3 \left( \delta \mu^2 (u_2^{2,*} - u_1^{2,*})+ x_2^2 \right)}{h^1 (x_1^1)^2 (\mu^1)^2 (\mu^2 u_2^{2,*} - \lambda^2 ) + h^2 \delta^2 (\mu^2)^2 (\mu^1 u_1^{1,*} - \lambda^1)^3} = -\frac{d u_1^{2,*}(\delta)}{d\delta}.
    \end{equation}

    The expressions above depend on the rates of changes of classes $1$ and $2$ in the first period, as well as the rate of change of class $2$ from the second period onwards. We introduce some temporary notation to simplify some of the algebra, and to make use of known information about how these rates vary. We define $A$ as the rate of change of class $2$ during the first period, $B$ the rate of change of class $2$ for all subsequent periods, and $C$ as the rate of change for class $1$ during the first period. We then have,
    \begin{equation}\label{eq: A def}
        A = \mu^2 u_1^{2,*} - \lambda^2 ,
    \end{equation}
    \begin{equation}\label{eq: B def}
        B = \mu^2 u_2^{2,*} - \lambda^2 ,
    \end{equation}
    \begin{equation}\label{eq: C def}
        C = \mu^1 u_1^{1,*} - \lambda^1 .
    \end{equation}
We can now derive inequalities for $A$, $B$ and $C$. First, as $\delta > \hat{\delta}$, we know class $1$ clears out in the first period and so $C = \mu^1 u_1^{1,*} - \lambda^1 > 0$. By Theorem \ref{thm:K_class optimal policy}, we also know $u_1^{2,*} = 1 - u_1^{1,*}$ and $u_2^{2,*} = 1 - \frac{\lambda^1}{\mu^1}$. Thus, we have,
\begin{equation}\label{eq: B - A}
    B - A = \mu^2 \left(1 - \frac{\lambda^1}{\mu^1}\right) - \mu^2\left( 1 - u_1^{1,*}\right) = \mu^2\left(u_1^{1,*}-\frac{\lambda^1}{\mu^1}\right) = \frac{\mu^2}{\mu^1} C > 0.
\end{equation} 
Because class $2$ eventually clears, we know $\mu^2 u_2^{2,*} - \lambda^2 > 0$.

We can now rewrite our expressions for \(\pdd{}{\delta}v_1\) and \(d u_1^{2,*}(\delta)/d\delta\), for the case where class \(2\) clears after period \(1\) and strictly before \(T\) as,
\begin{equation}\label{eq: pdd v_n delta A,B,C}
    \begin{aligned}
    \pdd{}{\delta}v_1
    & = -h^2 \frac{d u_1^{2,*}(\delta)}{d\delta}\mu^2 \left( \delta + \frac{x_2^2 - A\delta}{B} \right) - h^2 A \left(1 -\frac{A}{B}\right) \\
    & = -h^2 \frac{d u_1^{2,*}(\delta)}{d\delta} \mu^2 \left(\frac{x_2^2 + (B - A)\delta}{B} \right) - h^2 A \left( \frac{B - A}{B}\right),
    \end{aligned}
\end{equation}
\begin{equation}\label{eq: pd u_n^2 delta A,B,C}
    \begin{aligned}
    \frac{d u_1^{2,*}(\delta)}{d\delta}
    & = \frac{h^2 \mu^2 C^3 \left( \delta (B - A) + x_2^2 \right)}{h^1 (x_1^1)^2 (\mu^1)^2 B + h^2 \delta^2 (\mu^2)^2 C^3} \\
    & =    \frac{C^3 \left( \delta (B - A) + x_2^2 \right)}{\frac{h^1 \mu^1}{h^2 \mu^2}  (x_1^1)^2 \mu^1 B + \delta^2 \mu^2 C^3}.
    \end{aligned}
\end{equation}
Because $u_1^{1,*}, u_1^{2,*} \in (0, 1)$, we can use the cancellation identity \eqref{eq: Vector FONC} from Lemma \ref{lem:kkt_cancellation_regular_regime} to solve for $\frac{h^1 \mu^1}{h^2 \mu^2}$ and find,
\begin{equation}\label{eq: h1mu1 h2mu2 2 class interior sol}
    \frac{h^1 \mu^1}{h^2 \mu^2} = \frac{\delta C^2 (\delta B + 2 x_2^2)}{(x_1^1)^2 B},    
\end{equation} which lets us further rewrite \eqref{eq: pd u_n^2 delta A,B,C} as,
\begin{equation} \label{eq: eq: pd u_n^2 delta A,B,C h1mu1 h2mu2}
    \begin{aligned}
    \frac{d u_1^{2,*}(\delta)}{d\delta}
    & = \frac{C^3 \left( \delta (B - A) + x_2^2 \right)}{\delta C^2 (\delta B + 2 x_2^2)\mu^1 + \delta^2 \mu^2 C^3} \\
    & = \frac{C \left( \delta (B - A) + x_2^2 \right)}{\delta (\delta B + 2 x_2^2)\mu^1 + \delta^2  C \mu^2 }    . 
    \end{aligned}
\end{equation}
By \eqref{eq: B def}, \eqref{eq: C def} and \eqref{eq: B - A}, $d u_1^{2,*}(\delta)/d\delta$ is always positive. We notice that $A$ can be both non-negative, or negative. $A$ being non-negative means that the optimal allocation leads to class $2$ either decreasing, or staying the same during the first period, while $A$ negative means that class $2$ builds up during the first period. We deal with these two cases separately. First, we consider $A \geq 0$. Because $d u_1^{2,*}(\delta)/d\delta$, $B - A$ and $B$ are all positive, it follows that \eqref{eq: pdd v_n delta A,B,C} is negative. We now consider the case with $A < 0$. We thus seek to prove,
\begin{equation*}
    -h^2 \frac{d u_1^{2,*}(\delta)}{d\delta} \mu^2 \left(\frac{x_2^2 + (B - A)\delta}{B} \right) - h^2 A \left( \frac{B - A}{B}\right) < 0 .
\end{equation*} 
Rearranging, dividing both sides by $h^2$ and multiplying both sides by $B$ (which maintains the direction of the inequality as $h^2$ and $B$ are positive), we equivalently need to show,
\begin{equation*}
     -A \left( B - A \right) < \frac{d u_1^{2,*}(\delta)}{d\delta} \mu^2 \left( x_2^2 + (B - A)\delta \right) .
\end{equation*} Plugging in \eqref{eq: eq: pd u_n^2 delta A,B,C h1mu1 h2mu2}, the inequality can be written as,
\begin{equation*}
    -A \left( B - A \right) < \frac{C\mu^2 \left( x_2^2 + (B - A)\delta \right)^2}{\delta (\delta B + 2 x_2^2)\mu^1 + \delta^2 C \mu^2} =  \frac{ C\mu^2 \delta^2 \left( \frac{x_2^2}{\delta} + B - A\right)^2 }{\delta^2 \left( B+ 2\frac{x_2^2}{\delta} \right) \mu^1 + \delta^2 C \mu^2},
\end{equation*} i.e., we need to show,
\begin{equation*}
    -A \left( B - A \right) <\frac{ C\mu^2  \left( \frac{x_2^2}{\delta} + B - A\right)^2 }{\left( B+ 2\frac{x_2^2}{\delta} \right) \mu^1 + C \mu^2}.
\end{equation*} Using \eqref{eq: B - A}, we can rewrite the left hand side as,
\begin{equation*}
    C \frac{\mu^2}{\mu^1} (-A) < \frac{ C\mu^2  \left( \frac{x_2^2}{\delta} + B - A\right)^2 }{\left( B+ 2\frac{x_2^2}{\delta} \right) \mu^1 + C\mu^2},
\end{equation*} i.e. we need to show,
\begin{equation*}
    -A < \frac{ \left( \frac{x_2^2}{\delta} + B - A\right)^2 }{\left( B+ 2\frac{x_2^2}{\delta} \right)  + C \frac{\mu^2}{\mu^1}} = \frac{ \left( \frac{x_2^2}{\delta} + B - A\right)^2 }{\left( B+ 2\frac{x_2^2}{\delta} \right)  + B - A} =   \frac{ \left( \frac{x_2^2}{\delta} + B - A\right)^2 }{2B - A +  2\frac{x_2^2}{\delta} },
\end{equation*} where we once again used \eqref{eq: B - A} on the right hand side. Now, because the denominator is a positive quantity, we can multiply through while maintaining the direction of the inequality to see that we need to prove,
\begin{equation*}
     \left( \frac{x_2^2}{\delta} + B - A\right)^2 + A \left( 2B - A + 2 \frac{x_2^2}{\delta} \right) > 0,
\end{equation*} which is equivalent to,
\begin{equation*}
    \left( \frac{x_2^2}{\delta} \right)^2 + B^2 + A^2 + 2 \frac{x_2^2}{\delta} B - 2 \frac{x_2^2}{\delta} A - 2AB + 2AB - A^2 + 2 \frac{x_2^2}{\delta} A > 0,
\end{equation*} which is the same as showing,
\begin{equation*}
     \left( \frac{x_2^2}{\delta} \right)^2 + B^2 +  2 \frac{x_2^2}{\delta} B  = \left( B + \frac{x_2^2}{\delta} \right)^2 > 0,
\end{equation*} which clearly holds as both quantities are non-zero by assumption.

We can now deal with the case where class \(2\) does not clear by \(T\). We once again use \eqref{eq: der_of_fonc_interior} from Lemma \ref{lem:kkt_differentiated_stationarity} to solve for \(d u_1^{2,*}(\delta)/d\delta\). We have then that,
\begin{equation*}
    \frac{d u_1^{2,*}(\delta)}{d\delta} = \frac{h^2 \mu^2 (T - \delta) (\mu^1 u_1^{1,*} - \lambda^1)^3}{h^1 (x_1^1)^2 (\mu^1)^2}.
\end{equation*} As $T > \delta$ and class $1$ clears in the first period i.e. $\mu^1 u_1^{1,*} - \lambda^1 > 0$, we see that $d u_1^{2,*}(\delta)/d\delta>0$. It follows then that,
\begin{equation*}
    \pdd{}{\delta}v_1 = -h^2\mu^2 \left( (T - \delta) \frac{d u_1^{2,*}(\delta)}{d\delta} + \left( u_1^{1,*} - \frac{\lambda^1}{\mu^1} \right) \right)<0.
\end{equation*} The proof is complete. \Halmos\end{proof}

\subsection{Proofs for Theorem \ref{prop:K_class_eventual_shape}}

If \(\delta\) is large enough such that classes \(1,\dots,K-1\) empty in the first period, by Theorem \ref{thm:K_class optimal policy} we know that for all future periods \(i\geq2\),
\[
    u_i^{k,*}=\frac{\lambda^k}{\mu^k},\quad k\in\{1,\dots,K-1\},
    \qquad
    u_i^{K,*}=1-\sum_{k=1}^{K-1}\frac{\lambda^k}{\mu^k}.
\]
Assuming that class \(K\) does not clear in the first period, we have,
\begin{equation}\label{eq: g K class}
    g(\vecx_1, \vecu_1^*, \delta) = \sum_{k = 1}^{K-1} \left( h^k\int_0^{\sigma_{1}^{k}}x_1^k + s(\lambda^k - \mu^k u_1^{k,*})ds \right) + h^K \int_0^\delta x_1^{K} + s(\lambda^K - \mu^K u_1^{K,*})ds,
\end{equation}
\begin{equation}\label{eq: f K class}
    \vecf(\vecx_1, \vecu_1^*, \delta) = \left(
        0 , \dots ,0, x_1^K + \delta (\lambda^K - \mu^K u_1^{K,*})
   \right) ^\intercal ,
\end{equation}
with,
\begin{equation}\label{eq v n - 1 K class}
    v_2(\vecx_2, \delta) = 
    \begin{cases}
        h^K \int_0^{T - \delta} x_2^K + s(\lambda^K - \mu^K u_2^{K,*})ds, & \text{class \(K\) does not clear by \(T\)},\\
        h^K \int_0^{\sigma_2^K} x_2^K + s(\lambda^K - \mu^K u_2^{K,*})ds, & \text{class \(K\) clears by \(T\)}.
    \end{cases}
\end{equation}

The expressions for \(g\), \(\vecf\), and \(v_2\) are analogous to the two-class case with \(\delta>\tilde{\delta}\). Indeed, for \(K=2\), the condition that classes \(1,\dots,K-1\) clear during the first period is exactly the condition that class \(1\) clears during the first period. The next result presents the corresponding first- and second-derivative formulas in this multiclass regime. %

\begin{proposition}\label{thm: K class scheduling first der}
For the \(K\)-class scheduling problem, assume that
$\sigma_1^k<\delta$ for $k \in \{1,\dots,K-1\}$ and 
$\sigma_1^K>\delta$. We have,
\begin{equation}
    \pd{}{\delta}v_1 = 
\begin{cases}
    h^K \mu^K (T - \delta) (u_2^{K,*} - u_1^{K,*}), & \sigma_2^K>T-\delta,\\
        h^K x_2^K \left(1 + \frac{\lambda^K - \mu^K u_1^{K,*}}{\mu^K u_2^{K,*} - \lambda^K }\right), & \sigma_2^K\le T-\delta.
\end{cases}    
\end{equation}
If, in addition, \(\sigma_2^K\ne T-\delta\), then
\begin{equation}
  \begin{aligned}
    \pdd{}{\delta}v_1
    &=
    \begin{cases}
        h^K \mu^K \left( (\delta - T)\frac{d u_1^{K,*}(\delta)}{d\delta}
        - \left(u_2^{K,*} - u_1^{K,*} \right) \right),
        & \sigma_2^K>T-\delta, \\
        \begin{aligned}
        h^K\Bigg[&
        (\lambda^K - \mu^K u_1^{K,*})
        \left( 1 + \dfrac{\lambda^K - \mu^K u_1^{K,*}}{\mu^K u_2^{K,*} - \lambda^K} \right) \\
        &-\frac{d u_1^{K,*}(\delta)}{d\delta} \mu^K
        \left( \delta+\dfrac{x_2^K+\delta(\lambda^K - \mu^K u_1^{K,*})}
        {\mu^K u_2^{K,*} - \lambda^K} \right)
        \Bigg],
        \end{aligned}
        & \sigma_2^K<T-\delta.
    \end{cases}
  \end{aligned}
\end{equation}
\end{proposition}
The condition \(\sigma_1^k<\delta\), \(k\in\{1,\dots,K-1\}\), implies that the first \(K-1\) classes clear before the end of the first period, while \(\sigma_1^K>\delta\) implies that class \(K\) remains positive at the start of period \(2\). %

\begin{proof}{Proof of Proposition \ref{thm: K class scheduling first der}.}
    On regular subregions, we use Theorem \ref{thm:delta_der_value_function} i.e.
    \[
        \pd{}{\delta}v_1
        =
        \pd{g}{\delta}
        +
        \pd{\vecf}{\delta}\pd{v_2}{\vecx_2}
        +
        \pd{v_2}{\delta}
    \]
    to solve for \(\pd{}{\delta}v_1\). The terminal-clearing boundary for class \(K\) follows from the same explicit expression for \(v_2\). From \eqref{eq: g K class} and \eqref{eq: f K class} we have,
    \begin{equation*}
        \pd{g}{\delta} = h^K \left( x_1^K + \delta (\lambda^K - \mu^K u_1^{K,*})\right),
    \end{equation*}
    \begin{equation*}
        \pd{\vecf}{\delta}=
        \begin{bmatrix}
        0 & 0 & \dots & \lambda^K-\mu^K u_1^{K,*}
        \end{bmatrix}.
    \end{equation*}
    If \(\sigma_2^K\le T-\delta\), i.e., class \(K\) clears after period \(1\) and by the terminal time, we have
    \begin{equation*}
        \pd{v_2}{\vecx_2} = \begin{bmatrix}
            0 & \quad 0 & \quad \dots & \quad h^K \left( \frac{x_2^K}{\mu^K u_2^{K,*} - \lambda^K} \right)
        \end{bmatrix}^\intercal,
    \end{equation*}
    \begin{equation*}
        \pd{v_2}{\delta} =  0,
    \end{equation*} and so using \eqref{eq:delta_der_value_function} we have,
    \begin{equation*}
        \begin{aligned}
            \pd{}{\delta}v_1 
            & =h^K \left( x_1^K + \delta (\lambda^K - \mu^K u_1^{K,*})\right) + h^K(\lambda^K - \mu^K u_1^{K,*})\left(\frac{x_2^K}{\mu^K u_2^{K,*} - \lambda^K}  \right) \\
            & = h^K x_2^K + h^K(\lambda^K - \mu^K u_1^{K,*})\left(\frac{x_2^K}{\mu^K u_2^{K,*} - \lambda^K}  \right) \\
            & = h^K x_2^K \left( 1 + \frac{\lambda^K - \mu^K u_1^{K,*}}{\mu^K u_2^{K,*} - \lambda^K} \right).
        \end{aligned}
    \end{equation*}
    If instead \(\sigma_2^K>T-\delta\), i.e., class \(K\) does not clear by the terminal time, we have
    \begin{equation*}
        \pd{v_2}{\vecx_2} = \begin{bmatrix}
            0 & \quad 0 & \quad \dots & \quad h^K (T - \delta)
        \end{bmatrix} ^\intercal,
    \end{equation*}
    \begin{equation*}
        \pd{v_2}{\delta} = h^K \left( -x_2^K - (T - \delta)(\lambda^K - \mu^K u_2^{K,*}) \right),
    \end{equation*}
    and so 
    \begin{equation*}
        \begin{aligned}
            \pd{}{\delta}v_1 
            & = h^K \left(  x_1^K + \delta (\lambda^K - \mu^K u_1^{K,*}) + (\lambda^K - \mu^K u_1^{K,*})(T - \delta)   -x_2^K - (T - \delta)(\lambda^K - \mu^K u_2^{K,*})  \right) \\
            & = h^K \left( x_2^K - x_2^{K}+ (T - \delta)\left( \lambda^K - \mu^K u_1^{K,*} - \lambda^K + \mu^K u_2^{K,*} \right) \right)\\
            & = h^K \mu^K(T - \delta)( u_2^{K,*} -u_1^{K,*}),
        \end{aligned} 
    \end{equation*}
    This proves the first-derivative formula. We now differentiate these expressions on the strict clearing regimes \(\sigma_2^K>T-\delta\) and \(\sigma_2^K<T-\delta\) to obtain the displayed second derivatives. In this regime, for periods \(i\geq 2\),
    \[
        u_i^{K,*}
        =
        1-\sum_{k=1}^{K-1}\frac{\lambda^k}{\mu^k},
    \]
    so \(u_2^{K,*}\) is constant with respect to \(\delta\).

    First consider the case \(\sigma_2^K>T-\delta\), where class \(K\) does not clear by the terminal time. Differentiating
    \[
        \pd{}{\delta}v_1
        =
        h^K\mu^K(T-\delta)(u_2^{K,*}-u_1^{K,*})
    \]
    gives
    \begin{align*}
        \pdd{}{\delta}v_1
        &=
        h^K\mu^K\left(
            -(u_2^{K,*}-u_1^{K,*})
            -(T-\delta)\frac{d u_1^{K,*}(\delta)}{d\delta}
        \right)\\
        &=
        h^K\mu^K\left(
            (\delta-T)\frac{d u_1^{K,*}(\delta)}{d\delta}
            -(u_2^{K,*}-u_1^{K,*})
        \right).
    \end{align*}

    Next consider the case \(\sigma_2^K<T-\delta\), where class \(K\) clears after period \(1\) and strictly before the terminal time. Let
    \[
        A:=\lambda^K-\mu^K u_1^{K,*},
        \qquad
        B:=\mu^K u_2^{K,*}-\lambda^K .
    \]
    Then \(B\) is constant in \(\delta\),
    \(dA(\delta)/d\delta=-\mu^K d u_1^{K,*}(\delta)/d\delta\), and,
    \[
        x_2^K=x_1^K+\delta A
        \quad\Longrightarrow\quad
        \frac{d x_2^K(\delta)}{d\delta}
        =
        A-\delta\mu^K\frac{d u_1^{K,*}(\delta)}{d\delta}.
    \]
    Differentiating,
    \[
        \pd{}{\delta}v_1
        =
        h^K x_2^K\left(1+\frac{A}{B}\right),
    \]
    yields,
    \begin{align*}
        \pdd{}{\delta}v_1
        &=
        h^K
        \left[
            \left(A-\delta\mu^K\frac{d u_1^{K,*}(\delta)}{d\delta}\right)
            \left(1+\frac{A}{B}\right)
            -
            x_2^K\frac{\mu^K\frac{d u_1^{K,*}(\delta)}{d\delta}}{B}
        \right]\\
        &=
        h^K
        \left[
            A\left(1+\frac{A}{B}\right)
            -
            \frac{d u_1^{K,*}(\delta)}{d\delta}\mu^K
            \left(
                \delta+\frac{x_2^K+\delta A}{B}
            \right)
        \right]\\
        &=
        h^K
        \left[
            (\lambda^K-\mu^K u_1^{K,*})
            \left(
                1+\frac{\lambda^K-\mu^K u_1^{K,*}}
                {\mu^K u_2^{K,*}-\lambda^K}
            \right)
        \right.
        \\
        &\qquad\left.
            -
            \frac{d u_1^{K,*}(\delta)}{d\delta}\mu^K
            \left(
                \delta+
                \frac{x_2^K+\delta(\lambda^K-\mu^K u_1^{K,*})}
                {\mu^K u_2^{K,*}-\lambda^K}
            \right)
        \right],
    \end{align*}
    which is the second part of the claimed expression. \Halmos \end{proof}

We are now ready to present the proof of Theorem \ref{prop:K_class_eventual_shape}.

\proof{Proof of Theorem \ref{prop:K_class_eventual_shape}.}
  Since classes \(1,\ldots,K-1\) clear strictly before the end of period
  \(1\),
  \[
  D
  :=
  u_2^{K,*}-u_1^{K,*}
  =
  \sum_{k=1}^{K-1}
  \left(
  u_1^{k,*}-\frac{\lambda^k}{\mu^k}
  \right)
  >0.
  \]
  The first expression in Proposition
  \ref{thm: K class scheduling first der} is therefore positive. For the
  second expression,
  \[
  1+
  \frac{\lambda^K-\mu^K u_1^{K,*}}
  {\mu^K u_2^{K,*}-\lambda^K}
  =
  \frac{\mu^K D}
  {\mu^K u_2^{K,*}-\lambda^K}
  >0.
  \]
  This proves part (i).

  If \(u_1^{K,*}=0\), regularity gives
  \(d u_1^{K,*}(\delta)/d\delta=0\). Substitution into Proposition
  \ref{thm: K class scheduling first der} gives the signs in parts (ii)
  and (iii) as in the proof of Proposition
  \ref{prop: region 2 second der}. Now suppose \(u_1^{K,*}>0\), and set
  \[
  y_k:=\mu^k u_1^{k,*}-\lambda^k,
  \qquad
  D=\sum_{k=1}^{K-1}\frac{y_k}{\mu^k}.
  \]
  By the Cauchy--Schwarz inequality,
  \[
  \min_{\substack{y_k>0\\
  \sum_{k=1}^{K-1}y_k/\mu^k=D}}
  \sum_{k=1}^{K-1}
  \frac{h^k(x_1^k)^2}{2y_k}
  =
  \frac{a}{D},
  \qquad
  a:=
  \left(
  \sum_{k=1}^{K-1}
  x_1^k\sqrt{\frac{h^k}{2\mu^k}}
  \right)^2.
  \]
  Thus, after optimizing the allocation among classes
  \(1,\ldots,K-1\), the remaining problem is the same single-dimensional
  problem as in the proof of Theorem
  \ref{prop: Region 3 concavity}, with \(a/D\) replacing the class-\(1\)
  cost and \(u_1^{K,*}=u_2^{K,*}-D\). The same argument gives
  \(\pdd{}{\delta}v_1<0\), whether or not class \(K\) clears by \(T\).
  This proves parts (ii) and (iii).
\Halmos
 \endproof 

\section{Dynamic Cut-off Points in the Single-stage Cost}\label{ap:single-stage-example}
Consider the single-stage problem at the base review length \(\delta_0\), where there is no continuation value. The dynamic marginal value reduces to
\[
M^k(u^k):=-\frac{\partial g}{\partial u^k}
=h^k\mu^k\frac{(\delta_0\wedge\sigma^k(u^k))^2}{2}.
\]
If the positive-allocation set is the priority prefix
\(B=\{1,\ldots,m\}\), then SCS is vacuous when \(m=K\). When \(m<K\),
if the lowest-priority class receiving positive capacity, class \(m\),
does not clear during the period, then
\[
M^m=h^m\mu^m\frac{\delta_0^2}{2}
>
h^{m+1}\mu^{m+1}\frac{\delta_0^2}{2}
=M^{m+1}(0),
\]
and SCS follows. If class \(m\) clears, then SCS holds provided
\[
h^m\mu^m\bigl(\sigma^m(u^{m,*})\bigr)^2
>
h^{m+1}\mu^{m+1}\delta_0^2.
\]
Equivalently, with \(c_k=h^k\mu^k\),
\[
\rho_m:=1-\sum_{r=1}^m\frac{\lambda^r}{\mu^r},
\qquad
\Theta_m:=\sum_{r=1}^m\frac{x^r}{\mu^r}\sqrt{\frac{c_r}{c_m}},
\]
the clearing case satisfies SCS exactly when
\[
\frac{\Theta_m}{\rho_m}
>
\delta_0\sqrt{\frac{c_{m+1}}{c_m}},
\]
with the clearing regime itself requiring
\(\Theta_m/\rho_m<\delta_0\). Thus strict priority proves SCS in the
nonclearing marginal case. In the clearing marginal case, SCS can fail
only at a cutoff tie between class \(m\) and the first zero-allocation
class.

\section{A Four-class Numerical Example}\label{ap:4 class}
Figure \ref{fig: four class example} plots \(v_1\) versus \(\delta\) for a four class system. We once again observe non-differentiability in the value function only at divisors of \(\tilde{\delta}^k\) (see \eqref{eq: tilde delta general def}), with the function being monotone after \(\delta\) is large enough such that it is optimal for \(K - 1 = 3\) classes to clear out in the first period (i.e. for \(\delta > \tilde{\delta}^3\)).
\begin{figure}
    \FIGURE{
    \includegraphics[width=0.7\textwidth]{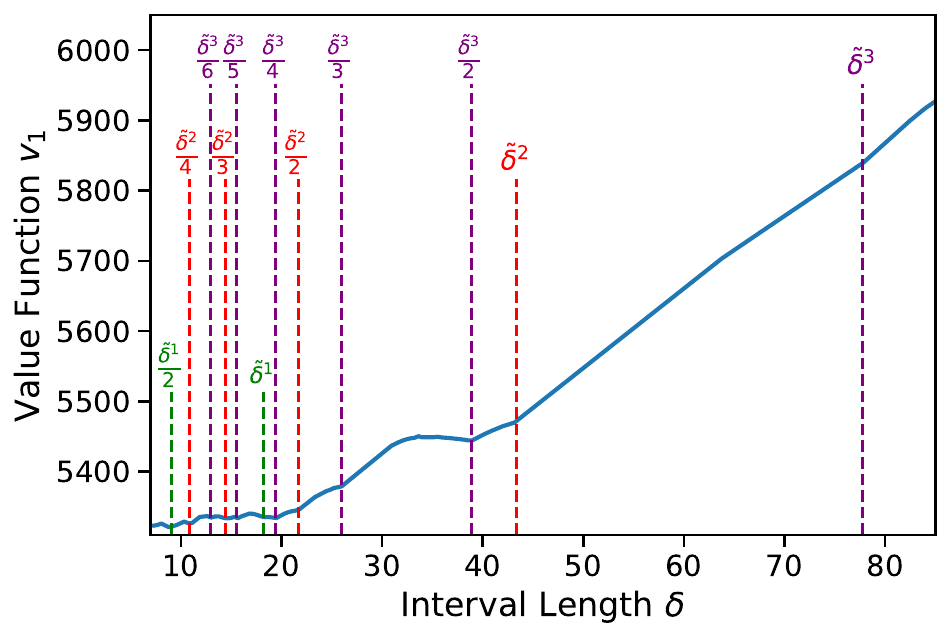}}
    {Value function \(v_1\) versus interval length \(\delta\) for the \(4\) class fluid scheduling problem. \label{fig: four class example}}
    { The green, red and purple dotted lines show the first few divisors of \(\tilde{\delta}^1\), \(\tilde{\delta}^2\) and \(\tilde{\delta}^3\) respectively (see \eqref{eq: tilde delta general def}). \(v_1\) is not differentiable at these countably infinite values. The system parameters are \((\lambda^1, \lamda^2, \lambda^3, \lamda^4) = (0.45, 0.25, 0.12, 0.1)\), \((\mu^1, \mu^2, \mu^3, \mu^4) = (1, 1, 1, 1)\), \((h^1, h^2, h^3, h^4) = (7, 6, 5, 4)\), \((x_1^1, x_1^2, x_1^3,x_1^4 ) = (4, 3, 1, 1)\), \(T = 100\).}
\end{figure}

\section{Additional Numerical Results}\label{ap:numerics}
\subsection{Additional Fluid Model Experiments}
We expand upon the numerical analysis of the two-class fluid scheduling problem from Section \ref{section: fluid model experiments}. Here, we study systems where $\mu^1 \ne \mu^2$. The initial conditions are set to $(8, 4)$ for all systems, with the horizon length $T = 100$, chosen to be sufficiently large enough to at least empty systems where the two classes have the lowest ratio between their $c\mu$ indexes. Figure \ref{fig: fluid numerics 2} plots the value function as a function of $\delta$ for $\rho\in \{0.7, 0.9\}$, $r\in \{0.62, 1.33, 2.86\}$ and $\nu \in \{2, 5, 20\}$.

We observe similar behavior as in the $\mu^1 = \mu^2$ systems of Section \ref{section: fluid model experiments}. Systems with the lowest $\nu$ value generally have the smallest magnitude of relative cost increase with increasing $\delta$, though the convex structure of $\nu = 20$ systems can once again cause their relative cost increase to decrease below that of $\nu = 2$ systems, for sufficiently small $\delta$. As in Section \ref{section: fluid model experiments} and Theorem \ref{prop: Region 3 concavity}, this convex structure only appears when the horizon is long enough for class $2$ to empty. This is most apparent when $r = 2.86$ and $\rho = 0.7$. We also observe that in high utilization systems dominated by class $1$ numbers (i.e. $r \in \{1.33, 2.86\}$, $\nu = 20$ suffers lesser cost increases with decreased control, as compared to $\nu = 5$ systems, analogous to Section \ref{section: fluid model experiments}. All systems also exhibit similar curvature as seen for the $\mu^1 = \mu^2$ systems of Section \ref{section: fluid model experiments}; $\nu = 2$ systems are initially linear before becoming concave, $\nu = 5$ systems are generally always concave while $\nu = 20$ systems are mostly concave, with a convex region appearing for systems where class $2$ is able to clear out in the horizon.

\begin{figure}
    \FIGURE{
    \includegraphics[width=\textwidth]{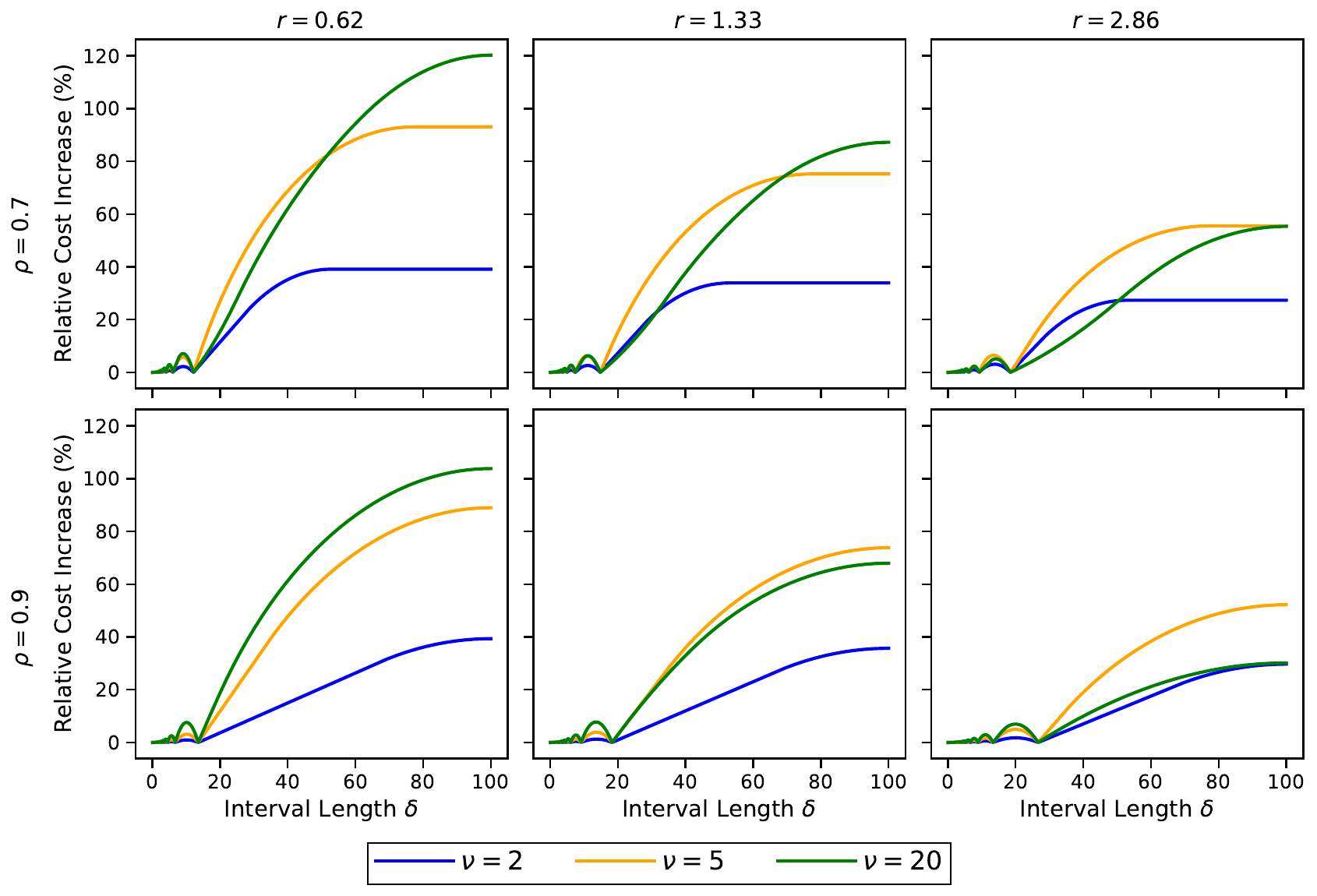}}
    {Additional Fluid System Experiments for Relative cost increase compared to continuous-time control versus $\delta$, for various cost ratios $\nu$, load ratios $r$, and utilizations $\rho$. \label{fig: fluid numerics 2}}
    { Constant $\rho$ across columns, constant Load Ratio across rows. Initial conditions fixed to \(\vecx_1 = (8,4)\) for all systems with \(( \mu^1, \mu^2) = (0.9 , 1.2)\) for all systems. The system parameters that vary are as follows: For the $\rho = 0.7$ systems (from left to right) , $(\lamda^1, \lamda^2) = (0.24, 0.52), (0.36, 0.36), (0.47, 0.22)$. For the $\rho = 0.9$ systems, (from left to right), $(\lamda^1, \lamda^2) = (0.31, 0.67), (0.46, 0.46), (0.60, 0.28)$.}
\end{figure}

\subsection{Additional Stochastic System Experiments}
We verify the validity of the results of our stochastic experiments of Section \ref{sec: stochastic system experiments} by rerunning each system, first with an increased number of samples and next with an increased maximum number in system. Figure \ref{fig: stoch sys_1 bigger samples} plots the first system of Section \ref{sec: stochastic system experiments} (see Figure \ref{fig: stoch sys_1}) with an increased number of samples. Figure  \ref{fig: stoch sys_1 bigger space} plots the same system with an increased maximum number in system. Similarly, Figures \ref{fig: stoch sys_2 bigger samples} and \ref{fig: stoch sys_2 bigger space} plot the second system of Section \ref{sec: stochastic system experiments} (see Figure \ref{fig: stoch sys_2}) with an increased number of samples and maximum number in system respectively. 
These figures appear almost identical to those of Section \ref{sec: stochastic system experiments}, with minimal changes in the scaled value functions as we vary interval length $\delta$. This suggests that the 
choices of maximum number in system for each system in Section \ref{sec: stochastic system experiments} are large enough that almost no arrivals are rejected, leading to minimal effect on system dynamics. Similarly, the sample sizes chosen for each system in Section \ref{sec: stochastic system experiments} are large enough that the choices of optimal policies for each systems is unaffected by sample numbers.

\begin{figure}
    \FIGURE{
    \includegraphics[width=\textwidth]{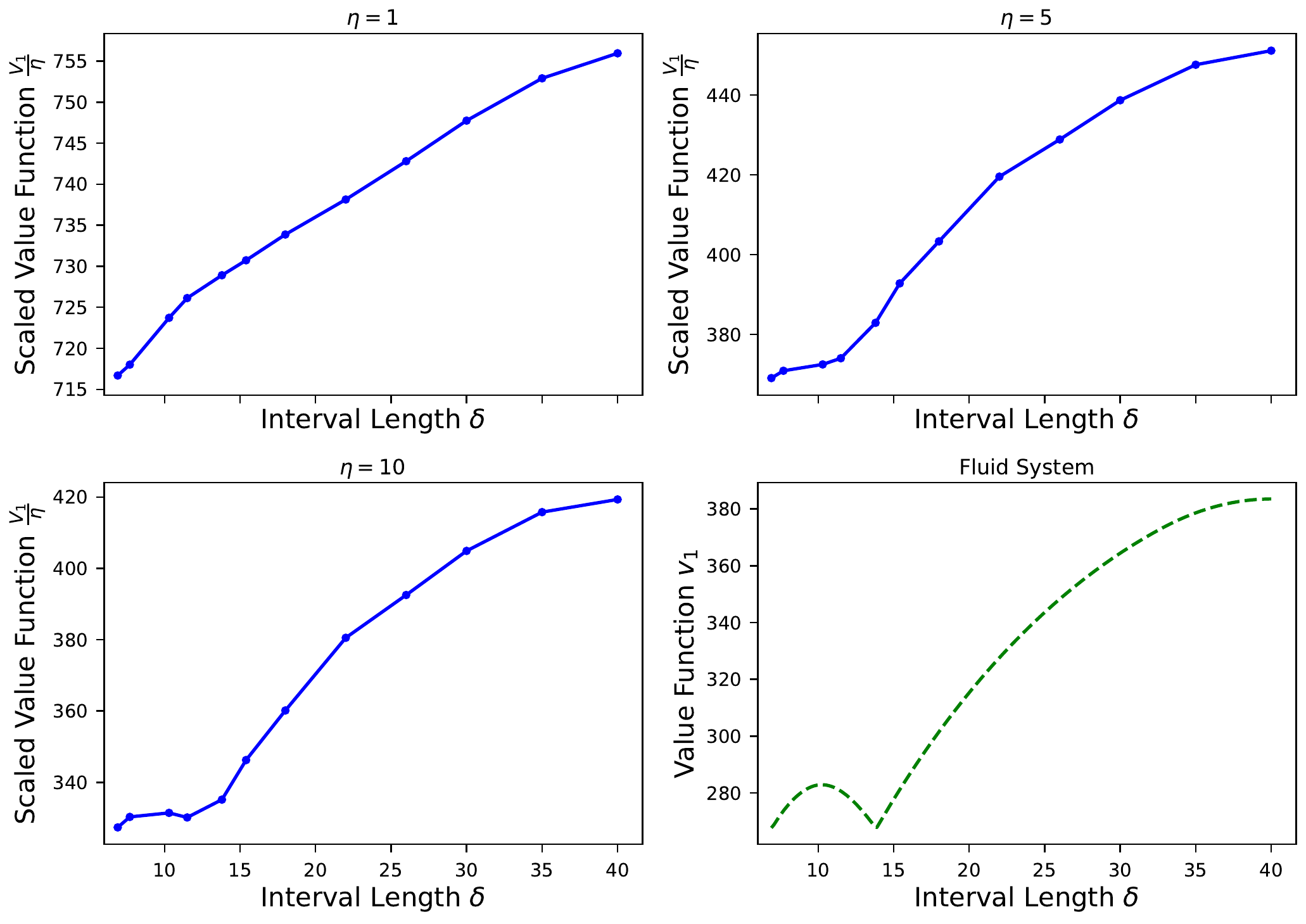}}
    {Scaled value function \(V_1 /\eta\) versus \(\delta\) for a sequence of stochastic systems along with a plot of the fluid value function \(v_1\) with the same parameters as the base stochastic system.\label{fig: stoch sys_1 bigger samples}}
    { The base stochastic system parameters are: $(\lambda^1, \lambda^2) = (0.35, 0.3)$, $(h^1 , h^2) = (3, 1)$, $(X^1, X^2) = (9, 1)$, $T = 40$. The fluid system has the same parameters and initial conditions as the base stochastic system, with $\mu^1 = \mu^2 = 1$. This is the same system as Figure \ref{fig: stoch sys_1}, with the sample numbers of each system doubled. }
\end{figure}

\begin{figure}
    \FIGURE{
    \includegraphics[width=\textwidth]{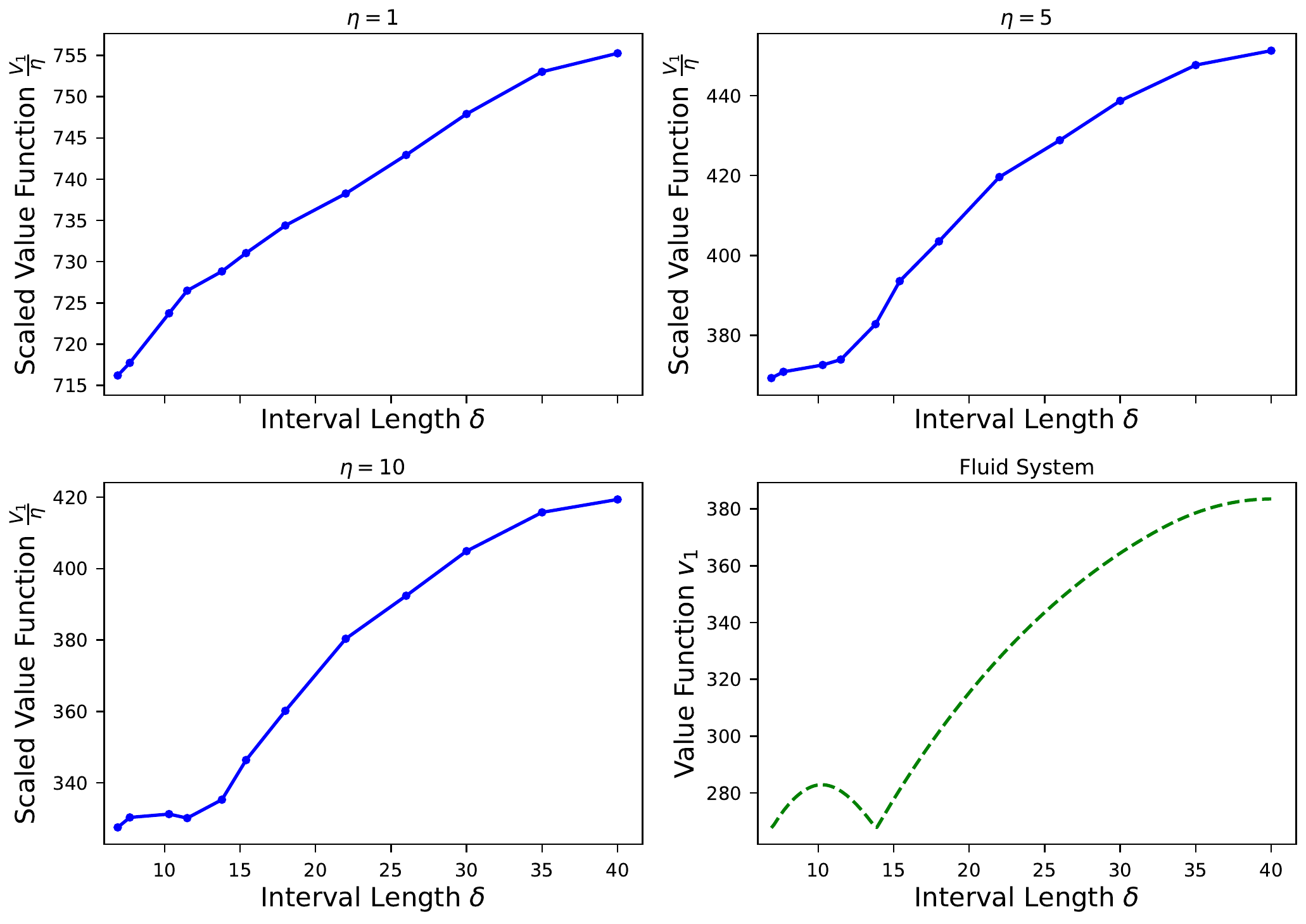}}
    {Scaled value function \(V_1 /\eta\) versus \(\delta\) for a sequence of stochastic systems along with a plot of the fluid value function \(v_1\) with the same parameters as the base stochastic system.\label{fig: stoch sys_1 bigger space}}
    { The base stochastic system parameters are: $(\lambda^1, \lambda^2) = (0.35, 0.3)$, $(h^1 , h^2) = (3, 1)$, $(X^1, X^2) = (9, 1)$, $T = 40$. The fluid system has the same parameters and initial conditions as the base stochastic system, with $\mu^1 = \mu^2 = 1$. This is the same system as Figure \ref{fig: stoch sys_1}, with the maximum number in system $X_{\max}$ increased to $X_{\max} \in \{35, 86, 150\}$, for $\eta \in \{1, 5, 10\}$. }
\end{figure}

\begin{figure}
    \FIGURE{
    \includegraphics[width=\textwidth]{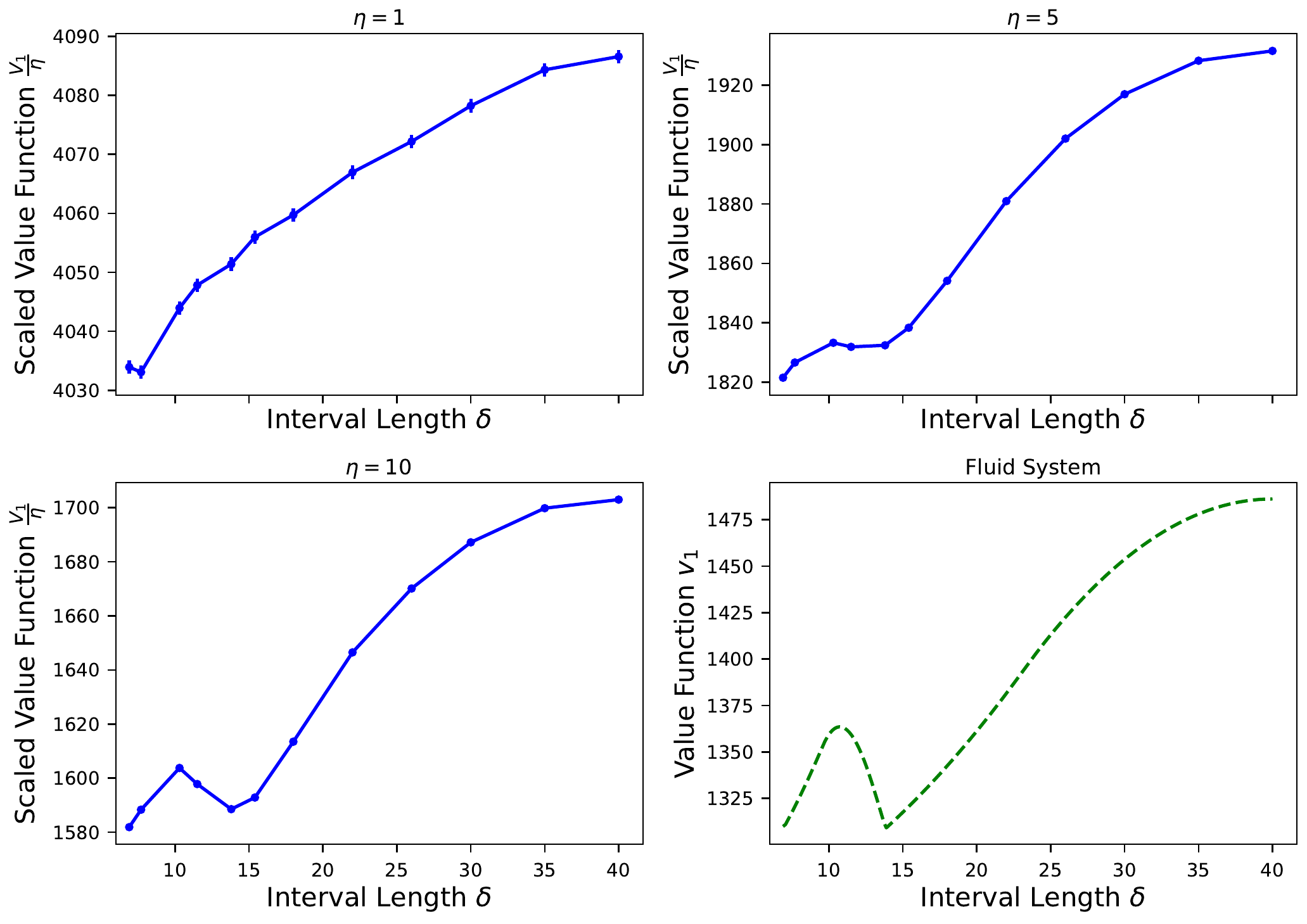}}
    {Scaled value function \(V_1 /\eta\) versus \(\delta\) for a sequence of stochastic systems along with a plot of the fluid value function \(v_1\) with the same parameters as the base stochastic system.\label{fig: stoch sys_2 bigger samples}}
    { The base stochastic system parameters are: $(\lambda^1, \lambda^2) = (0.35, 0.25)$, $(h^1 , h^2) = (20, 1)$, $(X^1, X^2) = (9, 1)$, $T = 40$. The fluid system has the same parameters and initial conditions as the base stochastic system, with $\mu^1 = \mu^2 = 1$. This is the same system as Figure \ref{fig: stoch sys_1},  with the sample numbers of each system doubled. }
\end{figure}

\begin{figure}
    \FIGURE{
    \includegraphics[width=\textwidth]{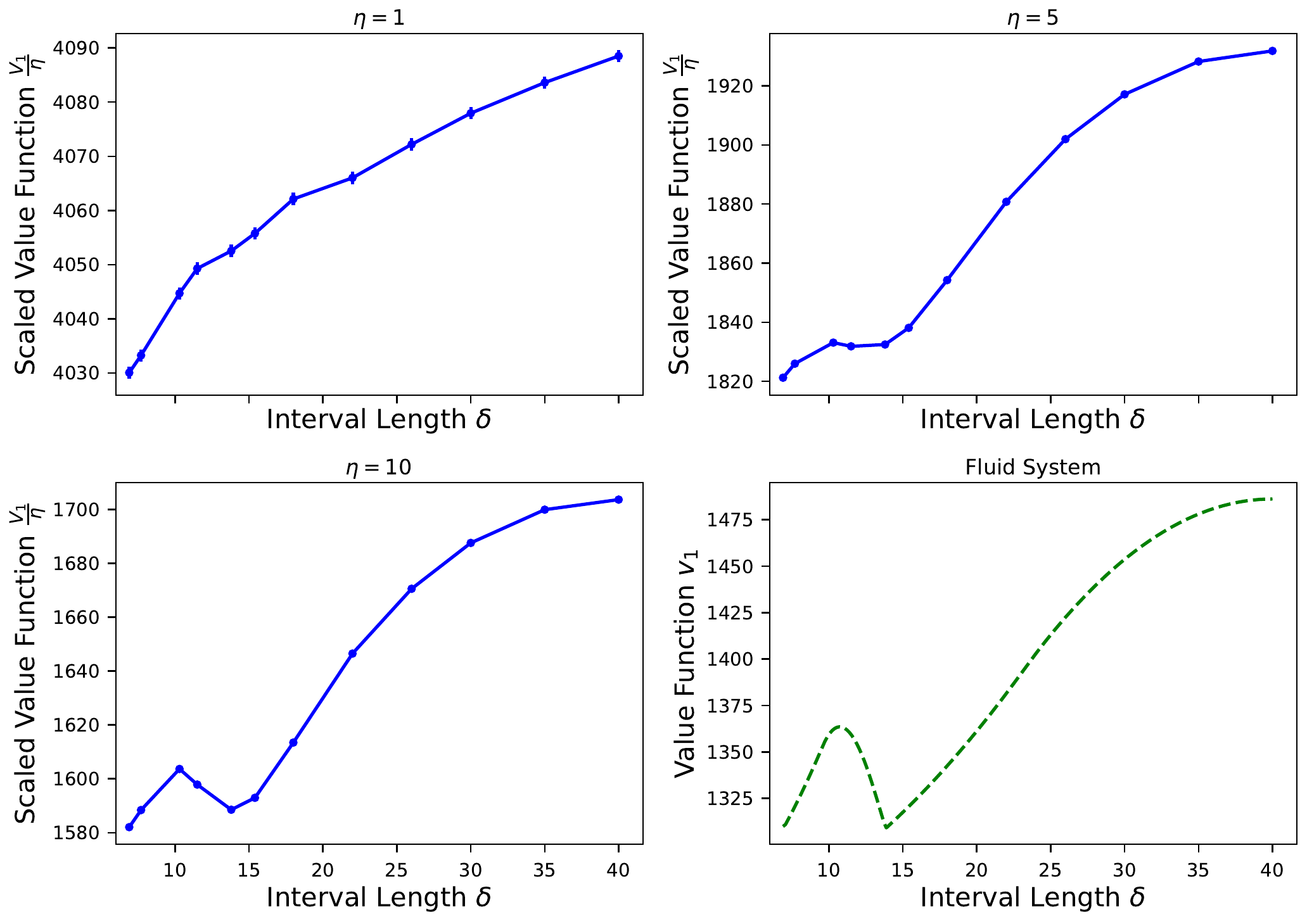}}
    {Scaled value function \(V_1 /\eta\) versus \(\delta\) for a sequence of stochastic systems along with a plot of the fluid value function \(v_1\) with the same parameters as the base stochastic system.\label{fig: stoch sys_2 bigger space}}
    { The base stochastic system parameters are: $(\lambda^1, \lambda^2) = (0.35, 0.25)$, $(h^1 , h^2) = (20, 1)$, $(X^1, X^2) = (9, 1)$, $T = 40$. The fluid system has the same parameters and initial conditions as the base stochastic system, with $\mu^1 = \mu^2 = 1$. This is the same system as Figure \ref{fig: stoch sys_2}, with the maximum number in system $X_{\max}$ increased to $X_{\max} \in \{40, 92, 155\}$, for $\eta \in \{1, 5, 10\}$. }
\end{figure}
\end{document}